\documentclass[11pt]{article}
\usepackage[utf8]{inputenc} 
\usepackage[english]{babel}
\usepackage[T1]{fontenc}
\usepackage[bbgreekl]{mathbbol}
\usepackage[margin=2cm,vmargin=2cm]{geometry}
\usepackage{caption}
\usepackage{cite}
\usepackage{tikz-cd}
\usepackage{amsmath,amsfonts,amssymb,amsthm}
\usepackage{authblk}
\usepackage{csquotes}
\usepackage{url}
\usepackage{fullpage,setspace}
\usepackage{cancel,verbatim,slashed,todonotes}
\usepackage{xcolor,mdframed,graphicx,color}
\usepackage{dsfont}
\usepackage{mathtools}
\usepackage{multirow,stmaryrd}
\usepackage{array}
\usepackage{makecell}
\DeclareSymbolFontAlphabet{\mathbbvar}{bbold}
\DeclareSymbolFontAlphabet{\mathbb}{AMSb}
\newcommand{\M}{\mathcal{M}}
\newcommand{\BU}{\mathbf{B}U(1)_\mathrm{conn}}
\newcommand{\Coo}{\mathcal{C}^{\infty}}
\newcommand{\Aut}{\mathbf{Aut}}
\newcommand{\Diff}{\mathrm{Diff}}

\newcommand{\Lim}[1]{\raisebox{0.5ex}{\scalebox{0.8}{$\displaystyle \lim_{#1}\;$}}}
\theoremstyle{definition}
\newtheorem{theorem}{Lemma}[section]
\theoremstyle{definition}
\newtheorem{definition}[theorem]{Definition}
\theoremstyle{definition}
\newtheorem{post}[theorem]{Postulate}
\theoremstyle{definition}
\newtheorem{remark}[theorem]{Remark}
\theoremstyle{definition}
\newtheorem{example}[theorem]{Example}
\theoremstyle{definition}
\newtheorem{digression}[theorem]{Digression}
\numberwithin{equation}{subsection}
\tikzcdset{%
    triple line/.code={\tikzset{%
        double equal sign distance, 
        double=\pgfkeysvalueof{/tikz/commutative diagrams/background color}}},
    quadruple line/.code={\tikzset{%
        double equal sign distance, 
        double=\pgfkeysvalueof{/tikz/commutative diagrams/background color}}},
    Rrightarrow/.code={\tikzcdset{triple line}\pgfsetarrows{tikzcd implies cap-tikzcd implies}},
    RRightarrow/.code={\tikzcdset{quadruple line}\pgfsetarrows{tikzcd implies cap-tikzcd implies}}
}    
\newcommand*{\tarrow}[2][]{\arrow[Rrightarrow, #1]{#2}\arrow[dash, shorten >= 0.5pt, #1]{#2}}
\usepackage{hyperref}
\title{
\begin{flushright}
\normalsize{QMUL-PH-19-35}
\end{flushright}
\vspace{1cm}
\bigskip
\bf
Global Double Field Theory is\\Higher Kaluza-Klein Theory}
\author{\sc Luigi Alfonsi}
\affil{\em\normalsize Centre for Research in String Theory,\\\em School of Physics and Astronomy,\\\em Queen Mary University of London,\\\em 327 Mile End Road, London E1 4NS, UK\\\vspace{4mm}\tt\href{mailto:l.alfonsi@qmul.ac.uk}{l.alfonsi@qmul.ac.uk}}
\date{}
\begin{document}
\maketitle

\vspace{0.5cm}
\abstract{
\noindent Kaluza-Klein Theory states that a metric on the total space of a principal bundle $P\rightarrow M$, if it is invariant under the principal action of $P$, naturally reduces to a metric together with a gauge field on the base manifold $M$. We propose a generalization of this Kaluza-Klein principle to higher principal bundles and higher gauge fields. For the particular case of the abelian gerbe of Kalb-Ramond field, this Higher Kaluza-Klein geometry provides a natural global formulation for Double Field Theory (DFT). In this framework the doubled space is the total space of a higher principal bundle and the invariance under its higher principal action is exactly a global formulation of the familiar strong constraint. The patching problem of DFT is naturally solved by gluing the doubled space with a higher group of symmetries in a higher category. Locally we recover the familiar picture of an ordinary para-Hermitian manifold equipped with Born geometry. Infinitesimally we recover the familiar picture of a higher Courant algebroid twisted by a gerbe (also known as Extended Riemannian Geometry). As first application we show that on a torus-compactified spacetime the Higher Kaluza-Klein reduction gives automatically rise to abelian T-duality, while on a general principal bundle it gives rise to non-abelian T-duality. As final application we define a natural notion of Higher Kaluza-Klein monopole by directly generalizing the ordinary Gross-Perry one. Then we show that under Higher Kaluza-Klein reduction, this monopole is exactly the NS5-brane on a $10d$ spacetime. If, instead, we smear it along a compactified direction we recover the usual DFT monopole on a $9d$ spacetime.}

\newpage
\tableofcontents
\newpage
\section{Introduction}

\subsection{Double Field Theory and its geometry}
Double Field Theory (DFT), officially born by \cite{HulZwi09}, is a T-duality covariant formulation of the bosonic sector of Type II supergravity. Seminal work in this area predating what is now called DFT was done in \cite{Siegel1} and \cite{Siegel2}. See \cite{BerTho14} for a broad review of the subject. In DFT the T-duality covariant fields live on a double-dimensional coordinate patch which can be seen as the fiber product of a spacetime patch and its T-dual. However they are constrained to depend only on half of the coordinates to avoid unphysical degrees of freedom. This condition is usually known as \textit{strong constraint}. Remarkably, by doubling the dimension of spacetime, DFT is able to geometrize at once two different and apparently unrelated features of String Theory:
\begin{itemize}
    \item differential T-dualities,
    \item local gauge transformations of the Kalb-Ramond field.
\end{itemize}
The fact that introducing these extra coordinates allows to describe both gauge transformations of the Kalb-Ramond field and T-duality is a strong hint of some bigger unification principle underlying. In this sense it has been supposed that a DFT geometry should be thought as a new stringy geometry, which should be for strings what Riemannian geometry is for usual point-particles. For instance see the s by \cite{Park11}, \cite{Hohm13space} and \cite{BerBla14}. However DFT is still plagued by a severe limit: it is only provided with a coordinate patch-wise formulation, with still no agreement on a definitive global formalization. Since T-duality is also a global topological transformation, this poses a serious problem.

\subsection{Papadopoulos' puzzle}
Let us consider the most conservative proposal of patching conditions for DFT which is compatible with strong constraint. For example the finite geometric transformations by \cite[p.11]{HohZwi12} on two-fold overlaps of patches $U_\alpha\cap U_\beta$ are given by
\begin{equation}\label{eq:papa}
    x_\beta = \varphi_{\alpha\beta}(x_\alpha), \quad \tilde{x}_{\beta} = \tilde{x}_{\alpha} + \Lambda_{\alpha\beta}(x_\alpha)
\end{equation}
where $\varphi_{\alpha\beta}$ and $\Lambda_{\alpha\beta}$ are respectively a diffeomorphism and a covector on $U_\alpha\cap U_\beta$. The consistency of \eqref{eq:papa} implies the cocycle condition $\Lambda_{\alpha\beta}+\Lambda_{\beta\gamma}+\Lambda_{\gamma\alpha}=0$ on three-fold overlaps of patches, which implies that $\Lambda_{\alpha\beta}$ is a \v{C}ech coboundary $\Lambda_{\alpha\beta}= \eta_\alpha - \eta_\beta$. Now, since $\Lambda_{\alpha\beta}$ is interpreted as a $1$-form gauge transformation of the Kalb-Ramond field, we have $B_\beta = B_\alpha + \mathrm{d}\Lambda_{\alpha\beta}$. The fact that $\Lambda_{\alpha\beta}$ is a coboundary poses a problem because, if we perform a gauge transformation $B_\alpha':=B_\alpha+\mathrm{d}\eta_\alpha$ on each patch, we get $B_\alpha' = B_\beta'$. This means that the curvature $H=\mathrm{d}B_\alpha$ is globally exact and then the $H$-flux determined by its cohomology class $[H]\in H^3(M,\mathbb{Z})$ is trivial. This problem has been indeed pointed out by \cite{Pap13}.

\vspace{0.25cm}
\noindent One of the most recent and promising approach to give a global formalization to DFT is the para-Hermitian framework.
Firstly the use of para-K\"aler geometry was proposed by \cite{Vais12}, then generalized to para-Hermitian geometry by \cite{Vai13} and further developed by \cite{Svo17}. Finally the full proposal of identifying a DFT with a para-Hermitian manifold equipped with a generalized metric by \cite{Svo18} and \cite{Svo19}. This ideas were further expanded by \cite{MarSza18} and \cite{MarSza19}, in particular by clarifying the para-Hermitian geometry of non-abelian and Lie-Poisson T-duality. However para-Hermitian geometry does not solve Papadopoulos' puzzle. To understand why we need to look at the geometric theory which formalize the Kalb-Ramond field: \textit{higher geometry}.

\vspace{0.25cm}
\begin{table}[!ht]\begin{center}
\begin{tabular}{ c | c c }
 & \textbf{Kaluza-Klein Theory} & \textbf{DFT (para-Hermitian)} \\[0.5ex] \hline
 Space & Circle bundle & Para-Hermitian fibration\\
 & $P=\bigsqcup_\alpha U_\alpha\times S^1 /\sim$ & $N=\bigsqcup_\alpha T^\ast U_\alpha/\sim$ \\ 
  & where & where \\ [1ex]
  & $ (x_\alpha,\theta_\alpha) \sim (x_\beta,\theta_\beta+f_{\alpha\beta})$ & $ (x_\alpha,\tilde{x}_\alpha) \sim (x_\beta,\tilde{x}_\beta+\Lambda_{\alpha\beta})$\\  [5ex] 
 Local symmetry & Atiyah algebroid & Courant algebroid \\
 & $TP=\bigsqcup_\alpha TU_\alpha\oplus\mathbb{R} /\sim$ & $TN\!=\bigsqcup_\alpha TU_\alpha\oplus T^\ast U_\alpha /\sim$ \\
  & with sections & with sections \\ [1ex]
  & $ \begin{pmatrix}X \\ g_\alpha \end{pmatrix} \sim \begin{pmatrix}1 & 0 \\ \mathrm{d}f_{\alpha\beta} & 1 \end{pmatrix}\begin{pmatrix}X \\ g_\beta \end{pmatrix} $ & $ \begin{pmatrix}X \\ \xi_\alpha \end{pmatrix} \sim \begin{pmatrix}1 & 0 \\ \mathrm{d}\Lambda_{\alpha\beta} & 1 \end{pmatrix}\begin{pmatrix}X \\  \xi_\beta \end{pmatrix} $\\  [5ex]
  Field & Electromagnetic field & Kalb-Ramond field \\
 & $A_{\alpha}-A_{\beta}=\mathrm{d}f_{\alpha\beta}$ & $B_{\alpha}-B_{\beta}=\mathrm{d}\Lambda_{\alpha\beta}$ \\
\end{tabular}\caption{A brief comparison of Kaluza-Klein geometry and para-Hermitian geometry for DFT.}\end{center}\end{table}

\vspace{-0.5cm}
\subsection{Higher Geometry in String Theory}
Over the last years a geometrical picture of the Kalb-Ramond field has been clarified, by generalizing the notion of principal bundle. The Kalb-Ramond field has been interpreted as the connection of a \textit{bundle gerbe} (or principal circle $2$-bundle), a concept originally introduced by \cite{Murray}. The proper definition of morphism (i.e. gauge transformation) of bundle gerbes was later introduced by \cite{Murray2}. See \cite{Murray3} for an intruductory review. Bundle gerbe was reformulated by \cite{Hit99} and recently generalized by \cite{Principal1} by the idea of $n$-bundle. A principal $n$-bundle is morally speaking a principal bundle where the ordinary gauge group is replaced with a $n$-group: an object that takes into account not only symmetries, but also symmetries of symmetries. Hence the Kalb-Ramond field is a connection of a circle $2$-bundle. This means that the local $2$-forms $B_{\alpha}\in\Omega^2(U_\alpha)$ are patched by local $1$-form gauge transformations $\Lambda_{\alpha\beta}\in\Omega^1(U_\alpha\cap U_\beta)$ which are themselves patched by scalar gauge transformations $G_{\alpha\beta\gamma}\in\Coo(U_\alpha\cap U_\beta \cap U_\gamma)$ satisfying the cocycle condition on four-fold overlaps of patches. Therefore the patching conditions of the differential local data of the Kalb-Ramond field can be summed up by 
\begin{equation}
    \begin{aligned}
    H \,&=\, \mathrm{d}B_\alpha \\
    B_\beta - B_\alpha \,&=\, \mathrm{d}\Lambda_{\alpha\beta} \\
    \Lambda_{\alpha\beta}+\Lambda_{\beta\gamma}+\Lambda_{\gamma\alpha} \,&=\, \mathrm{d}G_{\alpha\beta\gamma} \\
    G_{\alpha\beta\gamma}-G_{\beta\gamma\delta}+G_{\gamma\delta\alpha}-G_{\delta\alpha\beta} \,&\in\, 2\pi\mathbb{Z}
    \end{aligned}
\end{equation}
We can immediately see why the na\"{i}ve patching conditions \eqref{eq:papa} for DFT do not work: they do not take into account the intrinsic higher nature of the Kalb-Ramond field. In other words we have no way to geometrize gauge transformations $G_{\alpha\beta\gamma}$ on three-fold overlaps of patches.

\vspace{0.25cm}
\noindent This picture of the circle $2$-bundle yields a natural generalization of parallel transport which is along surfaces, instead of curves. This corresponds exactly to the Wess-Zumino term controlling the coupling of string worldsheets with the Kalb-Ramond field. Hence principal $n$-bundles are a natural framework to deal with \textit{higher gauge theories} (see \cite{Baez11}), which are a global formalization of $n$-form fields in physics. For instance higher gauge theory has been used by \cite{Saem17}, \cite{Saem19}, \cite{Saem19x} to formulate a $6d$ superconformal field theory which must be regarded as a very promising step in the direction of developing the desired M5-brane worldvolume theory.

\vspace{0.25cm}
\noindent Moreover higher geometry has been widely used in the direction of generalizing geometric prequantization from ordinary particles to string and branes. This proposal has been named \textit{Higher Prequantization}. This field of research started from the idea of quantization of $n$-plectic manifolds by \cite{Rog11}, \cite{SaSza11}, \cite{Rog13} and of loop spaces by \cite{SaSza11x}. These descriptions are now unified in terms of higher stacks and the theory has been further generalized by the research of \cite{SaSza13}, \cite{FSS13}, \cite{Sch16}, \cite{FSS16}, \cite{BSS16} and \cite{Sza19}.

\vspace{0.25cm}
\noindent Higher geometry is being also successfully applied to the fundamental understanding of M-theory by \cite{FSS12}, \cite{FSS15x}, \cite{FSS19x}, \cite{FSS19xxx}, \cite{BSS19} and \cite{HSS19}. In these references the topological and differential structure of M-theory is investigated, until remarkably a proposal for the generalized cohomology theory that charge-quantizes the supergravity $C$-field is made (\textit{Hypothesis H}) by \cite{FSS19coho}. This idea was further explored by \cite{BSS18}, \cite{SS19} and \cite{FSS19xx}.

\vspace{0.25cm}
\noindent It is very interesting that the Lie $n$-algebroids appearing in higher geometry have been revealed to naturally encompass BV–BRST formalism for quantization of field theories in an infinitesimal fashion. In this regard see \cite{Pau14}, \cite{Saem18bv} and \cite{Saem19bv}. Other properties of field theories related to higher Lie algebras have been explored by \cite{Hohm17}, \cite{Hohm17x}, \cite{Hohm19} and \cite{Hohm19x}.

\vspace{0.25cm}
\noindent Significantly the research in nonassociative physics, which emerges from open String Theory, has been linked not only to non-geometric fluxes (see \cite{SN1}, \cite{SN2}, and \cite{SN4}), but also to higher geometry. In this regard see \cite{SN3}, \cite{SN5}, \cite{SN6}, \cite{SN7}, and \cite{NS8}.

\subsection{Higher geometry in DFT}
History of DFT and higher geometry have been long entwined, even if not always explicitly. Indeed, as we explained, the Kalb-Ramond field is mathematically a higher gauge theory and this means that its geometrization is a circle $2$-bundle. Let us now recap their fundamental points of contacts from the literature.

\vspace{0.25cm}
\noindent Fist of all the Courant algebroid appearing in Generalized Geometry (see \cite{Gua11}) has been understood as an higher analogue of the ordinary Atiyah algebroid for a circle $2$-bundle. In other words the Courant algebroid is the algebroid of the (horizontal) gauge transformations of a circle $2$-bundle. It is possible to find details about this relation in \cite{Col11} and \cite{Rog13}, but the idea can be traced back to the beginning of the theory. In physics Type II supergravity has been revealed to be naturally formulated in terms of Generalized Geometry by \cite{Hull07}, \cite{Wald08E}, \cite{Wald08}, \cite{Wald11} and \cite{Wald12}. Moreover Generalized Geometry is a natural framework to explicitly embody T-duality in the theory, as shown by \cite{CavGua11}.

\vspace{0.25cm}
\noindent Remarkably higher geometry has been identified as the natural framework in which is possible to deal with T-duality. This has been formalized in terms of an isomorphism of gerbes which geometrically encode the one the Kalb-Ramond field and the other its T-dual by \cite{BunNik13}, \cite{FSS16x}, \cite{FSS17x}, \cite{FSS18}, \cite{FSS18x} and \cite{NikWal18}. Assume that we have two $T^n$-bundle spacetimes $M\xrightarrow{\pi}M_0$ and $\widetilde{M}\xrightarrow{\tilde{\pi}}M_0$ over a common base manifold $M_0$. In the references a couple of circle $2$-bundles $P\xrightarrow{\Pi}M$ and $\widetilde{P}\xrightarrow{\tilde{\Pi}}\widetilde{M}$, formalizing two Kalb-Ramond fields respectively on $M$ and $\widetilde{M}$, are geometric T-dual if the following isomorphism exists  
\begin{equation}
    \begin{tikzcd}[row sep={11ex,between origins}, column sep={11ex,between origins}]
    & P\times_{M_0} \widetilde{M}\arrow[rr, "\cong"', "\text{T-duality}"]\arrow[dr, "\Pi"']\arrow[dl, "\tilde{\pi}"] & & M\times_{M_0}\widetilde{P}\arrow[dr, "\pi"']\arrow[dl, "\tilde{\Pi}"] \\
    P\arrow[dr, "\Pi"'] & & M\times_{M_0}\widetilde{M}\arrow[dr, "\pi"']\arrow[dl, "\tilde{\pi}"] & & [-2.5em]\widetilde{P}\arrow[dl, "\tilde{\Pi}"] \\
    & M\arrow[dr, "\pi"'] & & \widetilde{M}\arrow[dl, "\tilde{\pi}"] & \\
    & & M_0 & &
    \end{tikzcd}
\end{equation}
Notice that this picture can be seen as the finite version of the one appearing in \cite{CavGua11} for Courant algebroids. In the references it is shown that this induces an isomorphism between the twisted cohomology theory of D-branes of Type IIA and of Type IIB String Theory.

\vspace{0.25cm}
\noindent It is not so surprising that research in DFT has been affected by these new geometric ideas. It was noticed by \cite{BCM14} that the doubled metric of DFT, to actually geometrize the Kalb-Ramond field, must carry a gerbe structure and have non-trivial local data three-fold overlaps. These arguments lead to the idea that a finite well-defined DFT geometry must be constructed in the context of higher geometry. Indeed a $2$-algebroid formalism was proposed by \cite{DesSae18}, generalized to Heterotic DFT by \cite{DesSae18x} and then applied to the particular case of nilmanifolds by \cite{DesSae19}. This successful idea was also translated to Exceptional Field Theory (ExFT) by \cite{Arv18}. Independently \cite{Hohm19DFT} showed that the gauge structure of the infinitesimal generalized diffeomorphisms of DFT has an higher algebra structure.

\subsection{Wishlist for a DFT geometry}
To formulate a finite geometry of Double Field Theory in the context of higher geometry we will require the following intuitive points:
\begin{itemize}
    \item DFT geometry must consist of a couple $(\mathcal{M},\mathcal{H})$ where
    \begin{itemize}
        \item $\M$ is some geometric object generalizing the notion of smooth manifold,
        \item $\mathcal{H}$ is some field over $\mathcal{M}$ generalizing the notion of Riemannian metric,
    \end{itemize}
    \item the strong constraint must have a global geometric interpretation (similarly to cylindricity condition in Kaluza-Klein theories),
    \item under strong constraint DFT geometry must reduce
    \begin{itemize}
        \item globally to an abelian gerbe structure on a Riemannian manifold,
        \item infinitesimally to Generalized Geometry,
    \end{itemize}
    \item T-duality must have a global geometric interpretation,
    \item DFT geometry should be deducible from a small set of simple assumptions.
\end{itemize}

\section{Background: fundamentals of higher geometry}
This section will be devoted to a basic summary of main ideas in higher geometry. The following introduction is meant to make the next discussion as self-contained as possible for the reader.

\vspace{0.2cm}
\noindent Higher geometry is essentially differential geometry where the notion of equality has been replaced by the weaker one of \textit{equivalence}. This is a natural framework in physics, since equivalences can be interpreted as gauge transformations (or even dualities). Indeed the question of whether two fields configurations are gauge-equivalent is physically more natural than the question whether they are equal.
We will mostly follow the notation presented by \cite{DCCTv2}, but the translation to the one used for instance in \cite{BSS16} is immediate.

\subsection{Higher smooth stacks}

In this subsection we will deal in a quite colloquial non-formal way with the fundamental geometric object in higher geometry: the \textit{higher smooth stack}. Morally speaking an higher smooth stack $\mathbf{S}$ is a sheaf of $n$-groupoids over manifolds. This means that for a manifold $M$ with good cover $\mathcal{U}=\{U_\alpha\}$, the higher groupoid $\mathbf{S}(M)$ can be described in local data by a collection of higher groupoids $\mathbf{S}(U_{\alpha_1}\cap\cdots\cap U_{\alpha_k})$ on any $k$-fold overlap of patches, which are glued by their groupoid morphisms. For a formal exposition see \cite{Principal2}.

\begin{example}[\v{C}ech groupoid]
Given a good open cover $\mathcal{U}:=\{U_\alpha\}$ of a smooth manifold $M$, its \textit{\v{C}ech groupoid} is defined by the following $1$-groupoid
\begin{equation}\label{eq:chechgroupoidex}
    \check{C}(\mathcal{U}) := \bigg(
    \begin{tikzcd}[row sep=scriptsize, column sep=3ex] \underset{\alpha\beta}{\bigsqcup}\; U_{\alpha}\cap U_\beta  \arrow[r, yshift=0.7ex, "s"] \arrow[r, yshift=-0.7ex, "t"'] & \; \underset{\alpha}{\bigsqcup}\; U_{\alpha}
    \end{tikzcd}
    \bigg)
\end{equation}
where source and target are the natural embeddings $s:U_{\alpha}\cap U_\beta\hookrightarrow U_\alpha$ and $t:U_{\alpha}\cap U_\beta\hookrightarrow U_\beta$. 
\end{example}

\noindent In the \v{C}ech groupoid the gluing conditions of the manifold $M$ between its patches $\{U_\alpha\}$ have been promoted to morphisms.

\begin{remark}[Higher smooth stacks]
The notion of $n$-groupoid can be modelled by a simplicial set which is in particular a \textit{Kan complex}. Morally speaking this means that a higher Lie groupoid is a \textit{functor from the category of simplices to the category of manifolds}, which maps $k$-simplices to the space of $k$-morphisms of the groupoid. They also must satisfy some subtler conditions which are explained in \cite{Principal2}. Analogously \textit{higher smooth stacks} can be modelled by simplicial sheaves over manifolds.
\end{remark}

\noindent Let us give a concrete simple example of a groupoid seen as a simplicial set.

\begin{example}[\v{C}ech groupoid as simplicial object]
In simplicial terms the \eqref{eq:chechgroupoidex} is given by
\begin{equation}
    \begin{tikzcd}[row sep=scriptsize, column sep=3ex]\; \cdots\; \arrow[r, yshift=1.8ex]\arrow[r, yshift=0.6ex]\arrow[r, yshift=-1.8ex]\arrow[r, yshift=-0.6ex]& \underset{\alpha\beta\gamma}{\bigsqcup} U_{\alpha}\cap U_\beta\cap U_\gamma \;
    \arrow[r, yshift=1.4ex] \arrow[r] \arrow[r, yshift=-1.4ex] & \; \underset{\alpha\beta}{\bigsqcup}\; U_{\alpha}\cap U_\beta  \arrow[r, yshift=0.7ex] \arrow[r, yshift=-0.7ex] & \; \underset{\alpha}{\bigsqcup}\; U_{\alpha}\; \arrow[r] &\; \check{C}(\mathcal{U})
    \end{tikzcd}
\end{equation}
\end{example}

\noindent Let us now give simple examples of $0$-stacks, which are just ordinary sheaves.

\begin{example}[Some useful $0$-stacks]
Interestingly $\Diff(M)$ can be thought as a $0$-stack sending a manifold $M$ to its group of diffeomorphisms, while $\Omega^{n}(M)$ is a $0$-stack sending $M$ to the vector space of its $n$-forms. Analogously $\Omega^{n}_{\mathrm{cl}}(M)$ is the $0$-stack of closed $n$-forms. However we remark that a $0$-stack of exact forms $\Omega^{n}_{\mathrm{ex}}(M)$ \textit{does not} exist, because it would not satisfy the gluing conditions on overlaps of patches $\Omega^n_{\mathrm{ex}}(U_\alpha\cap U_\beta)$. Given any smooth manifold $N$ there is also a natural $0$-stack $\Coo(-,N)$ which sends manifolds $M$ to the space of smooth maps $\Coo(M,N)$.
\end{example}

\noindent Let us now give an example of how the stack formalism can be powerful in generalizing ordinary moduli spaces of geometric structures to \textit{moduli stacks}. 

\begin{example}[Orthogonal structure moduli stack]\label{ex:orth}
An \textit{orthogonal structure moduli stack} is the stack which encodes a Riemannian metric structure on a manifold $M$. Let the transition functions of the tangent bundle $TM$ be functions $N_{\alpha\beta}\in\Coo\big(U_\alpha\cap U_\beta,\, GL(d)\big)$. Now a map
\begin{equation}
    M\simeq \check{C}(\mathcal{U}) \xrightarrow{\;\;(e,h)\;\;} \mathbf{Orth}(TM)
\end{equation}
is a collection $(e,h)$ of local $GL(d)$-functions $e_\alpha\in\Coo\big(U_\alpha,\,GL(d)\big)$ on patches and of local $O(d)$-functions $h_{\alpha\beta}\in\Coo\big(U_\alpha\cap U_\beta,\,O(d)\big)$ on overlaps of patches, such that they are patched by
\begin{equation}
    \begin{aligned}
        e_\alpha &= h_{\alpha\beta}\cdot e_\beta\cdot N_{\alpha\beta} \\
        h_{\alpha\gamma} &= h_{\alpha\beta}\cdot h_{\beta\gamma}
    \end{aligned}
\end{equation}
on two-fold and on three-fold overlaps. The morphisms $\eta:(e,h)\Mapsto(e',h')$ between these maps 
\begin{equation}
    \begin{tikzcd}[row sep=3ex, column sep=8ex]
         M\simeq \check{C}(\mathcal{U}) \arrow[r, bend left=50, ""{name=U, below}, "(e\text{,}h)"]
        \arrow[r, bend right=50, "(e'\text{,}h')"', ""{name=D}]
        & \mathbf{Orth}(TM)
        \arrow[Rightarrow, from=U, to=D, "\eta"]
    \end{tikzcd}
\end{equation}
are collections of local $O(d)$-functions $\eta_\alpha\in\Coo\big(U_\alpha,\,O(d)\big)$ on each patch, such that they give
\begin{equation}
    \begin{aligned}
        e'_\alpha &= \eta_\alpha \cdot e_\alpha \\
        h'_{\alpha\beta} &= \eta_\alpha\cdot h_{\alpha\beta}\cdot\eta_\beta^{-1}.
    \end{aligned}
\end{equation}
Notice the moduli space of an orthogonal structure is locally given by $\Coo\big(U_\alpha,\,GL(d)/O(d)\big)$ and globally by non-trivially gluing these spaces by using the transition functions of $TM$. Notice also that the $e_\alpha\in\Coo\big(U_\alpha,\,GL(d)\big)$ are the vielbein matrices of the Riemannian structure.
\end{example}

\noindent Now we will briefly present a correspondence which allows us to write abelian stacks in a very simple and immediate fashion. See \cite{DCCTv2} for a detailed discussion about it.

\begin{remark}[Dold-Kan correspondence]\label{rem:dkc}
\textit{Dold-Kan correspondence} exhibits an equivalence between abelian smooth higher stacks and cochain complexes of abelian sheaves over manifolds. In our notation to such a stack $\mathbf{A}$ will correspond a cochain complex $(\mathcal{A}_\bullet,\mathrm{d}_\bullet)$, i.e. explicitly
\begin{equation}
    \mathbf{A} \cong  \left(
    \begin{tikzcd}[row sep=scriptsize, column sep=3ex]\cdots\arrow[r, "\mathrm{d}_{3}"]& \mathcal{A}_2 \arrow[r, "\mathrm{d}_{2}"]& \mathcal{A}_1\arrow[r, "\mathrm{d}_{1}"]& \mathcal{A}_0
    \end{tikzcd}
    \right),
\end{equation}
where the $\mathcal{A}_i$ are all abelian sheaves.
The stack $\mathbf{BA}$, which is called \textit{delooping} of $\mathbf{A}$, is exactly the stack corresponding to the shifted cochain $\mathcal{A}_\bullet[1]$ of smooth sheaves, i.e. explicitly
\begin{equation}
    \mathbf{BA} \cong  \left(
    \begin{tikzcd}[row sep=scriptsize, column sep=3ex]\cdots \arrow[r, "\mathrm{d}_{2}"]& \mathcal{A}_1\arrow[r, "\mathrm{d}_{1}"]& \mathcal{A}_0 \arrow[r, "0"]& \,0\,
    \end{tikzcd}
    \right)
\end{equation}
Given a smooth manifold $M$, the morphisms of the groupoid $\mathbf{H}(M,\mathbf{A})$ correspond to maps of cochain complexes $\check{C}(\mathcal{U})_{\bullet}\rightarrow(\mathcal{A}_\bullet,\mathrm{d}_\bullet)$, where $\check{C}(\mathcal{U})$ is the \v{C}ech groupoid of the manifold $M$.
\end{remark}

\begin{example}[Abelian $1$-stacks and $2$-stacks]\label{ex:stacks}
The following are the relevant examples of abelian $1$-stacks and $2$-stacks we are going to use in the next discussion. They are presented through Dold-Kan correspondence (remark \ref{rem:dkc}) as cochain complexes of abelian sheaves. Notice that in this form they are Deligne complexes:
\begin{equation}
    \begin{aligned}
    \mathbf{B}U(1) \;&\cong\;  \left(
    \begin{tikzcd}[row sep=scriptsize, column sep=8ex]\Coo\left(-,U(1)\right)\;\arrow[r,"0"]&\; 0
    \end{tikzcd}
    \right) \\
    \mathbf{B}U(1)_{\mathrm{conn}} \;&\cong\;   \left(
    \begin{tikzcd}[row sep=scriptsize, column sep=8ex]\Coo\left(-,U(1)\right)\;\arrow[r,"\frac{1}{2\pi i}\mathrm{d}\cdot\log"]&\;\Omega^1(-)
    \end{tikzcd}
    \right) \\
    \mathbf{B}^2U(1)  \;&\cong\; \left(
    \begin{tikzcd}[row sep=scriptsize, column sep=8ex]\Coo\left(-,U(1)\right)\;\arrow[r,"0"]&\; 0\arrow[r,"0"]&\; 0
    \end{tikzcd}
    \right) \\
    \mathbf{B}(\mathbf{B}U(1)_{\mathrm{conn}})\;&\cong\;   \left(
    \begin{tikzcd}[row sep=scriptsize, column sep=8ex]\Coo\left(-,U(1)\right)\;\arrow[r,"\frac{1}{2\pi i}\mathrm{d}\cdot\log"]&\;\Omega^1(-)\arrow[r,"0"]&\; 0
    \end{tikzcd}
    \right) \\
    \mathbf{B}^2U(1)_{\mathrm{conn}}  \;&\cong\;  \left(
    \begin{tikzcd}[row sep=scriptsize, column sep=8ex]\Coo\left(-,U(1)\right)\;\arrow[r,"\frac{1}{2\pi i}\mathrm{d}\cdot\log"]&\;\Omega^1(-)\arrow[r,"\mathrm{d}"]&\; \Omega^2(-)
    \end{tikzcd}
    \right)
        \end{aligned}
\end{equation}
More generally we can write the following abelian $k$-stack for any $k\in\mathbb{N}$ by using Dold-Kan:
\begin{equation}
    \mathbf{B}^kU(1)_{\mathrm{conn}}  \;\cong\;  \left(
    \begin{tikzcd}[row sep=scriptsize, column sep=8ex]\Coo\left(-,U(1)\right)\arrow[r,"\frac{1}{2\pi i}\mathrm{d}\cdot\log"]&\Omega^1(-)\arrow[r,"\mathrm{d}"]& \cdots\arrow[r,"\mathrm{d}"]& \Omega^k(-)
    \end{tikzcd}\right)
\end{equation}
\end{example}

\begin{remark}[Forgetful functor]
Notice we can naturally introduce a \textit{forgetful functor} which forgets the $1$-degree $1$-form part of the cochain complex and retains only the $0$-degree sheaf for
\begin{equation}
    \begin{tikzcd}[row sep=scriptsize, column sep=7ex]
    \mathbf{B}U(1)_{\mathrm{conn}} \arrow[r, "\mathrm{frgt}"] & \mathbf{B}U(1)
\end{tikzcd}
\end{equation}
Analogously we can define natural forgetful functors for the $2$,$1$-degree sheaves of the cochains
\begin{equation}
    \begin{tikzcd}[row sep=scriptsize, column sep=7ex]
    \mathbf{B}^2U(1)_{\mathrm{conn}} \arrow[r, "\mathrm{frgt}"] & \mathbf{B}(\BU) \arrow[r, "\mathrm{frgt}"] &  \mathbf{B}^2U(1)
\end{tikzcd}
\end{equation}
\end{remark}

\noindent Let us now define some useful categories.

\begin{definition}[$(n,r)$-category]
An $(n,r)$\textit{-category} is an $n$-category such that all $k$-morphisms with $k>r$ are equivalences. For example an $n$-groupoid is nothing but an $(n,0)$-category.
\end{definition}

\begin{definition}[$(\infty,1)$-category of stacks]
We call $\mathbf{H}$ the $(\infty,1)$\textit{-category of stack}, such that
\begin{itemize}
    \item objects are higher smooth stacks,
    \item $k$-morphisms for any $k\in\mathbb{N}^+$ are $k$-morphisms of higher smooth stacks.
\end{itemize} 
\end{definition}

\noindent Notice that the category of smooth manifolds is naturally embedded into the $(\infty,1)$-category of smooth stacks: indeed any smooth manifold $M$ can be regarded as a $0$-stack $\,\Coo(-,M)\in\mathbf{H}$.

\begin{definition}[Hom $\infty$-groupoid]
Given any couple of smooth stacks $\mathbf{S}_1,\mathbf{S}_2\in\mathbf{H}$, according to \cite{DCCTv2} we can define the \textit{hom }$\infty$\textit{-groupoid} $\mathbf{H}(\mathbf{S}_1,\mathbf{S}_2)$ as an higher groupoid such that
\begin{itemize}
    \item objects are $1$-morphisms $f:\mathbf{S}_1\rightarrow\mathbf{S}_2$ in $\mathbf{H}$,
    \item $k$-morphisms are $(k+1)$-morphisms of stacks $f\Mapsto f'$ in $\mathbf{H}$.
\end{itemize}
\end{definition}

\noindent Notice that, given an higher smooth stack $\mathbf{S}$ over smooth manifolds, we have the natural equivalence $\mathbf{H}(M,\,\mathbf{S})\simeq\mathbf{S}(M)$ for any smooth manifold $M$ (regarded here as a smooth $0$-stack).

\begin{definition}[Internal hom $\infty$-stack]\label{def:inthom}
Given any couple of smooth stacks $\mathbf{S}_1,\mathbf{S}_2\in\mathbf{H}$, according to \cite{DCCTv2} we can define the \textit{internal hom} $\infty$\textit{-stack} $[\mathbf{S}_1,\mathbf{S}_2]$ by the equivalence
\begin{equation}
    \mathbf{H}\big(M,\,[\mathbf{S}_1,\mathbf{S}_2]\big) \,\simeq\, \mathbf{H}(M\times\mathbf{S}_1,\, \mathbf{S}_2).
\end{equation}
Notice that, if $\,\ast\,$ is a point, we have the natural equivalence $[\,\ast\,,\,\mathbf{S}\,]\,\simeq\,\mathbf{S}\,$ for any stack $\,\mathbf{S}\in\mathbf{H}$.
\end{definition}

\begin{example}[Loop space of a manifold]\label{ex:loopspace}
The \textit{loop space} of a manifold $M$ is $LM:=[S^1,M]$.
\end{example}

\begin{definition}[Slice $\infty$-category]\label{def:slice}
For any given object $\mathbf{S}\in\mathbf{H}$, according to \cite{DCCTv2} we can define the \textit{slice }$\infty$\textit{-category} $\mathbf{H}_{/\mathbf{S}}$ as the $\infty$-category such that
\begin{itemize}
    \item objects are $1$-morphisms $f:\mathbf{X}\rightarrow\mathbf{S}$ in $\mathbf{H}$, 
    \item $1$-morphisms $F:f_1\mapsto f_2$ are homotopy commutative diagrams of the following form
    \begin{equation}
        \begin{tikzcd}[row sep=scriptsize, column sep=8ex]
        \mathbf{X}_1 \arrow[rr, "F"{name=D}]\arrow[rd, "f_1"'] & & \mathbf{X}_2\arrow[ld, "f_2"] \\
        & \mathbf{S}\arrow[Rightarrow, to=D] &
    \end{tikzcd}
    \end{equation}
    \item and so on for $k$-morphisms with $k>1$.
\end{itemize}
\end{definition}

\begin{definition}[Loop space object of an $\infty$-category]
For any given object $X\in\mathbf{C}$ in an $\infty$-category $\mathbf{C}$, we can define the \textit{loop space object} $\Omega_{X}\mathbf{C}$ as the $\infty$-category such that
\begin{itemize}
    \item objects are $1$-morphisms $f:X\rightarrow X$ in $\mathbf{C}$, 
    \item $k$-morphisms are $(k+1)$-morphisms $f_1\Mapsto f_2$ in $\mathbf{C}$.
\end{itemize}
Notice this category must not be confused with the loop space of a manifold from example \ref{ex:loopspace}.
\end{definition}

\subsection{Principal n-bundles and gerbes}

In this subsection we will give a simple introduction to the theory of principal $n$-bundles developed by \cite{Principal1} and \cite{Principal2}. Moreover, from the general theory, we will recover the local differential data of abelian gerbes as presented by \cite{GERBE03}.

\vspace{0.2cm}
\noindent In ordinary differential geometry a principal $G$-bundle on a manifold $M$ is defined by an element of the first non-abelian $G$-cohomology group $H^1(M,G) \,\simeq\, G\mathrm{Bund}(M)_{/\cong}$. These are equivalence classes $\left[f_{\alpha\beta}\right]$ where the representatives are given by \v{C}ech $G$-cocycles $f_{\alpha\beta}\in\Coo(U_\alpha\cap U_\beta,G)$ on $M$ and the equivalence relation is given by \v{C}ech coboundaries $\eta_{\alpha}\in\Coo(U_\alpha,G)$ by $f_{\alpha\beta}\,\cong\,\eta_\alpha f_{\alpha\beta}\eta_\beta^{-1}$. We would like to refine this formalism to a stack description, where we consider $G$-bundles without slashing out gauge transformations.

\begin{remark}[Principal $1$-bundle]\label{modulirem}
The groupoid $\mathbf{H}(M,\mathbf{B}G)$ for a given manifold $M$ and Lie $1$-group $G$ has for objects all the nonabelian \v{C}ech $G$-cocycles $f_{\alpha\beta}$ on $M$ and for morphisms all the couboundaries $f_{\alpha\beta}\mapsto\eta_\alpha f_{\alpha\beta}\eta_\beta^{-1}$ between them. Schematically we have:
\begin{equation}
    \mathbf{H}(M,\mathbf{B}G)\,\simeq\,\left\{\begin{tikzcd}[row sep=scriptsize, column sep=12ex]
    M \arrow[r, bend left=50, ""{name=U, below}, "f_{\alpha\beta}"]
    \arrow[r, bend right=50, "\eta_\alpha f_{\alpha\beta}\eta_\beta^{-1}"', ""{name=D}]
    & \mathbf{B}G
    \arrow[Rightarrow, from=U, to=D, "\eta_\alpha"]
\end{tikzcd} \right\}
\end{equation}
In geometric terms the objects are all the principal $G$-bundles over $M$ and the morphisms are all the isomorphisms (i.e. gauge transformations) between them. Thus we will operatively define a principal $G$-bundle as just an object of groupoid $\mathbf{H}(M,\mathbf{B}G)$.
\end{remark}

\noindent To recover the previous ordinary picture we only need to take the set of path-connected components of the groupoid of principal $1$-bundles:
\begin{equation}\label{eq:pathconnected}
    H^1(M,G) \,=\, \pi_0\mathbf{H}(M,\mathbf{B}G).
\end{equation}

\begin{remark}[Principal $n$-bundles]The fundamental idea for defining principal $n$-bundles is letting the formalism \eqref{modulirem} work for $n$-groups with $n\neq 1$ too. Hence we define a principal $n$-bundle as an object of $\mathbf{H}(M,\mathbf{B}G)$ where $G$ is any $n$-group object in $\mathbf{H}$.
\end{remark}

\begin{example}[Gerbe]
An abelian gerbe is a principal $\mathbf{B}U(1)$-bundle (i.e. a circle $2$-bundle).
\end{example}

\begin{remark}[Dixmier-Douady class]
By taking the group of path-connected components of the groupoid $\mathbf{H}(M,\mathbf{B}^2U(1))$ of the abelian gerbes we obtain the $3$rd cohomology group
\begin{equation}
    \pi_0\mathbf{H}(M,\mathbf{B}^2U(1)) = H^2(M,U(1)) \cong H^3(M,\mathbb{Z}).
\end{equation}
Hence circle $2$-bundles $P\rightarrow M$ over a base manifold $M$ are topologically classified by their Dixmier-Douady class, i.e. by an element $\mathrm{dd}(P)\in H^3(M,\mathbb{Z})$ of the third integer cohomology group of the base manifold. This is totally analogous to how first Chern class $\mathrm{c}_1(P)\in H^2(M,\mathbb{Z})$ classifies ordinary circle bundles $P\rightarrow M$. In general we have a sequence of circle $n$-bundles:
\begin{equation}
    H^1(M,\mathbb{Z})\simeq \Coo(M,S^1), \;\; H^2(M,\mathbb{Z})\simeq S^1\mathrm{Bund}(M)_{/\cong}, \;\; H^3(M,\mathbb{Z})\simeq \mathrm{Gerb}(M)_{/\cong}, \;\;\dots
\end{equation}
where $S^1\mathrm{Bund}(M)_{/\cong}$ and $\mathrm{Gerb}(M)_{/\cong}$ are respectively the group of isomorphism classes of circle bundles and abelian gerbes over the base manifold $M$. Therefore, in this context, a global map in $\Coo(M,S^1)$ can be seen as a circle $0$-bundle.
\end{remark}

\begin{remark}[Gerbe in \v{C}ech picture]
An object of $\mathbf{H}(M,\mathbf{B}^2U(1))$ is given in \v{C}ech data for a good cover $\mathcal{U}=\{U_\alpha\}$ of $M$ by a collection $(G_{\alpha\beta\gamma})$ of local scalars on $U_\alpha\cap U_\beta\cap U_\gamma$ satisfying
\begin{equation}
    \begin{aligned}
        G_{\alpha\beta\gamma}-G_{\beta\gamma\delta}+G_{\gamma\delta\alpha}-G_{\delta\alpha\beta}\in 2\pi\mathbb{Z},
    \end{aligned}
\end{equation}
i.e. an abelian gerbe in \v{C}ech data. The $1$-morphisms between these objects are \v{C}ech coboundaries (in physical words the gauge transformations of the gerbe) given by collections $(\eta_{\alpha\beta})$ of local scalars on overlaps $U_\alpha\cap U_\beta$ so that
\begin{equation}
    \begin{aligned}
        G_{\alpha\beta\gamma} \mapsto G_{\alpha\beta\gamma} + \eta_{\alpha\beta}+\eta_{\beta\gamma}+\eta_{\gamma\alpha}
    \end{aligned}
\end{equation}
The $2$-morphisms between $1$-morphisms (in physical words the gauge-of-gauge transformations of the gerbe) are given by collections $(\epsilon_\alpha)$ of local scalars on each $U_\alpha$ so that
\begin{equation}
    \begin{aligned}
        \eta_{\alpha\beta} \Mapsto \eta_{\alpha\beta}+\epsilon_{\alpha}-\epsilon_{\beta} .
    \end{aligned}
\end{equation}
In terms of diagrams we can write this $2$-groupoid of abelian gerbes as it follows:
\begin{equation}
    \mathbf{H}\big(M,\mathbf{B}^2U(1)\big) \,\simeq\, \left\{\; \begin{tikzcd}[row sep=scriptsize, column sep=26ex]
    \;\;\; M \arrow[r, bend left=60, ""{name=U, below}, "(G_{\alpha\beta\gamma})"]
    \arrow[r, bend right=60, "(G'_{\alpha\beta\gamma})"', ""{name=D}]
    & \qquad\,\mathbf{B}^2U(1)
    \arrow[Rightarrow, from=U, to=D, bend left=55, "(\eta'_{\alpha\beta})", ""{name=R, below}] \arrow[Rightarrow, from=U, to=D, bend right=55, "(\eta_{\alpha\beta})"', ""{name=L}] \tarrow[from=L, to=R, end anchor={[yshift=0.6ex]}]{r} \arrow["(\epsilon_\alpha)", phantom, from=L, to=R, end anchor={[yshift=0.6ex]}, bend left=34]
\end{tikzcd}\right\}
\end{equation}
\end{remark}

\begin{remark}[Gerbe in Chatterjee-Hitchin picture]
There is an alternative but equivalent way to geometrically describe a gerbe: the \textit{Chatterjee-Hitchin} description by \cite{Hit99}.
A gerbe will be given by a circle bundle $P_{\alpha\beta}\in\mathbf{H}(U_\alpha\cap U_\beta,\,\mathbf{B}U(1))$ on each overlap of patches and an isomorphism between each tensor product $P_{\alpha\beta}\otimes P_{\beta\gamma}$ and $P_{\alpha\gamma}$ on every three-fold overlap of patches. The latter is a gauge transformation $G_{\alpha\beta\gamma}\in\Coo(U_\alpha\cap U_\beta\cap U_\gamma)$, so that
\begin{equation}
    P_{\alpha\beta}\otimes P_{\beta\gamma} \xrightarrow[\;G_{\alpha\beta\gamma}\;]{\cong} P_{\alpha\gamma}
\end{equation}
and which satisfies the cocycle condition on four-fold overlaps of patches. This notation is reminiscent of the transition functions $(G_{\alpha\beta})$ of an ordinary circle bundle $P\rightarrow M$, which are indeed $0$-gerbes $G_{\alpha\beta}\in\Coo(U_\alpha\cap U_\beta,\,U(1))$ and which satisfy exactly $G_{\alpha\beta}\cdot G_{\beta\gamma} = G_{\alpha\gamma}$.
\end{remark}

\begin{remark}[Gerbe with connection]
An abelian \textit{gerbe with connection} is given by a cocycle $M\rightarrow\mathbf{B}^2U(1)_{\mathrm{conn}}$ where we defined the stack $\mathbf{B}^2U(1)_{\mathrm{conn}}$ in example \ref{ex:stacks}.
\end{remark}

\begin{remark}[Gerbe with connection in \v{C}ech picture]
An object of $\mathbf{H}(M,\mathbf{B}^2U(1)_{\mathrm{conn}})$ is given in \v{C}ech data for a good cover $\mathcal{U}=\{U_\alpha\}$ of $M$ by a collection $(B_\alpha,\Lambda_{\alpha\beta},G_{\alpha\beta\gamma})$ of $2$-forms $B_\alpha\in\Omega^2(U_\alpha)$, $1$-forms $\Lambda_{\alpha\beta}\in\Omega^2(U_\alpha\cap U_\beta)$ and scalars $G_{\alpha\beta\gamma}\in\Coo(U_\alpha\cap U_\beta\cap U_\gamma)$, patched by
\begin{equation}
    \begin{aligned}
        B_\beta-B_\alpha &= \mathrm{d}\Lambda_{\alpha\beta}, \\
        \Lambda_{\alpha\beta}+\Lambda_{\beta\gamma}+\Lambda_{\gamma\alpha} &= \mathrm{d}G_{\alpha\beta\gamma} \\
        G_{\alpha\beta\gamma}-G_{\beta\gamma\delta}+G_{\gamma\delta\alpha}-G_{\delta\alpha\beta}&\in 2\pi\mathbb{Z}
    \end{aligned}
\end{equation}
i.e. an abelian gerbe with connection in \v{C}ech data. The $1$-morphisms between these objects are \v{C}ech coboundaries (in physical words the gauge transformations of the gerbe), given by collections $(\eta_{\alpha},\eta_{\alpha\beta})$ of local $1$-forms $\eta_\alpha\in\Omega^1(U_\alpha)$ and local scalars $\eta_{\alpha\beta}\in\Coo(U_\alpha\cap U_\beta)$, so that
\begin{equation}\label{eq:coboundaries}
    \begin{aligned}
        B_\alpha &\mapsto B_\alpha + \mathrm{d}\eta_{\alpha}, \\
        \Lambda_{\alpha\beta} &\mapsto \Lambda_{\alpha\beta}+\eta_{\alpha}-\eta_{\beta}+\mathrm{d}\eta_{\alpha\beta} \\
        G_{\alpha\beta\gamma} &\mapsto G_{\alpha\beta\gamma} + \eta_{\alpha\beta}+\eta_{\beta\gamma}+\eta_{\gamma\alpha}
    \end{aligned}
\end{equation}
The $2$-morphisms between $1$-morphisms (in physical words the gauge-of-gauge transformations of the gerbe) are given by collections $(\epsilon_\alpha)$ of local scalars on each $U_\alpha$ so that
\begin{equation}
    \begin{aligned}
        \eta_\alpha &\Mapsto \eta_\alpha + \mathrm{d}\epsilon_{\alpha}, \\
        \eta_{\alpha\beta} &\Mapsto \eta_{\alpha\beta}+\epsilon_{\alpha}-\epsilon_{\beta} .
    \end{aligned}
\end{equation}
In terms of diagrams we can write this $2$-groupoid of abelian gerbes with connection as it follows:
\begin{equation}
    \mathbf{H}\big(M,\mathbf{B}^2U(1)_{\mathrm{conn}}\big) \,\simeq\, \left\{\; \begin{tikzcd}[row sep=scriptsize, column sep=26ex]
    \;\;\; M \arrow[r, bend left=60, ""{name=U, below}, "(B_\alpha\text{,}\,\Lambda_{\alpha\beta}\text{,}\,G_{\alpha\beta\gamma})"]
    \arrow[r, bend right=60, "(B'_\alpha\text{,}\,\Lambda'_{\alpha\beta}\text{,}\,G'_{\alpha\beta\gamma})"', ""{name=D}]
    & \qquad\,\mathbf{B}^2U(1)_{\mathrm{conn}}
    \arrow[Rightarrow, from=U, to=D, bend left=55, "(\eta'_\alpha\text{,}\,\eta'_{\alpha\beta})", ""{name=R, below}] \arrow[Rightarrow, from=U, to=D, bend right=55, "(\eta_\alpha\text{,}\,\eta_{\alpha\beta})"', ""{name=L}] \tarrow[from=L, to=R, end anchor={[yshift=0.6ex]}]{r} \arrow["(\epsilon_\alpha)", phantom, from=L, to=R, end anchor={[yshift=0.6ex]}, bend left=34]
\end{tikzcd}\right\}
\end{equation}
\end{remark}

\begin{definition}[Flat and trivial gerbe]
A \textit{flat gerbe} is defined as a gerbe $(B_\alpha,\Lambda_{\alpha\beta},G_{\alpha\beta\gamma})$ with vanishing curvature $\mathrm{d}B_\alpha=0$. We use the symbol $\flat\mathbf{B}^2U(1)_{\mathrm{conn}}$ for the moduli stack of flat gerbes with connection. A \textit{trivial gerbe} is defined as a gerbe with trivial Dixmier-Douady class.
\end{definition}

\begin{remark}[Flat and trivial gerbe in \v{C}ech picture]
Let us express in local data a flat gerbe $(B_\alpha,\Lambda_{\alpha\beta},G_{\alpha\beta\gamma})\in\mathbf{H}(M,\flat\mathbf{B}^2U(1)_{\mathrm{conn}})$. Since $B_\alpha$ is closed on each patch $U_\alpha$ we can rewrite
\begin{equation}\label{eq:flatgerbe}
    \begin{aligned}
        B_\alpha \,=&\; \mathrm{d}\eta_\alpha, \\
        \Lambda_{\alpha\beta} \,=&\; \eta_\alpha-\eta_\alpha + \mathrm{d}\eta_{\alpha\beta}, \\
        G_{\alpha\beta\gamma} \,=&\; \eta_{\alpha\beta} + \eta_{\beta\gamma} + \eta_{\gamma\alpha} +c_{\alpha\beta\gamma}, \\
        &\; c_{\alpha\beta\gamma} - c_{\beta\gamma\delta} + c_{\gamma\delta\alpha} - c_{\delta\alpha\beta} \in 2\pi\mathbb{Z}
    \end{aligned}
\end{equation}
Hence flat gerbes are classified by \textit{holonomy classes} $[c_{\alpha\beta\gamma}]\in H^2(M,U(1)_{\mathrm{discr}})$. The \v{C}ech local data of a trivial gerbe will be exactly \eqref{eq:flatgerbe}, but with trivial constants $c_{\alpha\beta\gamma}=0$.
\end{remark}

\begin{definition}[Flat holonomy class]\label{eq:holonomyclass}
Flat gerbes are classified by elements of the cohomology group $H^2(M,U(1)_{\mathrm{discr}})\cong \mathrm{Hom}\big(H_2(M),U(1)_{\mathrm{discr}}\big)$, where $U(1)_{\mathrm{discr}}$ is the circle equipped with discrete topology. Such class is called \textit{flat holonomy class} of the gerbe.
\end{definition}

\noindent Hence a class $[c_{\alpha\beta\gamma}]\in H^2(M,U(1)_{\mathrm{discr}})$ encode the holonomy of the gerbe, meaning that to any surface $[\Sigma]\in H_2(M)$ of the base manifold will be associated an angle $\mathrm{hol}(\Sigma,B_\alpha)\in U(1)$.

\begin{remark}[Flat gerbe has torsion Dixmier-Douady class]
There exists a natural map $H^2(M,U(1)_{\mathrm{discr}})\rightarrow H^2(M,U(1))\cong H^3(M,\mathbb{Z})$ sending a flat gerbe to its Dixmier-Douady class. The Dixmier-Douady class of the flat gerbe has not to be zero, but its image in the de Rham cohomology $H^3(M,\mathbb{Z})\rightarrow H^3(M,\mathbb{R})\cong H^3_{\mathrm{dR}}(M)$ must be, since $\mathrm{d}B_\alpha=0$. This implies that the Dixmier-Douady class is, in general, torsion.
\end{remark}

\begin{definition}[Sections of a $n$-bundle]\label{def:secgroupoid}
Given any bundle $\pi:P\rightarrow M$ with $P,M\in\mathbf{H}$, according to \cite{Principal1} we can define its $n$-groupoid of its \textit{sections} on $M$ by 
\begin{equation}
    \Gamma(M,P) \,:=\, \mathbf{H}_{/M}(\mathrm{id}_M,\pi)
\end{equation}
where $\mathbf{H}_{/M}(-,-)$ is the internal hom $n$-groupoid (definition \ref{def:inthom}) of the slice $n$-category $\mathbf{H}_{/M}$ (definition \ref{def:slice}).
\end{definition}

\subsection{Finite symmetries of n-bundles}

In this subsection we will explain and apply some definitions from \cite{FSS16} to obtain the $n$-group of finite symmetries of a principal $n$-bundle.

\begin{definition}[Automorphism $n$-groupoid]
Given any stack $X\in\mathbf{H}$ we define its \textit{automorphism groupoid} $\Aut(X)$ as the subgroupoid of $\mathbf{H}(X,X)$ of invertible morphisms. For a given morphism $f:X\rightarrow Y$ the automorphism groupoid $\Aut_{/}(f)$ is analogously defined as the subgroupoid of $\mathbf{H}_{/}(f,f)$ of invertible morphisms.
\end{definition}

\begin{example}[Automorphisms of principal $n$-bundles]
Let $P\rightarrow M$ be a principal $n$-bundle given by $f:M\rightarrow\mathbf{B}G$. The automorphism $n$-group of $f$ (i.e. the $n$-group of automorphisms of $P$ preserving the principal structure) will sit at the center of a short exact sequence of $n$-groups
\begin{equation}\label{eq:autdef}
    1\longrightarrow\Omega_f\mathbf{H}(M,\mathbf{B}G)\longrightarrow \Aut_{/}(f) \longrightarrow \mathrm{Diff}(M)\longrightarrow 1.
\end{equation}
We will also equivalently use the semidirect product notation $\Aut_{/}(f) = \mathrm{Diff}(M)\ltimes\Omega_f\mathbf{H}(M,\mathbf{B}G)$.
\end{example}

\begin{example}[Automorphisms of ordinary $G$-bundles]\label{ex:nonabelianaut}
Let $G$ be an ordinary Lie group and let $P$ be an ordinary principal $G$-bundle given by the cocycle $f:M\rightarrow\mathbf{B}G$. Hence we have the isomorphism $\Omega_f\mathbf{H}(M,\mathbf{B}G)\cong \Gamma\big(M,\mathrm{Ad}(P)\big)$, where the associated bundle $\mathrm{Ad}(P) := P\times_G G$ with the adjoint action $\mathrm{Ad}:G\times G\rightarrow G$ is just the non-linear \textit{adjoint bundle} of $P$. So we have the usual automorphism group of a principal $G$-bundle
\begin{equation}
    1\longrightarrow\Gamma\big(M,\mathrm{Ad}(P)\big)\longrightarrow \Aut_{/}(f) \longrightarrow \mathrm{Diff}(M)\longrightarrow 1.
\end{equation}
\end{example}

\begin{example}[Automorphisms of circle bundles]
For the ordinary case $G=U(1)$ we have $\Omega_f\mathbf{H}(M,\mathbf{B}U(1))\cong \Coo(M,U(1))$ and hence the usual automorphism $1$-group of a circle bundle
\begin{equation}
    1\longrightarrow\Coo(M,U(1))\longrightarrow \Aut_{/}(f) \longrightarrow \mathrm{Diff}(M)\longrightarrow 1.
\end{equation}
\end{example}

\begin{example}[Automorphisms of gerbes]\label{ex:aut}
It is possible to prove there exists an equivalence of $2$-groups $\Omega_f\mathbf{H}\big(M,\mathbf{B}(\BU)\big) \cong \mathbf{H}\big(M,\BU\big)$ for any gerbe $f:M\rightarrow\mathbf{B}(\BU)$. Therefore global gauge transformations of this gerbe are global circle bundles with connection on $M$. Thus the $2$-group of automorphisms will sit at the center of the exact sequence of $2$-groups 
\begin{equation}\label{eq:finiteautomorphismseq}
    1\longrightarrow\mathbf{H}\big(M,\BU\big)\longrightarrow \Aut_{/}(f) \longrightarrow \mathrm{Diff}(M)\longrightarrow 1.
\end{equation}
Let us introduce the curvature map of stacks
\begin{equation}
    \mathrm{curv}:\BU\longrightarrow\Omega^2_{\mathrm{cl}}
\end{equation}
which maps a circle bundle $(\eta_\alpha,\eta_{\alpha\beta})$ over a manifold $M$ into a global closed $2$-form $b\in\Omega_{\mathrm{cl}}^2(M)$ such that $b|_{U_\alpha} = \mathrm{d}\eta_\alpha$. Then gauge transformations can be expressed as global $B$-shifts of the form $B_\alpha\mapsto B_\alpha+b$.
Notice $\mathrm{Diff}(M)\ltimes\Omega_{\mathrm{cl}}^2(M)$ is the gauge group proposed by \cite{Hull14} for DFT.
\end{example}

\noindent From remark \ref{ex:aut} we know $1$-morphisms between gerbes over $M$ are circle bundles over $M$ and $2$-morphisms are gauge transformations between these circle bundles. This corresponds, in general, to the idea that global gauge transformations of $n$-gerbes are $(n-1)$-gerbes and so on. This is a clear categorical feature of these geometrical objects.

\subsection{Infinitesimal symmetries of n-bundles and Generalized Geometry}

In this subsection we will deal with the infinitesimal automorphisms of a principal $n$-bundles and we will show how they are related to the more familiar Generalized Geometry (see \cite{Gua11}).

\begin{definition}[Atiyah $n$-algebroids]
Let $P\rightarrow M$ be a principal $n$-bundle corresponding to a map $f:M\rightarrow\mathbf{B}G$. The Atiyah $n$-algebroid of this principal $n$-bundle was defined in \cite{FSS16} as the Lie differentiation of its automorphism $n$-groupoid
\begin{equation}
    \mathfrak{at}(P) \,:=\, \mathrm{Lie}\big(\Aut_/(f)\big).
\end{equation}
\end{definition}

\noindent This $n$-algebra encodes the infinitesimal symmetries of the principal structure. By differentiating sequence \eqref{eq:autdef} we have that it will sit at the center of the short exact sequence of $n$-algebras
\begin{equation}
    0\longrightarrow\mathrm{Lie}\big(\Omega_f\mathbf{H}(M,\mathbf{B}G)\big)\longrightarrow \mathfrak{at}(P) \longrightarrow \mathfrak{X}(M)\longrightarrow 0.
\end{equation}

\begin{example}[Ordinary Atiyah algebroid]
If $P\rightarrow M$ is a principal $G$-bundle for some ordinary Lie group $G$ we get the short exact sequence of ordinary algebras
\begin{equation}
    0\longrightarrow\Gamma\big(M,\,\mathrm{ad}(P)\big)\longrightarrow \mathfrak{at}(P) \longrightarrow \mathfrak{X}(M)\longrightarrow 0.
\end{equation}
where $\mathrm{ad}(P):=P\times_G\mathfrak{g}$ with adjoint action $\mathrm{ad}:G\rightarrow\mathfrak{g}$ is the linear \textit{adjoint bundle} of $P$.
\end{example}

\begin{example}[Ordinary Atiyah algebroid of a circle bundle]
If $P\rightarrow M$ is a circle bundle we get the familiar short exact sequence of ordinary algebras
\begin{equation}
    0\longrightarrow\Coo(M,\mathbb{R})\longrightarrow \mathfrak{at}(P) \longrightarrow \mathfrak{X}(M)\longrightarrow 0.
\end{equation}
Locally, on any patch $U\subset M$, this reduces to the familiar algebra $\mathfrak{at}(P)|_U = \mathfrak{X}(U)\oplus\Coo(U)$ of infinitesimal gauge transformation of an abelian gauge field.
\end{example}

\begin{example}[Courant $2$-algebroid]
If $P\rightarrow M$ is a gerbe with connection data corresponding to a map $M\rightarrow\mathbf{B}(\BU)$, as explained in \cite{Col11}, we get that the Atiyah $2$-algebroid is the so-called \textit{Courant }$2$\textit{-algebra} sitting in the short exact sequence of $2$-algebras
\begin{equation}
    0\longrightarrow\mathbf{H}(M,\mathbf{B}\mathbb{R}_{\mathrm{conn}})\longrightarrow \mathfrak{at}(P) \longrightarrow \mathfrak{X}(M)\longrightarrow 0.
\end{equation}
Locally, on a patch $U\subset M$, this reduces the familiar Courant $2$-algebra of infinitesimal gauge transformations of the gerbe, whose underlying complex is just
\begin{equation}
    \mathfrak{at}(P)|_U \,\simeq\, \Big( \Coo(U)\,\xrightarrow{\;\mathrm{d}\;}\,\mathfrak{X}(U)\oplus\Omega^1(U) \Big).
\end{equation}
\end{example}

\section{Proposal: DFT is Higher Kaluza-Klein Theory}
In this section we will give a formal definition of Higher Kaluza-Klein Theory and we will explain how this can be interpreted as a global version of Double Field Theory (DFT). 

\vspace{0.2cm}

\noindent In the following, \textbf{Digression} paragraphs will be entirely devoted to discuss and clarify how main existing proposals of DFT geometry in literature are actually Higher Kaluza-Klein in disguise.

\subsection{Doubled space is the total space of a 2-bundle}
This subsection will devoted to explain how a globally defined version of doubled space can be defined as the total space of a $2$-bundle and why this solves the gluing problem of DFT. Moreover in our formulation we will recover many previous geometrical ideas for DFT, such as Papadopoulos' C-spaces, para-Hermitian geometry and Park's geometry.

\begin{post}[Doubled space]\label{post1}
A doubled space $\M$ is defined as the total space of a principal $\BU$-bundle $\M\xrightarrow{\bbpi}M$ over a Riemannian manifold $M$.
\end{post}

\begin{remark}[$\BU$ is a group-stack]\label{rem:gs}
The stack $\BU$ of circle bundles with connection is a \textit{group-stack}, which means that it satisfies the ordinary defining properties of a group up to an isomorphism. First of all $\BU$ is naturally equipped with a tensor product
\begin{equation}
    \otimes:\;\BU \,\times\, \BU\; \longrightarrow\; \BU
\end{equation}
which maps a couple of circle bundles $P_1\rightarrow M$ and $P_2\rightarrow M$ to a new one $P_1\otimes P_2\rightarrow M$. Moreover the dual bundle $P^\ast\rightarrow M$ of any circle bundle $P\rightarrow M$ plays the role of its \textit{inverse element}, while the trivial circle bundle $M\times U(1)\rightarrow M$ with trivial connection plays the role of the \textit{identity element} $\mathrm{id}$. It is easy to verify that ordinary group properties
\begin{equation}
    \begin{aligned}
    P^\ast\otimes P\cong\mathrm{id}, \quad P\otimes P^\ast\cong \mathrm{id}, \\
    P_1\otimes(P_2\otimes P_3)\cong (P_1\otimes P_2)\otimes P_3
    \end{aligned}
\end{equation}
are satisfied only up to gauge transformation of circle bundles. In local \v{C}ech data on the manifold $M$ we have  $(\eta_\alpha,\eta_{\alpha\beta})\otimes(\eta'_\alpha,\eta'_{\alpha\beta})= (\eta_\alpha+\eta'_\alpha,\eta_{\alpha\beta}+\eta'_{\alpha\beta})$ and  $(\eta_\alpha,\eta_{\alpha\beta})^\ast:=(-\eta_\alpha,-\eta_{\alpha\beta})$.
\end{remark}

\begin{remark}[Doubled space as a principal $2$-bundle in \v{C}ech picture]
Let us now give a concrete description of this geometrical object. The stack $\mathbf{B}(\BU)$ is nothing but the classifying $2$-stack of gerbes with with connective structure, but without "curving": we will explain this explicitly. If we call $M$ the smooth base manifold that we identify with usual spacetime, we can write the doubled space through the following pullback diagram in the category of $\infty$-stacks
\begin{equation}\begin{tikzcd}[row sep=7ex, column sep=6ex]
\M \arrow[d, "\bbpi"']\arrow[r] &\ast \arrow[d]  \\
M\arrow[r, "f"] &\mathbf{B}(\mathbf{B}U(1)_{\mathrm{conn}})
\end{tikzcd}\end{equation}
We can choose any good cover $\mathcal{U}:=\{U_\alpha\}$ for the base manifold $M$ and immediately write its \v{C}ech groupoid $\check{C}(\mathcal{U})$ as the following simplicial object
\begin{equation}\label{eq:cechgroupoidsimplicial}
    \begin{tikzcd}[row sep=scriptsize, column sep=3ex] \cdots\arrow[r, yshift=1.8ex]\arrow[r, yshift=0.6ex]\arrow[r, yshift=-0.6ex]\arrow[r, yshift=-1.8ex] & \bigsqcup_{\alpha\beta\gamma}U_{\alpha}\cap U_\beta\cap U_\gamma\arrow[r, yshift=1.4ex]\arrow[r]\arrow[r, yshift=-1.4ex]& \bigsqcup_{\alpha\beta}U_{\alpha}\cap U_\beta  \arrow[r, yshift=0.7ex] \arrow[r, yshift=-0.7ex] & \; \bigsqcup_{\alpha} U_{\alpha} \arrow[r] & \check{C}(\mathcal{U}).
    \end{tikzcd}
\end{equation}
Now, by using the natural equivalence of $\check{C}(\mathcal{U})$ and the manifold $M$ in the category of $\infty$-groupoids, we can express the functor between $M$ and the moduli stack as a map of the form
\begin{equation}
\begin{tikzcd}[row sep=scriptsize, column sep=8ex] 
    M \;\simeq\; \check{C}(\mathcal{U}) \arrow[r, "f"] & \mathbf{B}(\mathbf{B}U(1)_{\mathrm{conn}}).
\end{tikzcd}
\end{equation}
According to remark \ref{rem:dkc} we can use Dold-Kan correspondence to rewrite this map of stacks as map $f=(G,\Lambda,0)$ from the complex \eqref{eq:cechgroupoidsimplicial} to the following cochain complex of abelian sheaves
\begin{equation}
\begin{tikzcd}[row sep=scriptsize, column sep=8ex] 
    \check{C}(\mathcal{U})_\bullet \,\arrow[r, "(G\text{,}\Lambda\text{,}0)\;"] &\, \bigg(\Coo(-)\xrightarrow{\;\mathrm{d}\;}\Omega^1(-)\xrightarrow{\;0\;}0\bigg).
\end{tikzcd}
\end{equation}
These maps from the \v{C}ech groupoid \eqref{eq:cechgroupoidsimplicial} to the moduli stack are therefore of the form
\begin{equation}
    \Lambda \in\Omega^1\bigg(\bigsqcup_{\alpha\beta} U_{\alpha}\cap U_{\beta}\bigg), \;\quad
    G \in\Coo\bigg(\bigsqcup_{\alpha\beta\gamma} U_{\alpha}\cap U_{\beta}\cap U_\gamma\bigg),
\end{equation}
or equivalently a collection of $1$-forms $\Lambda_{\alpha\beta}\in\Omega^1(U_\alpha \cap U_\beta)$ and scalars $G_{\alpha\beta\gamma}\in\mathcal{C}^\infty(U_\alpha \cap U_\beta\cap U_\gamma)$, such that on three-fold and four-fold overlaps they are subject to the patching conditions
\begin{equation}\label{eq:localdatagerbe}
    \begin{aligned}
        \Lambda_{\alpha\beta}+\Lambda_{\beta\gamma}+\Lambda_{\gamma\alpha} &= \mathrm{d}G_{\alpha\beta\gamma},\\
        G_{\alpha\beta\gamma}-G_{\beta\gamma\delta}-G_{\gamma\delta\alpha}+G_{\delta\alpha\beta}&\in 2\pi\mathbb{Z}.
    \end{aligned}
\end{equation}
Now we can construct the groupoid $\Gamma(M,\mathcal{M}):=\mathbf{H}_{/M}(\mathrm{id}_M,\bbpi)$ of sections of the bundle $\mathcal{M}\xrightarrow{\bbpi}M$ according to definition \ref{def:secgroupoid}. They will be given by a collection $(\tilde{x}_\alpha,\phi_{\alpha\beta})\in\Gamma(M,\mathcal{M})$ where
\begin{equation}
    \tilde{x}_\alpha\in\Omega^1(U_\alpha), \;\quad
    \phi_{\alpha\beta}\in\Coo(U_\alpha\cap U_\beta)
\end{equation}
are local $1$-forms and scalars, such that they are patched on two-fold and three-fold overlaps by using $(\Lambda_{\alpha\beta},G_{\alpha\beta\gamma})$ as transition functions by
\begin{equation}\label{eq:coords}
    \begin{aligned}
        \tilde{x}_\alpha - \tilde{x}_\beta  \,&=\, -\Lambda_{\alpha\beta} + \mathrm{d}\phi_{\alpha\beta},\\
        \phi_{\alpha\beta}+\phi_{\beta\gamma}+\phi_{\gamma\alpha} \,&=\, G_{\alpha\beta\gamma} \quad\! \mathrm{mod}\,2\pi\mathbb{Z}.
    \end{aligned}
\end{equation}
Gauge transformations between global sections are given by a collection of local functions on each patch $\epsilon_\alpha\in\Coo(U_\alpha)$ so that $(\tilde{x}_\alpha,\phi_{\alpha\beta})\mapsto(\tilde{x}_\alpha+\mathrm{d}\epsilon_\alpha,\, \phi_{\alpha\beta}+\epsilon_\alpha-\epsilon_\beta)$.
Global sections $(\tilde{x}_\alpha,\phi_{\alpha\beta})$ and gauge transformations $(\epsilon_\alpha)$ are respectively the objects and the morphisms of the groupoid $\Gamma(M,\mathcal{M})$ of sections of the doubled space.
\end{remark}

\begin{remark}[Topological classification of doubled spaces]\label{rem:ddds}
Since doubled spaces $\mathcal{M}$ on a base manifold $M$ are defined by cocycles $(\Lambda,G):M\rightarrow\BU$, this immediately implies that they are topologically classified by their Dixmier-Douady classes $[G_{\alpha\beta\gamma}]\in H^3(M,\mathbb{Z})$.
\end{remark}

\begin{remark}[Sections of the doubled space are twisted circle bundles]
As explained by \cite{Principal1}, sections of a principal $2$-bundle $f:M\rightarrow\mathbf{B}G$ can be interpreted as ordinary principal bundles on $M$ twisted by the cocycle $f$. Coherently with this, in the case of the doubled space we can immediately interpret sections $(\tilde{x}_\alpha,\phi_{\alpha\beta})\in\Gamma(M,\M)$, which are patched according to \eqref{eq:coords}, as $U(1)$-bundles with connection on $M$ twisted by the \v{C}ech cocycle $(\Lambda_{\alpha\beta},G_{\alpha\beta\gamma})$.
\end{remark}

\noindent Let us now look at the example of doubled space with trivial Dixmier-Douady class in detail.

\begin{example}[Trivial doubled space]
Let us assume that the doubled space $\mathcal{M}_{\mathrm{triv}}$ is flat and trivial. Then its sections $(\tilde{x}_\alpha,\phi_{\alpha\beta})\in\Gamma(M,\mathcal{M}_{\mathrm{triv}})$ can be written accordingly to \eqref{eq:coords} as
\begin{equation}
    \begin{aligned}
        \tilde{x}_\alpha - \tilde{x}_\beta &= \mathrm{d}\phi_{\alpha\beta},\\
        \phi_{\alpha\beta}+\phi_{\beta\gamma}+\phi_{\gamma\alpha}&\in2\pi\mathbb{Z},
    \end{aligned}
\end{equation}
which is the local data of a globally defined circle bundle on $M$. The groupoid of sections of a trivial doubled space $\Gamma(M,\mathcal{M}_{\mathrm{triv}})\cong \mathbf{H}(M,\mathbf{B}U(1)_{\mathrm{conn}})$ is then the groupoid of circle bundles with connection on $M$. This is exactly the analogue of a trivial circle $1$-bundle on $M$, which admits globally defined sections $\Coo\big(M,U(1)\big)$. Notice that $\mathcal{M}_{\mathrm{triv}}\neq T^\ast M$, as one could think.
\end{example}

\begin{remark}[Groupoid interpretation of doubled space]
The doubled space $\mathcal{M}$ can be interpreted as a groupoid, locally made up of patches $T^\ast U_\alpha$ with $U_\alpha\subset M$. Indeed points $(x_\alpha,\tilde{x}_\alpha)\in T^\ast U_\alpha$ are glued together by morphisms $(x_\alpha,\tilde{x}_\alpha)\mapsto(x_\beta,\tilde{x}_\beta-\Lambda_{\alpha\beta}(x_\beta)+\mathrm{d}\phi_{\alpha\beta}(x_\beta))$ which must satisfy the composition rule $(\Lambda_{\alpha\beta},\,\phi_{\alpha\beta}) \circ (\Lambda_{\beta\gamma},\,\phi_{\beta\gamma}) := (\Lambda_{\alpha\gamma},\,\phi_{\alpha\gamma})$ on three-fold overlaps, according with \eqref{eq:localdatagerbe} and \eqref{eq:coords}. We will have the following picture
\begin{equation}
    \mathcal{M}\simeq \left\{ \begin{tikzcd}[row sep=12ex, column sep=5ex]
    & (x_\beta,\tilde{x}_\beta)\arrow[rd, "(\Lambda_{\beta\gamma}\text{,}\phi_{\beta\gamma})"] \arrow[Rightarrow, from=1-2, to=D, end anchor={[xshift=3.3ex, yshift=2.2ex]}, "\,G_{\alpha\beta\gamma}"] & \\
   (x_\alpha,\tilde{x}_\alpha)\arrow[rr, ""{name=D}, "(\Lambda_{\alpha\gamma}\text{,}\phi_{\alpha\gamma})"']{} \arrow[ru, "(\Lambda_{\alpha\beta}\text{,}\phi_{\alpha\beta})"] & & (x_\gamma,\tilde{x}_\gamma)
    \end{tikzcd} \right\}
\end{equation}
Patching conditions \eqref{eq:coords} do not uniquely identify two points in two different charts $T^\ast U_\alpha$ and $T^\ast U_\beta$, like we are used for manifolds, but they are given up to a choice of gauge transformation $\mathrm{d}\phi_{\alpha\beta}$ on overlaps. This idea that patches are glued together not by bare identities, but by gauge transformations is a totally new feature of stringy geometry.
\end{remark}

\noindent We just constructed a principal $\BU$-bundle without adding any extra structure nor condition. We obtained something which may look a unfamiliar at first sight, but actually we can yet immediately recognize some important properties from the existing literature which are supposed to appear in geometry of DFT. In this regard see digressions \ref{diga}, \ref{digb}, \ref{digc} and \ref{digd}.

\begin{digression}[Recovering Papadopolous' C-space]\label{diga}
Recall that sections $(\tilde{x}_\alpha,\phi_{\alpha\beta})\in\Gamma(M,\mathcal{M})$ of the doubled space are patched by condition \eqref{eq:coords}. Notice that these are exactly of the same form of the coordinates of a \textit{C-space}, defined by Papadopolous in \cite{Pap13}, \cite{Pap14} and further developed by \cite{HowPap17}, which was prescribed in the references to accommodate DFT geometry. Hence our construction of $\M$ may be seen also as a formalization of that intuition.
\end{digression}

\begin{digression}[Recovering Vaisman's para-K\"ahler geometry]\label{digb}
The doubled space can be seen as locally given by a collection of cotangent bundles $T^\ast U_\alpha\rightarrow U_\alpha$. In Darboux coordinates $\{x^\mu,\tilde{x}_\mu\}$ on each $T^\ast U_\alpha$, we can define tautological $1$-forms $\tilde{x}_\alpha:=\tilde{x}_\mu\mathrm{d}x^\mu$ and symplectic forms 
\begin{equation}
 \omega_{0\alpha} = \mathrm{d}\tilde{x}_\mu\wedge\mathrm{d}x^\mu, \qquad \omega_{0\beta}-\omega_{0\alpha} = \mathrm{d}\Lambda_{\alpha\beta}.
\end{equation}
By rearranging the projector $\bbpi_\alpha:T^\ast U_\alpha\rightarrow U_\alpha$ locally given by $(x^\mu,\tilde{x}_\mu)\mapsto (x^\mu)$ of the bundle we can canonically construct a local para-complex structure $J_{0\alpha} := 2(\bbpi_\alpha)_\ast-\mathrm{id}$ and thus a natural local para-K\"{a}hler metric $\eta_{0\alpha}(-,-):=\omega_{0\alpha}(J_{0\alpha}-,-)$. Hence we recover locally the para-K\"aler formalism which was presented in \cite{Vais12}. We could say that the doubled space $\mathcal{M}$ is locally para-K\"ahler, meaning that it can be locally thought as a collection of para-K\"ahler patches $(T^\ast U_\alpha,\, J_{0\alpha}, \,\eta_{0\alpha})$. However we remark the doubled space cannot globally be a para-K\"ahler manifold. Let $N\rightarrow M$ be a fibration patch-wise given by the projections $\bbpi_\alpha:T^\ast U_\alpha \rightarrow U_\alpha$. This manifold is then glued by a collection of $1$-forms $\Lambda_{\alpha\beta}\in\Omega^1(U_\alpha\cap U_\beta)$, but this time it satisfies the condition $\Lambda_{\alpha\beta}+\Lambda_{\beta\gamma}+\Lambda_{\gamma\alpha}=0$ on three-fold overlaps. Therefore it cannot be interpreted as \v{C}ech data of a gerbe \eqref{eq:localdatagerbe} and it does not solve Papadopoulos' puzzle \eqref{eq:papa}. Hence a full geometrization of the gerbe cannot be given by a manifold. However Vaisman's approach has been perfected by para-Hermitian geometry, which we deal with in digression \ref{digd}.
\end{digression}

\begin{remark}[Principal action on the doubled space]\label{principalaction}
By definition (see postulate \ref{post1}) the doubled space $\M$ is a principal $2$-bundle, therefore it will be canonically equipped with a higher principal action $\rho:\BU\times\M\longrightarrow\M$. This means that we will have not only transformations, but also isomorphisms between them. On the base manifold this is a functor
\begin{equation}
    \mathbf{H}(M,\BU)\,\times\, \Gamma(M,\mathcal{M}) \;\longrightarrow\; \Gamma(M,\mathcal{M}).
\end{equation}
Since local sections of $\mathcal{M}$ are local circle bundles $P_\alpha$, for any $Q\in\mathbf{H}(M,\BU)$, this action is locally given by the tensor product of $\BU$ from remark \ref{rem:gs}, i.e. by $(Q,\,P_\alpha)\mapsto Q\otimes P_\alpha$.
\end{remark}

\noindent Notice this generalizes the principal circle action of ordinary Kaluza-Klein Theory. Indeed a section of a circle bundle is in local data a collection of $U(1)$-functions $\theta_\alpha$ and the $U(1)$-action is given by a global shift $(g,\,\theta_\alpha)\mapsto g\cdot\theta_\alpha$, where $g$ is a global $U(1)$-function.

\begin{remark}[Principal action on the doubled space in \v{C}ech picture]
This principal action  on the doubled space is described in local \v{C}ech data on sections $(\tilde{x}_\alpha,\phi_{\alpha\beta})\in\Gamma(M,\mathcal{M})$ as follows:
\begin{itemize}
    \item a circle bundle on $M$, given by the \v{C}ech cocycle $(\eta_\alpha,\eta_{\alpha\beta})\in\mathbf{H}(M,\BU)$ acting by
    \begin{equation}\label{eqgauge}
        \begin{aligned}
            x_\alpha \,&\mapsto\, x_\alpha \\
            \tilde{x}_\alpha \,&\mapsto\, \tilde{x}_\alpha + \eta_\alpha \\
            \phi_{\alpha\beta} \,&\mapsto\, \phi_{\alpha\beta} + \eta_{\alpha\beta},
        \end{aligned}
    \end{equation}
    \item gauge transformations $(\eta_\alpha,\eta_{\alpha\beta})\Mapsto(\eta'_\alpha,\eta'_{\alpha\beta})$ are just gauge transformations between circle bundles and are given in local data by a collection of functions $\epsilon_\alpha\in\Coo(U_\alpha)$ acting by
    \begin{equation}\label{eqgaugeofgauge}
        \begin{aligned}
            \eta_\alpha \,&\Mapsto\, \eta_\alpha + \mathrm{d}\epsilon_\alpha \\
            \eta_{\alpha\beta} \,&\Mapsto\, \eta_{\alpha\beta}+\epsilon_\alpha-\epsilon_\beta.
        \end{aligned}
    \end{equation}
\end{itemize}
In terms of diagrams we can rewrite the principal action on the doubled space as the groupoid
\begin{equation}
   \mathbf{H}(M,\,\BU) \;\simeq\; \left\{\begin{tikzcd}[row sep=scriptsize, column sep=15ex]
    \M \arrow[r, bend left=50, ""{name=U, below}, "(\eta_\alpha\text{, }\eta_{\alpha\beta})"]
    \arrow[r, bend right=50, "(\eta_\alpha + \mathrm{d}\epsilon_\alpha\text{, }\eta_{\alpha\beta}+\epsilon_\alpha-\epsilon_\beta)"', ""{name=D}]
    & \M
    \arrow[Rightarrow, from=U, to=D, "(\epsilon_\alpha)"]
\end{tikzcd}\right\}
\end{equation}
\end{remark}

\begin{remark}[Principal action gives global gauge transformations]\label{rem:globgauge}
Suppose to equip our \v{C}ech cocycle  $(\Lambda_{\alpha\beta},G_{\alpha\beta\gamma})$ with a gerbe connection $B_\alpha\in\Omega^2(U_\alpha)$ satisfying $B_\beta-B_\alpha =\mathrm{d}\Lambda_{\alpha\beta}$. Then principal action $\rho$ (remark \ref{principalaction}) gives global gauge transformations $(\eta_{\alpha},\eta_{\alpha\beta})\in\mathbf{H}\big(M,\BU\big)$ of the gerbe connection. From the expression of coboundaries \eqref{eq:coboundaries} we find
\begin{equation}
    \begin{aligned}
        B_\alpha \;&\mapsto\; B_\alpha + \mathrm{d}\eta_\alpha \\
        \Lambda_{\alpha\beta} \;&\mapsto\; \Lambda_{\alpha\beta} + \eta_\alpha- \eta_\beta +\mathrm{d}\eta_{\alpha\beta} \;\;\,\,=  \,\;\Lambda_{\alpha\beta}\\
        G_{\alpha\beta\gamma} \;&\mapsto\; G_{\alpha\beta\gamma} + \eta_{\alpha\beta}  + \eta_{\beta\gamma}  + \eta_{\gamma\alpha} = \,\;G_{\alpha\beta\gamma}.
    \end{aligned}
\end{equation}
and the \eqref{eqgaugeofgauge} are the gauge transformations of these gauge transformations. The gerbe curvature $H=\mathrm{d}B_\alpha$ is clearly unaffected.
Transformation \eqref{eqgauge} and \eqref{eqgaugeofgauge} can be also understood as a change of local trivialization (see \cite{Hit99}) for the gerbe cocycle $(\Lambda_{\alpha\beta},G_{\alpha\beta\gamma})$: where for a circle bundle this is a global $U(1)$-valued function, for a gerbe this is indeed a global circle bundle.
By using the curvature functor $\mathrm{curv}:\BU\longrightarrow\Omega^2_{\mathrm{cl}}$, which sends a circle bundle $(\eta_{\alpha},\eta_{\alpha\beta})\in\mathbf{H}\big(M,\BU\big)$ in the closed global $2$-form $b\in\Omega^2_{\mathrm{cl}}(M)$ defined by $b|_{U_\alpha} = \mathrm{d}\eta_\alpha$, we can rewrite this transformation as a global $B$-shift $B_\alpha\mapsto B_\alpha + b$.
\end{remark}

\noindent This generalizes ordinary Kaluza-Klein, where the principal $U(1)$-action on the circle bundle is given by global shifts in the angular coordinates $\theta_\alpha\mapsto\theta_\alpha+\eta_\alpha$ with $\eta_\alpha=\eta_\beta$ and it encodes global gauge transformations $A_\alpha\mapsto A_\alpha+\mathrm{d}\eta_\alpha$ and $G_{\alpha\beta}\mapsto G_{\alpha\beta} + \eta_\alpha - \eta_\beta = G_{\alpha\beta} $.

\begin{digression}[Recovering coordinate gauge symmetry]\label{digc}
The idea of \textit{coordinate gauge symmetry} was proposed for the doubled space by \cite{Park13}, then further explored by \cite{Par13x}, \cite{Par16x}, \cite{Par17x} and \cite{Par18x}. The reference noticed that, under a local gauge transformation on the doubled space $(x^\mu,\,\tilde{x}_\mu)\mapsto(x^\mu,\,\tilde{x}_\mu+\eta_\mu)$ with $\eta=\varphi\,\mathrm{d}\phi$ where $\phi$ and $\varphi$ are scalars, the Kalb-Ramond field transforms by $B\mapsto B+\mathrm{d}\eta$. Since these transformations are abelian in nature, we can generally have sums $\eta=\varphi^i\,\mathrm{d}\phi_i$. These local transformations agree with remarks \ref{rem:globgauge} and \ref{principalaction}. Moreover their geometric meaning is clarified in the Higher Kaluza-Klein framework: they are the local data of the principal action on the doubled space.
Notice that for $\varphi=1$ these reduce to shifts of the form $(x^\mu,\,\tilde{x}_\mu)\mapsto(x^\mu,\,\tilde{x}_\mu+\partial_\mu\phi)$, which do not change the Kalb-Ramond field $B\mapsto B$. Hence the gauge transformations $\mathrm{d}\phi_{\alpha\beta}$ which appears on overlaps of patches \eqref{eq:coords} are exactly of this form and they do not affect the Kalb-Ramond field. See also digression \ref{dig:park2} for further discussion of the implications of coordinate gauge symmetry.
\end{digression}

\begin{definition}[Global differential forms on the doubled space]\label{def:globalformM}
We define a notion of global differential forms on the doubled space $\mathcal{M}$ as it follows
\begin{equation}
    \Omega^\bullet(\mathcal{M})_{\mathrm{glob}}\,:=\,\big\{\xi_\alpha\in\Omega^\bullet(T^\ast U_\alpha)\,|\;\xi_\alpha=\xi_\beta \text{  on  } T^\ast(U_\alpha\cap U_\beta) \big\}
\end{equation}
\end{definition}

\begin{theorem}[Gerbe connection]\label{thm:gerbeconn}
Let $B_\alpha$ be a local $2$-form on $M$ such that $H=\mathrm{d}B_\alpha$ and patched by $B_\beta-B_\alpha = -\mathrm{d}\Lambda_{\alpha\beta}$ on $U_\alpha \cap U_\beta$. Thus, given the tautological $1$-form $\tilde{x}_\alpha$ on each $T^\ast U_\alpha$, the $2$-form $\omega_B:=\mathrm{d}\tilde{x}_\alpha - B_\alpha\in\Omega^2(\mathcal{M})_{\mathrm{glob}}$ is a global $2$-form on the doubled space $\mathcal{M}$.
\end{theorem}
\begin{proof}
By differentiating the first equation $\tilde{x}_\alpha - \tilde{x}_\beta  = -\Lambda_{\alpha\beta} + \mathrm{d}\phi_{\alpha\beta}$ of \eqref{eq:coords} we obtain that two local $2$-forms $\mathrm{d}\tilde{x}_\alpha$ are patched on overlaps $T^\ast(U_\alpha \cap U_\beta)$ by
\begin{equation}
    \mathrm{d}\tilde{x}_\alpha - \mathrm{d}\tilde{x}_\beta = -\mathrm{d}\Lambda_{\alpha\beta}.
\end{equation}
We can lift $B_\alpha\in\Omega^2(U_\alpha)$ to a $2$-form on $T^\ast U_\alpha$ trough the bundle projector $\bbpi:T^\ast U_\alpha\rightarrow U_\alpha$. Therefore, by calculating their difference on overlaps $T^\ast(U_\alpha\cap U_\beta)$, we have
\begin{equation}\label{2conn}
    (\mathrm{d}\tilde{x}_\alpha - B_\alpha) - (\mathrm{d}\tilde{x}_\beta - B_\beta) = -\mathrm{d}\Lambda_{\alpha\beta} -B_\alpha + B_\beta = 0,
\end{equation}
so that the $2$-form $\omega_B = \mathrm{d}\tilde{x}_\alpha - B_\alpha$ is patched trivially on overlaps and hence the conclusion.
\end{proof}

\begin{remark}[Gerbe curvature]
If we take the differential of this global $2$-form we get a globally defined $3$-form on the base manifold $M$ that can be identified with the gerbe curvature
\begin{equation}
    H = -\mathrm{d}\omega_B.
\end{equation}
This global $2$-form $\omega_B=\mathrm{d}\tilde{x}_\alpha - B_\alpha$ is therefore for the doubled space $\M$ the higher analogue of the global $1$-form connection $\omega_A=\mathrm{d}\theta_\alpha+A_\alpha$ of a circle bundle.
\end{remark}

\begin{digression}[Recovering para-Hermitian geometry]\label{digd}
Let us go back to para-Hermitian geometry to improve the argument of digression \ref{digb}. We can write the connection $2$-form $\omega_B = \mathrm{d}\tilde{x}_\alpha - B_\alpha$ of lemma \ref{thm:gerbeconn} on each patch $T^\ast U_\alpha$ in Darboux coordinates to obtain
\begin{equation}\label{eq:omegaB}
 \omega_B = (\mathrm{d}\tilde{x}_\mu + B_{\mu\nu}\mathrm{d}x^\nu)\wedge\mathrm{d}x^\mu
\end{equation}
with curvature $H=-\mathrm{d}\omega_B$. Therefore the \textit{fundamental }$2$\textit{-form} $\omega_B$ of the para-Hermitian framework (see \cite{MarSza18} and \cite{Svo18}) is nothing but the higher analogue of a circle bundle connection. The success of the para-Hermitian framework is then explained by $\omega_B\in\Omega_{\mathrm{glob}}^2(\mathcal{M})$ being globally defined on the doubled space.
We can naturally construct a natural para-complex structure $J_{B}$ and a natural $O(d,d)$-tensor $\eta_{B}(-,-):=\omega_{B}(J_{B}-,-)$ on the patch $T^\ast U_\alpha$ as it follows:
\begin{equation}
    \eta_B = (\mathrm{d}\tilde{x}_\mu + B_{\mu\nu}\mathrm{d}x^\nu)\odot\mathrm{d}x^\mu, \quad J_B = \frac{\partial}{\partial\tilde{x}_\mu}\otimes(\mathrm{d}\tilde{x}_\mu + B_{\mu\nu}\mathrm{d}x^\nu)+\frac{\partial}{\partial x^\mu}\otimes \mathrm{d}x^\mu.
\end{equation}
We can also interpret these structure as the $B_\alpha$-shifted versions of the local para-K\"{a}hler structure $(\omega_{0\alpha},\eta_{0\alpha},J_{0\alpha})$ from remark \ref{digb}.
However in the para-Hermitian approach to DFT the doubled space is a global manifold $N$ equipped with para-Hermitian structure $(\omega_B,\eta_B,J_B)$, on the contrary in our stack formalism this is true only locally. 
\end{digression}

\noindent The subgroup of local diffeomorphisms $\Diff(T^\ast U_\alpha)$ preserving the para-complex structure $J_{B \alpha}$ is given by couples of diffeomorphisms of the base $U_\alpha$ and diffeomorphisms of the fibre $T^\ast_x U_\alpha$ smoothly depending on $x\in U_\alpha$. Infinitesimally this is $GL(d)\times GL(d)$. The subgroup preserving the connection $\omega_{B}$ is linearly $Sp(2d,\mathbb{R})$, while the one preserving local metric $\eta_{B \alpha}$ is linearly $O(d,d)$. The subgroup of local diffeomorphisms $\Diff(T^\ast U_\alpha)$ preserving the whole almost para-Hermitian structure $(J_B,\omega_B,\eta_B)$ is just the group of local diffeomorphisms of the base $\mathrm{Diff}(U_\alpha)$, which is linearly $GL(d)$. Hence we have the usual identities
\begin{equation}\begin{aligned}
    O(d,d) \,\cap\, Sp(2d,\mathbb{R}) \,&=\, GL(d) \\
    Sp(2d,\mathbb{R}) \,\cap\, \big(GL(d)\times GL(d)\big) \,&=\, GL(d) \\
    \big(GL(d)\times GL(d)\big) \,\cap\, O(d,d) \,&=\, GL(d)
\end{aligned}\end{equation}

\noindent Let us now remark why our doubled space is more general than a para-Hermitian manifold.

\begin{digression}[The problem with ordinary Para-Hermitian geometry]
Let us remark why doubled space cannot globally be an (almost) para-Hermitian manifold. Let $(N,J_B,\omega_B)$ be a global almost para-Hermitian manifold such that $N\rightarrow M$ is a fibration, locally given by the projectors $\bbpi_\alpha:T^\ast U_\alpha\rightarrow U_\alpha$. An (almost) para-Hermitian manifold does not need to be a fibration on a spacetime in general, but it needs it if it wants to recover at least one geometric background. The transition functions of the fibration will be $\Lambda_{\alpha\beta}+\Lambda_{\beta\gamma}+\Lambda_{\gamma\alpha}=0$. Now assume that $\omega_B = (\mathrm{d}\tilde{x}_\mu + B_{\mu\nu}\mathrm{d}x^\nu)\wedge\mathrm{d}x^\mu$, where $B_\alpha$ is a collection of local $2$-forms, is globally defined. Hence they will have to satisfy $B_\beta-B_\alpha = \mathrm{d}\Lambda_{\alpha\beta}$. But this implies the gerbe is trivial with $G_{\alpha\beta\gamma}=0$. Secondly in para-Hermitian geometry the topology of the fibered manifold $N$ is not determined by the Dixmier-Douady class $[H]\in H^3(M,\mathbb{Z})$ of the Kalb-Ramond field, differently from Higher Kaluza-Klein geometry. In some sense $[H]$ is not an intrinsic property of the geometry of para-Hermitian manifold, but it is something which can only be introduced by hand. Hence we cannot use a bare manifold to describe the doubled space: we need the stack formalism of Higher Kaluza-Klein Theory, which reduces to an ordinary para-Hermitian manifold only in a local sense.
\end{digression}

\noindent Let us now conclude this section by looking at a very different topic. Indeed Higher Kaluza-Klein Theory is also able to explicitly link DFT with a distinct field of research: prequantum geometry for String Theory. See the following digression for a brief discussion.

\begin{digression}[Recovering global Higher Geometric Prequantization]\label{dig:preq}
Notice our Higher Kaluza-Klein is as closely related to Higher Geometric Prequantization as ordinary Kaluza-Klein is to ordinary geometric prequantization. The parallel transport of a section $(\theta_\alpha)\in\Gamma(M,P)$ along a vector flow $\ell(t,x)$ with $\ell(0,x)=x$ of some Hamiltonian vector field $X$ is given by
\begin{equation}\label{eq:circlepreq}
    \theta_\alpha\big(\ell(t,x)\big) \,=\, \exp 2\pi i\Bigg(\sum_{\ell_\alpha}\int_{\ell_\alpha} \!A_\alpha + \sum_{x_{\alpha\beta}}f_{\alpha\beta}(x_{\alpha\beta})\Bigg)\cdot\theta_\alpha(x)
\end{equation}
which is a global gauge transformation in $\Coo(M,U(1))$ at any $t\in\mathbb{R}$. Recall that the underlying vector space of an ordinary prequantization Hilbert space is $\Gamma(M,\,P\times_{U(1)}\mathbb{C})$, i.e. the space of sections of the associated bundle $P\times_{U(1)}\mathbb{C}$. Hence parallel transport \eqref{eq:circlepreq} can be immediately generalized to prequantum states $(\psi_{\alpha})\in\Gamma(M,\,P\times_{U(1)}\mathbb{C})$. Analogously in Higher Geometric Prequantization we can define a parallel transport of a section $(\tilde{x}_\alpha,\phi_{\alpha\beta})\in\Gamma(M,\mathcal{M})$ of the doubled space along a vector flow $\ell(\tau,x)$ with $\ell(0,x)=x$ of a Hamiltonian vector field $X$ by
\begin{equation}\label{eq:higherpreq}
    \ell(\tau,-)_\ast\big(\tilde{x}_\alpha,\phi_{\alpha\beta}\big) \,=\, \Bigg(\sum_{l_\alpha}\int_{l_\alpha} \! B_\alpha + \sum_{x_{\alpha\beta}}\Lambda_{\alpha\beta}(x_{\alpha\beta}), \; \sum_{l_\alpha}\int_{l_\alpha} \! \Lambda_{\alpha\beta} + \!\!\!\sum_{x_{\alpha\beta\gamma}}\!G_{\alpha\beta\gamma}(x_{\alpha\beta\gamma}) \Bigg)\otimes \big(\tilde{x}_\alpha,\phi_{\alpha\beta}\big)
\end{equation}
which is a global gauge transformation in $\mathbf{H}(M,\BU)$ at any $\tau\in\mathbb{R}$. In \cite{FSS13}, \cite{FSS16} and \cite{Sza19} it is explained that the underlying groupoid of a \textit{prequantization }$2$\textit{-Hilbert space} is $\Gamma(M,\,\mathcal{M}\times_{\mathbf{B}U(1)_{\mathrm{conn}}}\!\mathbf{B}U_{\mathrm{conn}})$, which is the groupoid of sections of the associated $2$-bundle $\mathcal{M}\times_{\mathbf{B}U(1)_{\mathrm{conn}}}\!\mathbf{B}U_{\mathrm{conn}}$ where the fiber stack is the direct limit $\mathbf{B}U_{\mathrm{conn}} := \Lim{ N \to \infty}\mathbf{B}U(N)_{\mathrm{conn}}$. Hence parallel transport \eqref{eq:higherpreq} can be immediately generalized to \textit{prequantum }$2$\textit{-states} of the form $(\psi_{\alpha},\psi_{\alpha\beta})\in\Gamma(M,\,\mathcal{M}\times_{\mathbf{B}U(1)_{\mathrm{conn}}}\!\mathbf{B}U_{\mathrm{conn}})$. These are principal $U(N)$-bundles on the base manifold $M$ for any $N\in\mathbb{N}^+$, twisted by the doubled space $\mathcal{M}\xrightarrow{\bbpi}M$ with transition functions $(\Lambda_{\alpha\beta},G_{\alpha\beta\gamma})$. They are given in \v{C}ech data by $\psi_\alpha\in\Omega^1\big(U_\alpha,\mathfrak{u}(N)\big)$ and $\psi_{\alpha\beta}\in\Coo\big(U_\alpha\cap U_\beta,U(N)\big)$ patched by
\begin{equation}
    \begin{aligned}
    \psi_\alpha - \psi_{\alpha\beta}^{-1}\big(\psi_\beta + \mathrm{d} \big)\psi_{\alpha\beta} \,&=\, - \Lambda_{\alpha\beta} \\
    \psi_{\alpha\beta}\cdot\psi_{\beta\gamma}\cdot\psi_{\gamma\alpha} \,&=\, \exp{i2\pi G_{\alpha\beta\gamma}}
    \end{aligned}
\end{equation}
and they can be interpreted as states of $N$ coincident D-branes in a Kalb-Ramond field background.
In \cite{Sza19} this was explicitly calculated, in an infinitesimal fashion, for $M=\mathbb{R}^d$ and its non-associative behaviour was pointed out.
As explained by \cite{Rog13} and \cite{FSS13} we recover an ordinary prequantization on the loop space $LM:=[S^1,M]$. Indeed, given any loop $S_0^1\subset M$, its evolution $S^1_\tau := \ell(\tau,\,S^1_0)$ with $\tau\in\mathbb{R}$ can be seen both as a surface $\Sigma\subset M$ with $\partial\Sigma=S^1_0\sqcup S^1_\tau$ or as a path in the loop space $LM$. By integrating along surface $\Sigma$ and taking the trace we get
\begin{equation*}
\begin{aligned}
    \mathrm{Hol}_{\ell(\tau,-)^\ast(\psi_{\alpha},\psi_{\alpha\beta})}(S^1_\tau)\,=\, \mathrm{Hol}_{(B_{\alpha},\Lambda_{\alpha\beta},G_{\alpha\beta\gamma})}(\Sigma)\,\cdot\, \mathrm{Hol}_{(\psi_{\alpha},\psi_{\alpha\beta})}(S^1_0) \\
\end{aligned}
\end{equation*}
where the holonomy of a twisted $U(N)$-bundle $(\psi_{\alpha},\psi_{\alpha\beta})$ along a loop $S^1\subset M$ and the holonomy of a gerbe $(B_{\alpha},\Lambda_{\alpha\beta},G_{\alpha\beta\gamma})$ along a surface $\Sigma\subset M$ are respectively given by the usual expressions
\begin{equation}\begin{aligned}
    \mathrm{Hol}_{(\psi_{\alpha},\psi_{\alpha\beta})}(S^1) \,&:=\, \mathrm{Tr}\,\mathcal{P}\bigg(\exp{ 2\pi i\bigg(\int_{l_\alpha}\!\psi_\alpha}\bigg)\cdot\prod_{x_{\alpha\beta}}\psi_{\alpha\beta}(x_{\alpha\beta})\bigg), \\
    \mathrm{Hol}_{(B_{\alpha},\Lambda_{\alpha\beta},G_{\alpha\beta\gamma})}(\Sigma) \,&:=\, \exp 2\pi i\Bigg(\!\sum_\alpha\int_{\Sigma_\alpha} \!\!\!B_\alpha + \sum_{l_{\alpha\beta}}\int_{l_{\alpha\beta}} \!\!\!\Lambda_{\alpha\beta} + \!\!\sum_{x_{\alpha\beta\gamma}}\!G_{\alpha\beta\gamma}(x_{\alpha\beta\gamma})\Bigg).
\end{aligned}\end{equation}
More generally, for an open surface $\Sigma\subset M$ such that the boundary $\partial\Sigma=\sqcup_iS_i^1$ is a disjoint union of loops of any orientation, we recover Bohr-Sommerfeld condition for a surface with boundary 
\begin{equation}
    \mathrm{Hol}_{(B_{\alpha},\Lambda_{\alpha\beta},G_{\alpha\beta\gamma})}(\Sigma) \,\cdot\, \prod_i \mathrm{Hol}_{(\psi_{\alpha},\psi_{\alpha\beta})}(S^1_i) \,=\,1
\end{equation}
Notice that this is nothing but the bosonic part of \textit{Freed-Witten anomaly} cancellation of a string.

\begin{figure}[!ht]\centering
\includegraphics[scale=0.35]{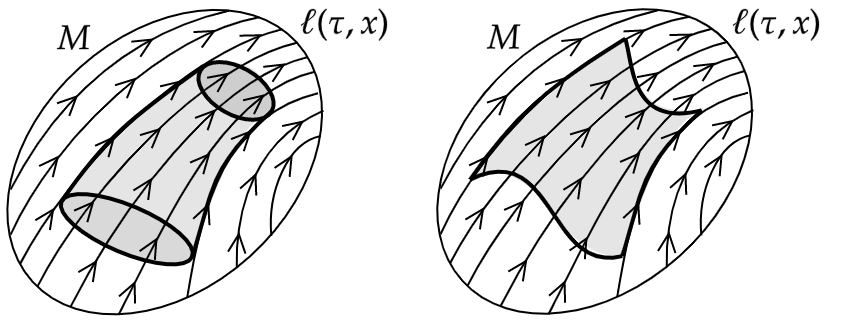}
\caption{Parallel transport along surfaces $\Sigma$.}\label{FIGstring}\end{figure}
\end{digression}

\begin{remark}[A prequantum interpretation of doubled coordinates]
It is suggestive to notice that in Higher Prequantization (see digression \ref{dig:preq}) a section of the doubled space $(\tilde{x}_\alpha,\phi_{\alpha\beta})$ on the world-volume of $N$ coincident D-branes plays the role of a "higher phase" for the $U(N)$-field $(\psi_\alpha,\psi_{\alpha\beta})$. This is analogous to the well-known fact that a section $(\theta_\alpha)$ of the prequantum circle bundle plays the role of the phase of a wave-function $(\psi_\alpha)$. This provides an evocative alternative interpretation of the extra coordinates of DFT in the context of prequantization.
\end{remark}

\begin{table}[ht!]\begin{center}\vspace{0.5cm}\begin{tabular}{ c | c c }
 & $\;\;\qquad$Prequantization$\qquad\quad$ & $\qquad\qquad$Higher Prequantization$\qquad\quad$ \\[0.8ex] \hline \\[-1.5ex]
Phase$\,$ & $\theta_\alpha -\theta_\beta = G_{\alpha\beta} \;\mathrm{mod}\,2\pi\mathbb{Z}$ & \makecell{$\tilde{x}_\alpha - \tilde{x}_\beta -\mathrm{d}\phi_{\alpha\beta}  \,=\, -\Lambda_{\alpha\beta}$, \\ $\phi_{\alpha\beta}+\phi_{\beta\gamma}+\phi_{\gamma\alpha} \,=\, G_{\alpha\beta\gamma} \quad\!\! \mathrm{mod}\,2\pi\mathbb{Z}$} \\[6.2ex]
Matter$\,$ & $\psi_\alpha\cdot\psi_\beta^{-1} = \exp i2\pi G_{\alpha\beta}$ & \makecell{$\psi_\alpha - \psi_{\alpha\beta}^{-1}\big(\psi_\beta + \mathrm{d} \big)\psi_{\alpha\beta} \,=\, - \Lambda_{\alpha\beta}$ \\ $\psi_{\alpha\beta}\cdot\psi_{\beta\gamma}\cdot\psi_{\gamma\alpha} \,=\, \exp{i2\pi G_{\alpha\beta\gamma}}$ }
\end{tabular}\end{center}\caption{A comparison of phases and states between ordinary and Higher Prequantization.}\end{table}

\subsection{Finite symmetries constitute the gauge 2-group}

In this subsection we will deal with finite symmetries of the doubled space and we will prove they are exactly the gauge transformations we expect for DFT.

\begin{remark}[$2$-group of gauge transformations of DFT]\label{subgroupoidgauge}
The automorphisms $2$-group of the principal structure (example \ref{ex:aut}) of the doubled space $\mathcal{M}\xrightarrow{\bbpi}M$ is exactly the $2$-group extending the diffeomorphisms of the base $M$ through gauge transformations of the gerbe
\begin{equation}\label{eq:DFTgroup}
    \begin{tikzcd}[row sep=scriptsize, column sep=6ex]
    \mathbf{Aut}_{/}(\Lambda,G) \,=\, \Diff(M)\ltimes\mathbf{H}(M,\BU).
    \end{tikzcd}
\end{equation}
This $2$-group is the stack refined version of the gauge group of DFT proposed by \cite{Hull14}, i.e.
\begin{equation}\label{eq:hull}
   G_{\mathrm{NS}} \,=\, \Diff(M)\ltimes\Omega^2_{\mathrm{cl}}(M),
\end{equation}
which is obtained by taking the curvature of the circle bundle, as explained in example \ref{ex:aut}.
\end{remark}

\noindent Notice \eqref{eq:DFTgroup} is exactly the analogous to the familiar $\Diff(M)\ltimes\Coo(M,U(1))\subset\Diff(P)$ of gauge transformations in ordinary Kaluza-Klein theory, where $P$ is a circle bundle.
\vspace{0.2cm}

\noindent The map between \eqref{eq:DFTgroup} and \eqref{eq:hull} is just the curvature $\mathrm{curv}:\mathbf{H}(M,\BU)\rightarrow\Omega^2_{\mathrm{cl}}(M)$, which maps $(\eta_\alpha,\eta_{\alpha\beta})\mapsto b$ where $b|_{U_\alpha}:=\mathrm{d}\eta_\alpha$ on each patch, in accord with remark \ref{rem:globgauge}. This global closed $2$-form $b\in\Omega^2_{\mathrm{cl}}(M)$ is usually called $B$-shift in DFT literature.

\begin{remark}[$2$-group of gauge transformations of DFT in \v{C}ech data]\label{diffgaugeconc}
The $2$-group of gauge transformations $\mathbf{Aut}_/(\Lambda,G)=\Diff(M)\ltimes\mathbf{H}(M,\BU)$ from remark \ref{subgroupoidgauge} will naturally define an action on the groupoid $\Gamma(M,\mathcal{M})$ of sections of the doubled space by the functor
\begin{equation}
    \mathbf{Aut}_/(\Lambda,G)\,\times\, \Gamma(M,\mathcal{M}) \;\longrightarrow\; \Gamma(M,\mathcal{M}).
\end{equation}
In local \v{C}ech data on the base manifold $M$ this action will be given by the following.
\begin{itemize}
\item objects of $\mathbf{Aut}_/(\Lambda,G)$ are triples $(f,\eta_\alpha,\eta_{\alpha\beta})$, consisting of a diffeomorphism $f\in\Diff(M)$ and a circle bundle $(\eta_\alpha,\eta_{\alpha\beta})\in\mathbf{H}(M,\BU)$ and acting on sections $\Gamma(M,\mathcal{M})$ by
\begin{equation}
\begin{aligned}
    (f,\eta_\alpha,\,\eta_{\alpha\beta})\, :\, (\tilde{x}_\alpha,\phi_{\alpha\beta}) \,&\mapsto\, \big(f_\ast\tilde{x}_\alpha+\eta_\alpha, \,\phi_{\alpha\beta}+\eta_{\alpha\beta}\big) \\
\end{aligned}
\end{equation}
\item isomorphisms of $\mathbf{Aut}_/(\Lambda,G)$ between these objects are just ordinary gauge transformations of circle bundles, consisting of local functions $\epsilon_\alpha\in\Coo(U_\alpha)$ and acting by 
\begin{equation}
   (\epsilon_\alpha):\, (f,\eta_\alpha,\,\eta_{\alpha\beta}) \,\Mapsto\, (f,\eta_\alpha+\mathrm{d}\epsilon_\alpha,\,\eta_{\alpha\beta}+\epsilon_\alpha-\epsilon_\beta).
\end{equation}
\end{itemize}
In terms of diagrams we can rewrite the $2$-group of automorphisms of the doubled space $\M$ as
\begin{equation}
    \mathbf{Aut}_/(\Lambda,G) \,\simeq\, \left\{ \begin{tikzcd}[row sep=scriptsize, column sep=18ex]
    \M \arrow[r, bend left=50, ""{name=U, below}, "(f\text{,}\eta_\alpha\text{,}\,\eta_{\alpha\beta})"]
    \arrow[r, bend right=50, "(f\text{,}\eta_\alpha+\mathrm{d}\epsilon_\alpha\text{,}\,\eta_{\alpha\beta}+\epsilon_\alpha-\epsilon_\beta)"', ""{name=D}]
    & \M
    \arrow[Rightarrow, from=U, to=D, "\,(\epsilon_\alpha)"]
\end{tikzcd} \right\}
\end{equation}
Hence the action of the sub-$2$-group $\mathbf{H}(M,\BU)\subset\mathbf{Aut}_/(\Lambda,G)$ of automorphisms which cover the identity $\mathrm{id}_M\in\mathrm{Diff}(M)$ is exactly the principal action of the doubled space $\mathcal{M}$ from remark \ref{principalaction}. This is directly analogous to Kaluza-Klein Theory, where the translation along the compactified dimension coincides with the principal circle action.
\end{remark}

\begin{digression}[Generalized diffeomorphisms]
In the DFT literature the automorphisms $\mathbf{Aut}_/(\Lambda,G)$ from remark \ref{diffgaugeconc} are usually called \textit{generalized diffeomorphisms} (or \textit{large gauge transformations}). Notice that the research by \cite{DesSae18} and by \cite{Hohm19DFT} are independently already pointing in the direction of generalized diffeomorphisms having a higher group structure. 
\end{digression}

\noindent Now we can explain how automorphisms of definition \ref{subgroupoidgauge} can be used to glue the doubled space in a way that is not affected by Papadopolous' puzzle \eqref{eq:papa}.

\begin{remark}[Doubled space must be glued in a $(2,1)$-category]
Gluing local patches of the doubled space has always been a puzzle in DFT (which becomes even more problematic in Exceptional Field Theory). In our proposal the solution to this puzzle is given by the fact that our doubled space $\mathcal{M}$ is not a manifold, but a stack: therefore it is glued not in the category of smooth manifolds, but in the $(2,1)$-category of stacks.
Let us call $\mathcal{M}_\alpha:= \mathcal{M}|_{U_\alpha}$. These are trivial principal $\BU$-bundles on each patch $U_\alpha$ and their groupoid of sections $\Gamma(U_\alpha,\mathcal{M}_\alpha)=\mathbf{H}(U_\alpha,\BU)$ are just the groupoids of local circle bundles $P_\alpha$ equipped with connection $\tilde{x}_\alpha\in\Omega^1(U_\alpha)$. Hence we can glue together local sections of $\mathcal{M}$ as follows:
\begin{equation}
    \Gamma(M,\mathcal{M}) \,=\, \left\{
    \begin{tikzcd}[row sep=7ex, column sep=0ex]
    & P_\beta|_{U_\alpha\cap U_\beta\cap U_\gamma}\arrow[rd, "(\Lambda_{\beta\gamma}\text{,}\phi_{\beta\gamma})"] \arrow[Rightarrow, from=1-2, to=D, "\,G_{\alpha\beta\gamma}"] & \\
    P_\alpha|_{U_\alpha\cap U_\beta\cap U_\gamma}\arrow[rr, ""{name=D}, "(\Lambda_{\alpha\gamma}\text{,}\phi_{\alpha\gamma})"']{} \arrow[ru, "(\Lambda_{\alpha\beta}\text{,}\phi_{\alpha\beta})"] & & P_\gamma|_{U_\alpha\cap U_\beta\cap U_\gamma}
    \end{tikzcd}
    \right\}
\end{equation}
where a $1$-morphism $(\Lambda_{\alpha\beta},\phi_{\alpha\beta}):P_\alpha|_{U_\alpha\cap U_\beta} \rightarrow P_\beta|_{U_\alpha\cap U_\beta}$ is given not only by any gauge transformation $\phi_{\alpha\beta}$, but also by a shift $\Lambda_{\alpha\beta}$ in the connection, so that  $\tilde{x}_\alpha\mapsto\tilde{x}_\beta-\Lambda_{\alpha\beta}+\mathrm{d}\phi_{\alpha\beta}$.
The existence of the $2$-morphism $G_{\alpha\beta\gamma}$ implies that $1$-morphisms must satisfy the patching conditions $\phi_{\alpha\beta}+\phi_{\beta\gamma}+\phi_{\gamma\alpha}=G_{\alpha\beta\gamma}$ and $\Lambda_{\alpha\beta}+\Lambda_{\beta\gamma}+\Lambda_{\gamma\alpha}=\mathrm{d}G_{\alpha\beta\gamma}$ on three-fold overlaps. Now we can look at the sub-spaces $\mathcal{M}_\alpha$, $\mathcal{M}_\beta$ and $\mathcal{M}_\gamma$. They are then patched by automorphisms $e_{\alpha\beta}\in\mathrm{Diff}(U_\alpha\cap U_\beta)\ltimes\mathbf{H}(U_\alpha\cap U_\beta,\,\BU)$ on two-fold overlaps $\mathcal{M}_\alpha|_{U_\alpha\cap U_\beta}\simeq\mathcal{M}_\beta|_{U_\alpha\cap U_\beta}$. On three-fold overlaps we have gauge transformations of automorphisms (see remark \ref{diffgaugeconc}) which are given by $e_{\alpha\beta}\circ e_{\beta\gamma}\xRightarrow{G_{\alpha\beta\gamma}\;} e_{\alpha\gamma}$, or equivalently by the $2$-commuting diagram
\begin{equation}
    \begin{tikzcd}[row sep=7ex, column sep=0ex]
    & \mathcal{M}_\beta|_{U_\alpha\cap U_\beta\cap U_\gamma}\arrow[rd, "\text{}e_{\beta\gamma}"] \arrow[Rightarrow, from=1-2, to=D,  "\,G_{\alpha\beta\gamma} "] & \\
    \mathcal{M}_\alpha|_{U_\alpha\cap U_\beta\cap U_\gamma}\arrow[rr, ""{name=D}, "\text{}e_{\alpha\gamma}"']{} \arrow[ru, "\text{}e_{\alpha\beta}"] & & \mathcal{M}_\gamma|_{U_\alpha\cap U_\beta\cap U_\gamma}
    \end{tikzcd}
\end{equation}
Therefore we can glue the local $\{\mathcal{M}_\alpha\}$ by automorphisms to get the global doubled space $\mathcal{M}$.
\end{remark}

\noindent This idea of higher gluing the doubled space in a $(2,1)$-category, contrary to appearances, is not totally unprecedented. Let us mention some of its relevant progenitors in the literature.

\begin{digression}[Precedents of higher gluing]
Notice that \cite{BCM14} proposed for the first time non-trivial patching conditions on three-fold overlaps of patches of the doubled space. More recently the local differential-graded patches $T^\ast[2]T[1]U_\alpha$ with $U_\alpha\subset M$ proposed by \cite{DesSae18} would need to be glued together in the $2$-category of derived spaces to give a global picture. This is because differential-graded manifolds are exactly a simple model for derived spaces (see \cite{dman}). But this is consistent with our formalism: we will see in the next subsection that the formalism by \cite{DesSae18} can be indeed seen as an infinitesimal geometry of our doubled space.
\end{digression}

\subsection{Infinitesimal symmetries constitute the Courant 2-algebra}

In this subsection we will deal with infinitesimal symmetries of the doubled space and we will prove they locally reduce to the one expected from DFT. Indeed we will show that the Courant $2$-algebroid formalism can be recovered as infinitesimal descriptions of the geometry of the doubled space. Finally, from the Courant $2$-algebroid, we will explicitly recover ordinary Generalized Geometry. Recall that Generalized Geometry has been revealed by \cite{Wald08}, \cite{Wald11} and \cite{Wald12} to be the natural language to express Type II supergravity.

\begin{definition}[Infinitesimally thickened point]
An \textit{infinitesimally thickened point} is defined (see \cite{DCCTv2}) as the locally ringed space given by the spectrum of the ring of dual numbers
\begin{equation}
    \mathbb{D}^1 \;:=\; \mathrm{Spec}\left(\frac{\mathbb{R}[\epsilon]}{\langle\epsilon^2\rangle}\right).
\end{equation}
Hence its underlying topological space is just a single point $\{\ast\}$, but its smooth algebra of functions is $\mathcal{O}_{\mathbb{D}^1}(\{\ast\}) = \mathbb{R}[\epsilon]/\langle\epsilon^2\rangle \cong \mathbb{R}\oplus\epsilon\mathbb{R}$ with $\epsilon^2=0$, i.e. it is the ring of dual numbers.
\end{definition}

\begin{example}[Tangent bundle of a manifold]
This idea is widely used in algebraic geometry to define tangent spaces. For example a map $\mathbb{D}^1\xrightarrow{X} M$ is given on the underlying topological spaces by sending $\{\ast\}$ to a point $x$ in a manifold $M$ and on the algebras of smooth functions by a map $\Coo(M)\xrightarrow{X^\sharp}\mathbb{R}\oplus\epsilon\mathbb{R}$, which is given by $f \mapsto f(x)+\epsilon X^\mu\partial_\mu f(x)$ for some $X\in T_xM$. Thus vectors $X$ on $M$ can equivalently be seen as maps $\mathbb{D}^1\xrightarrow{X} M$ and therefore $TM\simeq[\mathbb{D}^1,\,M]$.
\end{example}

\noindent This motivates the following definition for the tangent stack of the doubled space $\M$.

\begin{definition}[Doubled tangent bundle]\label{def:doubledtangentbundle}
We define the tangent bundle of a doubled space $\mathcal{M}$ (defined in postulate \ref{post1}) as the internal hom stack (definition \ref{def:inthom}) of maps from $\mathbb{D}^1$ to $\mathcal{M}$
\begin{equation}
    T\mathcal{M} \;:=\; \big[\mathbb{D}^1,\, \mathcal{M}\big]
\end{equation}
\end{definition}

\begin{theorem}[Doubled tangent bundle in local data]\label{thm:dtbild}
On a patch $U\subset M$ of the base manifold
\begin{equation}\label{eq:DTBs}
    \Gamma(U,\; T\mathcal{M}) \;\simeq\; \Gamma(U,\,\mathcal{M})\times\Big(\mathfrak{X}(U)\ltimes\mathbf{H}(U,\mathbf{B}\mathbb{R}_{\mathrm{conn}})\Big)
\end{equation}
These sections will be non-trivially patched on the whole smooth base manifold $M$ as it follows:
\begin{itemize}
    \item local sections $(\tilde{x}_\alpha,\phi_{\alpha\beta})\in\Gamma(U,\,\mathcal{M})$ will be patched as usual by \eqref{eq:coords},
    \item local vectors $X\in\mathfrak{X}(U)$ will be patched to a global vector on $M$,
    \item local line bundles $(\xi_\alpha,\eta_{\alpha\beta})\in\mathbf{H}(U,\mathbf{B}\mathbb{R}_{\mathrm{conn}})$ will be patched by the Lie derivatives of the transition functions of the doubled space $(\mathcal{L}_{X}\Lambda_{\alpha\beta},\,\mathcal{L}_{X}G_{\alpha\beta\gamma})$, i.e. by
    \begin{equation}\label{eq:localsymdata}
        \begin{aligned}
        \xi_\alpha - \xi_\beta \;&=\; -\mathcal{L}_{X}\Lambda_{\alpha\beta} +\mathrm{d}\eta_{\alpha\beta}, \\
       \eta_{\alpha\beta}+ \eta_{\beta\gamma} + \eta_{\gamma\alpha} \;&=\; \mathcal{L}_{X}G_{\alpha\beta\gamma}.
        \end{aligned}
    \end{equation}
\end{itemize}
Notice that without loss of generality we can reparametrize the scalars $f_{\alpha\beta} :=  \eta_{\alpha\beta} - \iota_X\Lambda_{\alpha\beta}$ in \eqref{eq:localsymdata} and obtain sections $\mathbbvar{X}:=(X+\xi_\alpha,f_{\alpha\beta})$, which are now patched by the familiar condition
    \begin{equation}
        \begin{aligned}
        \xi_\alpha - \xi_\beta \;&=\; -\iota_{X}\mathrm{d}\Lambda_{\alpha\beta} +\mathrm{d}f_{\alpha\beta}, \\
        f_{\alpha\beta}+f_{\beta\gamma} + f_{\gamma\alpha} \;&=\; 0,
        \end{aligned}
    \end{equation}
\end{theorem}

\begin{proof}
Since local sections of $\mathcal{M}$ are local circle bundles $\Gamma(U,\mathcal{M})\cong\mathbf{H}(U,\BU)$, we can now find the groupoid of sections of the stack $T\mathcal{M}|_U \simeq \left[\mathbb{D}^1,\,U\times\BU\right]$ by calculating 
\begin{equation}
    \begin{aligned}
    \mathbf{H}(U\times\mathbb{D}^1, U) \,&\simeq\, \mathrm{Diff}(U)\times\mathfrak{X}(U), \\
    \mathbf{H}(U\times\mathbb{D}^1, \BU) \,&\simeq\, \mathbf{H}(U, \BU) \times \mathbf{H}(U, \mathbf{B}\mathbb{R}_{\mathrm{conn}})
    \end{aligned}
\end{equation}
and then by considering only the subgroupoid of maps covering the identity $\mathrm{id}_U\in\mathrm{Diff}(U)$. Hence a local section in $\Gamma(U,\, T\mathcal{M}) \simeq \Gamma(U\times\mathbb{D}^1,\, \mathcal{M})$ must be given by \eqref{eq:DTBs}. The transition functions of the doubled tangent bundle $T\mathcal{M}\rightarrow\mathcal{M}$ are actually functions on the base manifold $M$, because the transition functions $(\Lambda_{\alpha\beta},G_{\alpha\beta\gamma})$ of $\mathcal{M}$ are functions on $M$. Since the doubled space is patched by bundle automorphisms $\mathrm{Diff}(U)\ltimes\mathbf{H}(U,\BU)$, its doubled tangent bundle will be patched by their infinitesimal version $\Coo\big(U,\,GL(d)\big)\ltimes\mathbf{H}(U,\,\mathbf{B}\mathbb{R}_{\mathrm{conn}})$. Hence $T\mathcal{M}$ naturally induces a bundle with transition functions $M\rightarrow\mathbf{B}(GL(d)\ltimes\mathbf{B}\mathbb{R}_{\mathrm{conn}})$. The transition functions $M\rightarrow\mathbf{B}GL(d)$ are just the ones of $TM$, while the map $M\rightarrow\mathbf{B}\mathbb{R}_{\mathrm{conn}}$ is just the collection $(\mathcal{L}_{X}\Lambda_{\alpha\beta},\,\mathcal{L}_{X}G_{\alpha\beta\gamma})$. Therefore sections are patched by condition \eqref{eq:localsymdata}.
\end{proof}

\noindent Notice \eqref{eq:DTBs} is analogous to the familiar idea that the tangent bundle $TP$ of a circle bundle $P$ is locally of the form $TU\times U(1)\times\mathbb{R}$ and patched by using the transition functions of $P$.

\vspace{0.2cm}

\noindent Hence the \v{C}ech data of a doubled vector $\mathbbvar{X}=(X+\xi_\alpha,f_{\alpha\beta})$ are the data of an infinitesimal gauge transformation of the form $x\mapsto x+\epsilon X$ and $(\tilde{x}_\alpha,\,\phi_{\alpha\beta})\mapsto(\tilde{x}_\alpha+\epsilon\xi_\alpha,\,\phi_{\alpha\beta}+\epsilon f_{\alpha\beta}+\epsilon\iota_X\Lambda_{\alpha\beta})$. 

\vspace{0.2cm}
\noindent Indeed the only doubled vectors we are actually interested in are the ones which are invariant under the principal action of the doubled space $\mathcal{M}$, i.e. the doubled vectors which satisfies the strong constraint. These are also the ones called \textit{strongly foliated} by \cite{Vais12}.
The following argument can be seen as formalization Papadopoulos' one in \cite{Pap14}.

\begin{definition}[Courant $2$-algebroid]\label{def:courantalg}
The $2$-algebra of sections of the \textit{Courant }$2$\textit{-algebroid} over the base manifold $M$ sits in the center of the following short exact sequence of $2$-algebras:
\begin{equation}\label{eq:courasection}
\begin{tikzcd}[row sep=scriptsize, column sep=8ex]
   0 \arrow[r] & \mathbf{H}\big(M,\mathbf{B}\mathbb{R}_{\mathrm{conn}}\big)\; \arrow[r, hook, "\mathrm{injection}"] & \;\mathfrak{at}(\mathcal{M})\; \arrow[r, two heads, "\mathrm{anchor}"] & \; \mathfrak{X}(M) \arrow[r] & 0,
\end{tikzcd}
\end{equation}
where $\mathfrak{X}(M)$ is the algebra of vector fields on $M$ and $\mathbf{H}\big(M,\mathbf{B}\mathbb{R}_{\mathrm{conn}}\big)$ is the abelian $2$-algebra of line bundles with connection on $M$, i.e. of infinitesimal gauge transformations of the gerbe. This is obtained by differentiating the finite automorphisms sequence \eqref{eq:finiteautomorphismseq}.
\end{definition}

\begin{remark}[Analogy with Atiyah $1$-algebroid]
Definition \ref{def:courantalg} is analogous to Atiyah algebroid in Kaluza-Klein, which encodes vectors on a circle bundle $P$ invariant under principal action. As explained in \cite{Col11} and \cite{Rog13} exact sequence \eqref{courantseq} is the higher version of the ordinary Atiyah sequence
\begin{equation}
\begin{tikzcd}[row sep=scriptsize, column sep=9ex]
   0 \arrow[r] & \Coo(M,\mathbb{R})\; \arrow[r, hook,"\mathrm{injection}"] & \;\mathfrak{at}(P)\; \arrow[r, two heads, "\mathrm{anchor}"] & \; \mathfrak{X}(M) \arrow[r] & 0.
\end{tikzcd}
\end{equation}
\end{remark}

\begin{definition}[Standard Courant $2$-algebroid]
We define the \textit{standard Courant }$2$\textit{-algebroid} by the semidirect sum of $2$-algebras $\mathfrak{X}(M)\oplus\mathbf{H}\big(M,\mathbf{B}\mathbb{R}_{\mathrm{conn}}\big)$. This means its sections will be of the form $(X+\xi_\alpha, f_{\alpha\beta})$ where $X\in\mathfrak{X}(M)$ is a global vector and $(\xi_\alpha, f_{\alpha\beta})$ is a \v{C}ech cocycle
\begin{equation}
\begin{aligned}
    \begin{aligned}
        \xi_\alpha - \xi_\beta \;&=\; \mathrm{d}f_{\alpha\beta}, \\
        f_{\alpha\beta}+f_{\beta\gamma}+f_{\gamma\alpha} \;&=\; 0.
    \end{aligned}
\end{aligned}
\end{equation}
The morphisms will be gauge transformations $(\varepsilon_\alpha):(X+\xi_\alpha, f_{\alpha\beta})\mapsto(X+\xi_\alpha+\mathrm{d}\varepsilon_\alpha,\, f_{\alpha\beta}+\varepsilon_\alpha-\varepsilon_\beta)$ between line bundles, as usual. By slightly extending \cite{Col11}, the $2$-algebra structure of this semidirect sum is isomorphic to the bracket structure given by the following bracket structure
\begin{gather}\label{eq:brac}
\begin{aligned}
    \big\llbracket (\varepsilon_\alpha) \big\rrbracket_{\mathrm{std}} &= (\mathrm{d}\varepsilon_\alpha, \, \varepsilon_\alpha-\varepsilon_\beta) \\
    \big\llbracket (X+\xi_\alpha, f_{\alpha\beta}), (Y+\eta_\alpha, g_{\alpha\beta})\big\rrbracket_{\mathrm{std}}  &= \bigg([X,Y]+\mathcal{L}_X\eta_\alpha-\mathcal{L}_Y\xi_\alpha - \frac{1}{2}\mathrm{d}(\iota_X\eta_\alpha-\iota_Y\xi_\alpha), \\
    & \hspace{4.0cm} \frac{1}{2}X(g_{\alpha\beta})-\frac{1}{2}Y(f_{\alpha\beta})\bigg) \\
    \big\llbracket (X+\xi_\alpha, \, f_{\alpha\beta}), (\varepsilon_\alpha)\big\rrbracket_{\mathrm{std}}  &= (\mathcal{L}_X\varepsilon_\alpha) \\
    \big\llbracket (X+\xi_\alpha, f_{\alpha\beta}), (Y+\eta_\alpha, g_{\alpha\beta}), (Z+\zeta_\alpha, h_{\alpha\beta})\big\rrbracket_{\mathrm{std}}  &= \frac{1}{3!}\bigg( \iota_X\iota_Y\mathrm{d}\zeta_\alpha + \frac{3}{2}\iota_X\mathrm{d}\iota_Y\zeta_\alpha + \text{perm.}\bigg)
\end{aligned}\raisetag{4.1cm}
\end{gather}
Let us notice that the underlying groupoid of sections $\mathfrak{X}(M)\oplus\mathbf{H}\big(M,\mathbf{B}\mathbb{R}_{\mathrm{conn}}\big)$ of the standard Courant $2$-algebroid is nothing but the stackification (in other words the globalization) of the familiar local $\Coo(U,\mathbb{R})\,\xrightarrow{\;\mathrm{d}\;}\,\mathfrak{X}(U)\oplus\Omega^1(U)$ on patches $U\subset M$ of the base manifold.
\end{definition}

\begin{theorem}[General Courant $2$-algebroid]\label{gengeom}
Given the higher Atiyah sequence \eqref{eq:courasection}, for any choice of splitting homomorphism
\begin{equation}\label{courantseq}
\begin{tikzcd}[row sep=scriptsize, column sep=8ex]
   0 \arrow[r] &\mathbf{H}\big(M,\mathbf{B}\mathbb{R}_{\mathrm{conn}}\big)\; \arrow[r, hook, "\mathrm{injection}"] & \;\mathfrak{at}(\mathcal{M})\; \arrow[r, two heads, "\mathrm{anchor}"] & \; \mathfrak{X}(M) \arrow[l, bend right=50, hook', "\mathrm{splitting}"'] \arrow[r] & 0.
\end{tikzcd}
\end{equation}
we get the following results.
\begin{itemize}
\item A section $\mathbbvar{X}:=(X+\xi_\alpha, \, f_{\alpha\beta})$ consists a global vector field $X\in\mathfrak{X}(M)$, a collection of $1$-forms $\xi_\alpha\in\Omega^1(U_\alpha)$ on each patch $U_\alpha$ of $M$ and a collection of functions $f_{\alpha\beta}\in\Coo(U_\alpha\cap U_\beta)$ on each overlap $U_\alpha\cap U_\beta$ of $M$,
such that they are glued according to
\begin{equation}\label{eq:pathcthevector}
\begin{aligned}
    \begin{aligned}
        \xi_\alpha - \xi_\beta \;&=\; -\iota_{X}\mathrm{d}\Lambda_{\alpha\beta}+\mathrm{d}f_{\alpha\beta}, \\
        f_{\alpha\beta}+f_{\beta\gamma}+f_{\gamma\alpha} \;&=\; 0.
    \end{aligned}
\end{aligned}
\end{equation}
\item A morphism between two sections is a gauge transformation $(\varepsilon_\alpha)$ of line bundles, given in local data by a collection of local functions $\varepsilon_\alpha\in\Coo(U_\alpha)$ so that
\begin{equation}
(\varepsilon_\alpha):\,(X+\xi_\alpha,\, f_{\alpha\beta}) \;\;\mapsto\;\; (X + \xi_\alpha+\mathrm{d}\varepsilon_{\alpha},\, f_{\alpha\beta}+\varepsilon_{\alpha}-\varepsilon_{\beta})
\end{equation}
\item The bracket structure on $\mathfrak{at}(\mathcal{M})$ is given by the one of the standard Courant algebroid by
\begin{equation}\label{eq:bracketfromstd}
\big\llbracket (s(X)+\xi_\alpha,\,f_{\alpha\beta}),\,(s(Y)+\eta_\alpha,\,f_{\alpha\beta})\big\rrbracket \,:=\, \big\llbracket (X+\xi_\alpha,\,f_{\alpha\beta}),\,(Y+\eta_\alpha,\,f_{\alpha\beta})\big\rrbracket_{\mathrm{std}}
\end{equation}
and analogously for the other brackets by using the remaining \eqref{eq:brac}.
\end{itemize}
\end{theorem}

\begin{proof}
By using the splitting $s$ and the injection $i$ in \eqref{courantseq} we can construct an isomorphism of $2$-algebras $s\oplus i:\,\mathfrak{X}(M)\oplus\mathbf{H}\big(M,\mathbf{B}\mathbb{R}_{\mathrm{conn}}\big)\xrightarrow{\;\cong\;}\mathfrak{at}(\mathcal{M})$. As explained in \cite{Rog13}, we have a splitting for any connection local data $B_\alpha\in\Omega^2(U_\alpha)$ which satisfiy $B_\beta - B_\alpha = \mathrm{d}\Lambda_{\alpha\beta}$.
The isomorphism $s\oplus i$ is then given by the map $(X+\xi'_\alpha,\,f_{\alpha\beta}) \mapsto (X+ \iota_XB_\alpha +\xi'_\alpha,\, f_{\alpha\beta})$ for objects and by the identity for gauge transformations. Recall that $(\xi'_\alpha,f_{\alpha\beta})$ is patched by $\xi'_\alpha-\xi'_\beta = \mathrm{d}f_{\alpha\beta}$ and $f_{\alpha\beta}+f_{\beta\gamma}+f_{\gamma\alpha}=0$. Now we only have to redefine $\xi_\alpha := \iota_XB_\alpha+\xi'_\alpha$ to get the wanted patching conditions.
\end{proof}

\begin{digression}[Twisted and untwisted generalized vectors]
The isomorphism of $2$-algebras
\begin{equation}\label{eq:isotw}
    \underbrace{\mathfrak{X}(M)\oplus\mathbf{H}\big(M,\mathbf{B}\mathbb{R}_{\mathrm{conn}}\big)}_\text{untwisted gen. \!vectors}\;\;\cong\underbrace{\mathfrak{at}(\mathcal{M})}_\text{twisted gen. \!vectors}
\end{equation}
we used in lemma \ref{gengeom} is exactly the isomorphism locally presented by \cite{Hull14} between \textit{twisted} and \textit{untwisted} generalized vectors, but globally defined. Indeed in the reference, on a patch $U_\alpha\subset M$, if $X+\xi'_\alpha$ is an untwisted generalized vector, then $X+\iota_XB_\alpha+\xi'_\alpha$ is its twisted form. This also gives to this notion a precise geometrical interpretation: the connection $B_\alpha$ splits the Courant $2$-algebroid in a \textit{horizontal bundle} $\mathfrak{X}(M)$ and \textit{vertical bundle} $\mathbf{H}(M,\mathbf{B}\mathbb{R}_{\mathrm{conn}})$.
\end{digression}

\noindent This is analogous to how the tangent bundle of an ordinary circle bundle $P\rightarrow M$ is split in horizontal and vertical bundle $TP \cong TM\oplus\mathbb{R}$ by a connection $A_\alpha$.

\begin{digression}[Relation with para-complex $\pm 1$-eigenbundles]\label{dig:paraherbundle}
The tangent bundle of a para-complex manifold $(N,J)$ naturally separates in $TN = L_+ \oplus L_-$, where the subbundles $L_{\pm}$ are the $\pm 1$-eigenbundles of the para-complex structure $J$. Now we will show that these eigenbundles give a wrong globalization of the tangent stack $T\mathcal{M}$ of lemma \ref{thm:dtbild}. In the case of a fiber bundle $N\xrightarrow{}M$, as explained in \cite{MarSza19}, we get that its tangent bundle naturally separates as
\begin{equation}\label{eq:parahermtangbund}
    TN \; = \; TM \,\oplus\, \bigsqcup_\alpha T^\ast U_\alpha /\sim
\end{equation}
where the equivalence $\sim$ is given on sections by patching $\xi_\alpha - \xi_\beta = - \iota_X\mathrm{d}\Lambda_{\alpha\beta}$. Notice that this patching condition is reminiscent of \eqref{eq:pathcthevector}. We then have an isomorphism $TM \oplus T^\ast M \cong TN$ for any local $2$-form $B_\alpha\in\Omega^2(U_\alpha)$ which satisfies $B_\beta-B_\alpha=\mathrm{d}\Lambda_{\alpha\beta}$. Such isomorphism is given by mapping $X+\xi\in\Gamma(M,TM \oplus T^\ast M)$ to $(X+\iota_XB_\alpha+\xi)\in\Gamma(M,TN)$ and this is reminiscent of lemma \ref{gengeom}. However, since $N$ is a manifold, transition functions will satisfy $\Lambda_{\alpha\beta}+\Lambda_{\beta\gamma}+\Lambda_{\gamma\alpha}=0$. This means that $B_\alpha$ is not a full gerbe connection and $TN$ cannot actually recover a non-trivial ordinary Courant algebroid appearing in Generalized Geometry (see \cite{Rog13}). On the other hand Higher Kaluza-Klein framework can clarify the geometrical meaning of these structures and give them a working globalization.
\end{digression}

\begin{remark}[Twisted bracket]
By explicitly writing \eqref{eq:bracketfromstd} we recover the \textit{twisted bracket} 
\begin{equation*}
\begin{aligned}
    \big\llbracket (X+\xi_\alpha, f_{\alpha\beta}), (Y+\eta_\alpha, g_{\alpha\beta})\big\rrbracket &= \bigg([X,Y]+\mathcal{L}_X\eta_\alpha-\mathcal{L}_Y\xi_\alpha - \frac{1}{2}\mathrm{d}(\iota_X\eta_\alpha-\iota_Y\xi_\alpha) + \iota_X\iota_YH, \\
    & \hspace{5.5cm} \frac{1}{2}X(g_{\alpha\beta})-\frac{1}{2}Y(f_{\alpha\beta})\bigg)
\end{aligned}
\end{equation*}
\end{remark}

\begin{remark}[Pushforward of automorphisms as local $O(d,d)$ transformations]
We will see now see how the Courant $2$-algebroid transforms under automorphisms of the doubled space. Given an automorphism $(\varphi,\eta_{\alpha},\eta_{\alpha\beta})\in\mathrm{Diff}(M)\ltimes\mathbf{H}(M,\BU)$, any section $(X+\xi_\alpha,f_{\alpha\beta})\in\mathfrak{at}(\mathcal{M})$ will transform under its pushforward, which is its infinitesimal version, by
\begin{equation}
    \big(X+\xi_\alpha,\;f_{\alpha\beta} \big) \;\,\mapsto\;\, \big(\varphi_\ast X+(\varphi^\ast)^{-1}\xi_\alpha +\iota_X\mathrm{d}\eta_\alpha,\;f_{\alpha\beta}\circ\varphi\big).
\end{equation}
Therefore on each local patch $U_\alpha\subset M$ we recover the following transformations
\begin{equation}
    \begin{pmatrix}X\\\xi_\alpha\end{pmatrix} \;\,\mapsto\;\, \begin{pmatrix}\varphi_\ast & 0 \\\mathrm{d}\eta_\alpha & (\varphi^\ast)^{-1} \end{pmatrix}\begin{pmatrix}X\\\xi_\alpha\end{pmatrix}
\end{equation}
Since $\varphi_\ast\in\Coo(M,GL(d))$ and $\mathrm{d}\eta_\alpha\in\Omega_{\mathrm{cl}}^2(M)$, we can see this as a local $GL(d)\ltimes\wedge^2\mathbb{R}^d\subset O(d,d)$ transformation. Hence we recover the local geometric $O(d,d)$ transformations of DFT.
\end{remark}

\begin{remark}[Recovering ordinary Generalized Geometry]\label{rem:horizontalsec}
The $2$-algebra $\mathfrak{at}(\mathcal{M})$ immediately reduces to the one appearing in \cite{DesSae18} and \cite{Col11} for sections of the form $(X+\xi_\alpha):=(X+\xi_\alpha,0)$ with $f_{\alpha\beta}=0$. Therefore these sections satisfy the patching condition $\xi_\alpha - \xi_\beta =- \iota_{X}\mathrm{d}\Lambda_{\alpha\beta}$ on overlaps of patches and their morphisms are gauge transformations given by global functions $\varepsilon\in\Coo(M)$. These are exactly the sections which \cite{Col11} calls "horizontal lifts" of a vector $X$. Moreover their brackets \eqref{eq:brac} reduce to the $2$-algebra structure which appears in \cite{DesSae18}:
\begin{gather}\label{eq:generalized geometry}
\begin{aligned}
    \big\llbracket \,\varepsilon\, \big\rrbracket &= \,\mathrm{d}\varepsilon \\
    \big\llbracket (X+\xi_\alpha), (Y+\eta_\alpha)\big\rrbracket &= \bigg([X,Y]+\mathcal{L}_X\eta_\alpha-\mathcal{L}_Y\xi_\alpha - \frac{1}{2}\mathrm{d}(\iota_X\eta_\alpha-\iota_Y\xi_\alpha)+\iota_X\iota_YH\bigg) \\
    \big\llbracket (X+\xi_\alpha), \,\varepsilon\, \big\rrbracket &= \,\mathcal{L}_X\varepsilon  \\
    \big\llbracket (X+\xi_\alpha), (Y+\eta_\alpha), (Z+\zeta_\alpha)\big\rrbracket &= \frac{1}{3!}\bigg( \iota_X\iota_Y\mathrm{d}\zeta_\alpha + \frac{3}{2}\iota_X\mathrm{d}\iota_Y\zeta_\alpha + \text{perm.}\bigg)
\end{aligned}\raisetag{20pt}
\end{gather}
For this $2$-algebra we will use the symbol $\mathfrak{at}(\mathcal{M})_{\mathrm{hor}}$. Notice that these horizontal sections $(X+\xi_\alpha)\in\mathfrak{at}(\mathcal{M})_{\mathrm{hor}}$ equipped with the bracket $\llbracket -,-\rrbracket$ from \eqref{eq:generalized geometry} can be also seen as sections $(X+\xi_\alpha)\in\Gamma(M,C)$ of an ordinary Courant algebroid $C$ appearing in Generalized Geometry at the center of a short exact sequence $T^\ast M\rightarrow C \rightarrow TM$. In other words the underlying chain complex of the $2$-algebra $\mathfrak{at}(\mathcal{M})_{\mathrm{hor}}$ will be $\Coo(M)\xrightarrow{\mathrm{d}}\Gamma(M,C)$. Hence if we restrict to horizontal sections we recover explicitly ordinary Generalized Geometry (see \cite{Gua11} for details).
\end{remark}

\noindent However notice that the horizontal sections $(X+\xi_\alpha)\in\mathfrak{at}(\mathcal{M})_{\mathrm{hor}}$ cannot be seen as sections of the para-Hermitian tangent bundle \eqref{eq:parahermtangbund} from digression \ref{dig:paraherbundle}, even if they look similar. This is because the transition functions $\Lambda_{\alpha\beta}$ of the para-Hermitian manifold satisfy $\Lambda_{\alpha\beta}+\Lambda_{\beta\gamma}+\Lambda_{\gamma\alpha}=0$, which is not an actual gerbe cocycle like the one patching the horizontal sections in remark \ref{rem:horizontalsec}.

\vspace{0.2cm}
\noindent Notice from remark \ref{rem:horizontalsec} that $\mathrm{d}\Lambda_{\alpha\beta}$ satisfies the cocycle condition, even if $\Lambda_{\alpha\beta}$ does not. This is why transition functions of the Courant $2$-algebroid define a global vector bundle: the one underlying the ordinary Courant algebroid. Now the reader may wonder why we considered sections from lemma \ref{gengeom} with non-zero gauge transformations on two-fold overlaps of patches instead of horizontal ones. The answer is that there are applications where these data cannot be neglected, such as in geometry of T-duality in the next section.

\vspace{0.2cm}
\noindent Finally, let us very briefly explain how the formalism of $NQ$-manifolds, widely used in literature to deal with $n$-algebroids, can be easily rlated with our formalism.

\begin{digression}[Recovering Extended Riemannian Geometry by \cite{DesSae18}]
Let us recall that a differential-graded manifold (or $NQ$-manifold) is a locally ringed space $(N,\mathcal{O}_{N})$ where $N$ is a topological space and $\mathcal{O}_{N}$ a sheaf of differential-graded algebras on $N$ satisfying some extra properties. The differential graded manifold $T^\ast[2]T[1]U$ with $U\subset M$ and local coordinates $\{x^\mu,\zeta_\mu,\chi^\mu,p_\mu\}$ respectively in degree $0$, $1$, $1$, $2$
and a sheaf of functions $\Coo(-)$ with differential $Q:=\mathrm{d}_{TU}$. As shown by \cite{Roy02}, remarkably, its differential-graded algebra of functions is exactly the Chevalley-Eilenberg algebra of the $2$-algebra of local sections $\mathfrak{at}(\mathcal{M}|_U)$, i.e.
\begin{equation}
    \Coo(T^\ast[2]T[1]U) \,=\, \mathrm{CE}\big(\mathfrak{at}(\mathcal{M}|_U)\big).
\end{equation}
In other words $T^\ast[2]T[1]U$ is just an alternative way to write $\mathfrak{at}(\mathcal{M}|_U)$. Notice that this recovers Extended Riemannian Geometry by \cite{DesSae18} in the simple case of Generalized Geometry. We will explain how we can to recover their geometry of doubled torus bundles in the section 4.
\end{digression}

\subsection{Doubled metric is an orthogonal structure}

In this subsection we will give a global definition of a doubled metric on our doubled space. This will clarify the fundamental intuition by \cite{BCM14}, who discusses the intrinsic higher geometric nature of a doubled metric structure. However it will presented as a structure reduction of the doubled tangent space, in the spirit of \cite{DCCTv2}.

\begin{definition}[Doubled metric]\label{defdoubledmetric}
Let us define the \textit{doubled metric} as a globally defined Riemannian metric on $\mathcal{M}$ in terms of global forms (definition \ref{def:globalformM}) as
\begin{equation}
    \mathcal{H} \, \in \, \Omega^1(\mathcal{M})_{\mathrm{glob}}^{\odot 2}
\end{equation}
\end{definition}

\begin{remark}[Doubled orthogonal structure]
The doubled metric (definition \ref{defdoubledmetric}) defines a structure on the doubled space, which we will write as
\begin{equation}
    \mathcal{M} \longrightarrow \mathbf{Orth}(T\mathcal{M}),
\end{equation}
generalizing the example \ref{ex:orth} of Riemannian geometry. This is obtained by a reduction of the structure group $GL(d)\ltimes \mathbf{B}\mathbb{R}_{\mathrm{conn}}$ of $T\mathcal{M}$ (see lemma \ref{thm:dtbild}) to $O(2d)$, i.e. a diagram
\begin{equation}
   \begin{tikzcd}[row sep=7ex, column sep=5ex]
    & \mathbf{B}O(2d) \arrow[d]\\
    M \arrow[ru]\arrow[r] & \mathbf{B}\big(GL(d)\ltimes \mathbf{B}\mathbb{R}_{\mathrm{conn}}\big)
    \end{tikzcd}
\end{equation}
This can be seen as locally given by local matrices $\mathcal{E}_\alpha\in\Coo\big(T^\ast U_\alpha,\,GL(2d)\big)$ patched by 
\begin{equation}
    \begin{aligned}
    \mathcal{E}_\alpha &= \mathcal{O}_{\alpha\beta}\cdot\mathcal{E}_\beta\cdot\mathcal{N}_{\alpha\beta}, \\
    \mathcal{O}_{\alpha\gamma}  &= \mathcal{O}_{\alpha\beta}\cdot \mathcal{O}_{\beta\gamma},
    \end{aligned}
\end{equation}
where $\mathcal{O}_{\alpha\beta}$ is a $O(2d)$-cocycle and $\mathcal{N}_{\alpha\beta}$ are the transition functions of the doubled tangent bundle $T\mathcal{M}$ (definition \ref{def:doubledtangentbundle}).
The doubled metric from definition \ref{defdoubledmetric} can be recovered by imposing $\mathcal{H}_\alpha := \mathcal{E}_\alpha^\mathrm{T}\mathcal{E}_\alpha$ and patching by $\mathcal{H}_\alpha=\mathcal{N}_{\alpha\beta}^{\mathrm{T}}\mathcal{H}_\alpha\mathcal{N}_{\alpha\beta}$.
\end{remark}

\noindent Hence moduli the space of a general doubled metric is locally $\Coo\big(T^\ast U_\alpha,\,GL(2d)/O(2d)\big)$, which is what we expect from a globally defined metric on $\mathcal{M}$, in the spirit of definition \ref{defdoubledmetric}.

\begin{definition}[Metric doubled space]\label{def:metricdb}
We call \textit{metric doubled space} a couple $(\M,\mathcal{H})$ of a doubled space $\mathcal{M}$ (postulate \ref{post1}) and a doubled metric $\mathcal{H}$ (definition \ref{defdoubledmetric}) on it. 
\end{definition}

\subsection{Global strong constraint is higher cylindricity condition}

In this subsection we will give a global definition of the strong constraint and we will explain how it can be used to formulate a Higher Kaluza-Klein reduction of the doubled metric.

\begin{post}[Global strong constraint]\label{post3}
The background $(\M,\mathcal{H})$ is a metric doubled space (definition \ref{def:metricdb}) such that the doubled metric structure $\mathcal{M}\xrightarrow{\mathcal{H}}\mathbf{Orth}(T\mathcal{M})$ is equivariant under the principal action of $\mathcal{M}$ (remark \ref{principalaction}).
\end{post}

\noindent This is exactly an higher version of the \textit{cylindricity condition} in Kaluza-Klein Theory, which forbids the dependence of the bundle metric on the extra coordinate by asking it is invariant under the principal $U(1)$-action. Since for a principal circle bundle $P$ we have $P/U(1)\cong M$, the bundle metric will be actually a structure on the base manifold $M$. Analogously for the doubled space we have the following.

\begin{remark}[Geometry of the global strong constraint]\label{rem:geomsc}
Postulate \ref{post3}, combined with postulate \ref{post1}, gives locally the strong constraint we know in DFT. Recall the principal action $\rho$ from remark \ref{principalaction} on the doubled space. The equivariance implies that the doubled metric is actually not a structure on $\M$, but on the (homotopy) quotient induced by the principal action
\begin{equation}
    \M /\!/\!_\rho \,\BU \,\cong\, M,
\end{equation}
i.e. on the base manifold. Thus doubled metric locally depends only on physical $x$ coordinates.
\end{remark}

\begin{digression}[Doubled-yet-gauged spacetime of \cite{Park13}]\label{dig:park2}
The global strong constraint of postulate \ref{post3} is exactly a global version of the strong constraint in Park's formulation. Indeed locally the doubled space is given by $T^\ast U_\alpha$ and the principal action $\rho$ by $\tilde{x}_\alpha\mapsto\tilde{x}_\alpha+\eta_\alpha$ for any circle bundle $(\eta_\alpha,\eta_{\alpha\beta})\in\mathbf{H}(M,\BU)$. If we slash out these gauge transformations we end up with the base manifold patch $U_\alpha$. Under strong constraint the dual coordinates are non-physical and in particular physical points correspond to gauge orbits of the doubled space.
\end{digression}

\begin{definition}[Higher Kaluza-Klein reduction]\label{def:dimred}
In \cite{DCCTv2} it is shown that, for any principal bundle $P\xrightarrow{\;\pi\;\,}M$ defined by the map $M\xrightarrow{\;f\;\,}\mathbf{B}G$ and any stack $\mathbf{S}\in\mathbf{H}$, there is an equivalence
\begin{equation}
    \begin{tikzcd}[row sep=5ex, column sep=2.7ex]
    \mathbf{H}\big(P,\, \mathbf{S}\big) \arrow[rr, "\text{reduction}", yshift=1.5ex] & \simeq & \mathbf{H}_{/\mathbf{B}G}\big(M,\, [G,\mathbf{S}]/G\big) \arrow[ll, "\text{oxidation}", yshift=-1.5ex]
    \end{tikzcd}
\end{equation}
where the \textit{reduction} is given by the following map
\begin{equation}
    \left(\begin{tikzcd}[row sep=5ex, column sep=4ex]
    P \arrow[r, "s"] & \mathbf{S}
    \end{tikzcd}\right) \quad\mapsto\quad
    \left(\begin{tikzcd}[row sep=7ex, column sep=5ex]
    & \left[G,\mathbf{S}\right]/G \arrow[d]\\
    M \arrow[ru, "\text{[}G\text{,}s\text{]}/G"]\arrow[r, "f"] & \mathbf{B}G
    \end{tikzcd}\right)
\end{equation}
This equivalence is called \textit{double dimensional reduction} in the reference.
\end{definition}

\begin{theorem}[Higher Kaluza-Klein reduction of equivariant structure]
When the structure $P\xrightarrow{s}\mathbf{S}$ on a principal $G$-bundle $P$ is $G$-equivariant, its Higher Kaluza-Klein reduction (definition \ref{def:dimred}) will be a structure $M\xrightarrow{s/G}\mathbf{S}$ on the base manifold.
\end{theorem}
\begin{proof}
Structure $P\xrightarrow{s}\mathbf{S}$ is equivariant if there exists a map $P/G\simeq M\xrightarrow{s/G}\mathbf{S}$ such that $s/G \circ \pi \cong s$, where $\pi:P\rightarrow M$ is the bundle projection. But this means that $[G,s]/G\cong s/G$.
\end{proof}

\noindent In the following discussion we will apply this abstract definition to the concrete cases of ordinary Kaluza-Klein Theory and finally of the metric doubled space.

\begin{example}[Ordinary Kaluza-Klein reduction of the circle bundle]
Before looking at what happens to doubled space let us reformulate ordinary Kaluza-Klein reduction. Let us consider a circle bundle $P\rightarrow M$ and an equivariant Riemannian metric $\mathcal{H}$ on $P$. We have the reduction
\begin{equation}
    \left(\begin{tikzcd}[row sep=5ex, column sep=4ex]
    P \arrow[r, "\mathcal{H}"] & \mathbf{Orth}(TP)
    \end{tikzcd}\right) \quad\mapsto\quad
    \left(\begin{tikzcd}[row sep=5ex, column sep=8ex]
    & \mathbf{Orth}(TM)\times\BU \arrow[d, "\ast\,\times\,\mathrm{frgt}"]\\
    M \arrow[ru, "(g\text{,}A\text{,}f)"]\arrow[r, "f"'] & \mathbf{B}U(1)
    \end{tikzcd}\right)
\end{equation}
Let us recall that $\mathfrak{at}(P)$ is locally given by sections of $TU_\alpha\times\mathbb{R}$ patched on overlaps $U_\alpha\cap U_\beta$ by
\begin{equation}
    \mathcal{N}_{\alpha\beta} =   \begin{pmatrix}
 N_{\alpha\beta} & 0 \\
 \mathrm{d}f_{\alpha\beta} & 1
 \end{pmatrix}.
\end{equation}
Hence to have a $U(1)$-equivariant orthogonal structure we need to restrict the structure group $GL(d)\times\mathbb{R}^d\subset GL(d+1)$. Since $\Omega^1(U_\alpha)\cong\Coo(U_\alpha,\mathbb{R}^d)$ are isomorphic we can write vielbein as
\begin{equation}
    \mathcal{E}_{\alpha} =   \begin{pmatrix}
 e_{\alpha} & 0 \\
 A_{\alpha} & 1
 \end{pmatrix}.
\end{equation}
But also $O(d)\subset O(d+1)$. Therefore $e_\alpha = h_{\alpha\beta}\cdot e_\beta\cdot N_{\alpha\beta}$ and $A_\alpha-A_\beta = \mathrm{d}f_{\alpha\beta}$. This means that the cocycle to $\mathbf{Orth}(TP)$ is reduced to a one to $\mathbf{Orth}(TM)\times\BU$ and the map $\mathrm{frgt}$ is just the forgetful functor $\BU\longrightarrow\mathbf{B}U(1)$ which forgets the connection data. Hence, if we call $g_\alpha:=e_\alpha^{\mathrm{T}}e_\alpha$, we can rewrite our metric on each local patch as
\begin{equation}
    \mathcal{H}_\alpha=\mathcal{E}^\mathrm{T}_{\alpha}\mathcal{E}_{\alpha} =   \begin{pmatrix}
 g_\alpha + A^{\mathrm{T}}_\alpha A_\alpha & A^{\mathrm{T}}_\alpha \\
 A_\alpha & 1
 \end{pmatrix}
\end{equation}
satisfying $\mathcal{H}_\alpha=\mathcal{N}_{\alpha\beta}^{\mathrm{T}}\mathcal{H}_\beta\mathcal{N}_{\alpha\beta}$ on two-fold overlaps of patches $U_\alpha\cap U_\beta$. This assures that $g$ is a global Riemannian metric on $M$ and the $1$-form is a $U(1)$-gauge field patched by $A_\beta-A_\alpha = \mathrm{d}f_{\alpha\beta}$.
\end{example}

\noindent Let us now Higher Kaluza-Klein reduce our doubled metric structure in a totally analogous way.

\begin{theorem}[Higher Kaluza-Klein reduction of the doubled space]\label{bosfieldsfromgeom}
A metric doubled space $(\mathcal{M},\mathcal{H})$ (definition \ref{defdoubledmetric}) which satisfies global strong constraint (definition \ref{post3}) reduces to a bosonic supergravity background $M$ with a Riemannian metric $g$ and a gerbe structure with connection $(B_\alpha,\Lambda_{\alpha\beta},G_{\alpha\beta\gamma})$. In diagrams we have the following reduction
\begin{equation}
    \left(\begin{tikzcd}[row sep=5ex, column sep=4ex]
    \mathcal{M} \arrow[r, "\mathcal{H}"] & \mathbf{Orth}(T\mathcal{M})
    \end{tikzcd}\right) \quad\mapsto\quad
    \left(\begin{tikzcd}[row sep=5ex, column sep=8ex]
    & \mathbf{Orth}(TM)\times\mathbf{B}^2U(1)_{\mathrm{conn}} \arrow[d, "\ast\,\times\,\mathrm{frgt}"]\\
    M \arrow[ru, "(g\text{,}B\text{,}\Lambda\text{,}G)"]\arrow[r, "(\Lambda\text{,}G)"'] & \mathbf{B}\big(\BU\big)
    \end{tikzcd}\right)
\end{equation}
\end{theorem}

\begin{proof}
As we explained in lemma \ref{thm:dtbild} the local data of a Courant $2$-algebroid is given by a cocycle $N_{\alpha\beta}:M\longrightarrow\mathbf{B}GL(d)$ which describes the tangent bundle $TM$ and the cocycle $\mathrm{d}\Lambda_{\alpha\beta}:M\longrightarrow\mathbf{B}\wedge^2\mathbb{R}^d$. Hence its sections $\mathfrak{at}(\mathcal{M})\cong\mathfrak{X}(M)\ltimes\mathbf{H}(M,\mathbf{B}\mathbb{R}_{\mathrm{conn}})$ are patched by transition functions in $\Coo\big(U_\alpha\cap U_\beta,\, GL(d)\ltimes\wedge^2\mathbb{R}^d\big)$ where $GL(d)\ltimes\wedge^2\mathbb{R}^d \subset O(d,d)$ is exactly the so-called geometric group. Hence we can rewrite the transition functions as $O(d,d)$-valued functions on overlaps of patches, up to gauge transformations of the cocycle $(\Lambda_{\alpha\beta},G_{\alpha\beta\gamma})$, as
\begin{equation}
    \mathcal{N}_{\alpha\beta} =   \begin{pmatrix}
 N_{\alpha\beta} & 0 \\
 \mathrm{d}\Lambda_{\alpha\beta} & N_{\alpha\beta}^{-\mathrm{T}}
 \end{pmatrix},
\end{equation}
Consequently we can write the vielbein as a local $O(d,d)$-valued function on each patch $U_\alpha$ as
\begin{equation}
    \mathcal{E}_\alpha =   \begin{pmatrix}
 e_\alpha & 0 \\
 -e^{-\mathrm{T}}_\alpha B_\alpha & e^{-\mathrm{T}}_\alpha
 \end{pmatrix}
\end{equation}
where $e_{\alpha}$ and $B_\alpha$ are respectively a $GL(d)$-valued and a $\wedge^2\mathbb{R}^d$-valued function on patches $U_\alpha$. Since $O(d,d)\cap O(2d) = O(d)\times O(d)$, the local $O(2d)$ symmetry of the veilbein breaks to $O(d)\times O(d)$. Hence we can write the doubled metric as
\begin{equation}\label{eq:doubledmetric0}
    \mathcal{H}_\alpha=\mathcal{E}^\mathrm{T}_{\alpha}\mathcal{E}_{\alpha} =   \begin{pmatrix}
 g_\alpha - B^{ }_\alpha g^{-1}_\alpha B_\alpha & B^{ }_\alpha g^{-1}_\alpha \\
 -g^{-1}_\alpha B_\alpha & g^{-1}_\alpha
 \end{pmatrix}
\end{equation}
where we called the symmetric matrix $g_\alpha:=e_\alpha^\mathrm{T}e_\alpha$. On two-fold overlaps of patches $U_\alpha\cap U_\beta$ the doubled metric is hence patched by $\mathcal{H}_\alpha = \mathcal{N}^{\mathrm{T}}_{\alpha\beta}\mathcal{H}_\beta \mathcal{N}_{\alpha\beta}$, which assures that $g$ is a globally defined tensor on $M$ and that $B_\alpha$ is patched by $B_\beta-B_\alpha = \mathrm{d}\Lambda_{\alpha\beta}$ and hence it is a gerbe connection. 
Therefore $\mathbf{Orth}(T\mathcal{M})$ breaks to $\mathbf{Orth}(TM)\times\mathbf{B}^2U(1)_{\mathrm{conn}}$ and the map $\mathrm{frgt}$ is just the forgetful functor $\mathbf{B}^2U(1)_{\mathrm{conn}}\rightarrow\mathbf{B}(\BU)$ which forgets the connection.
\end{proof}

\noindent The moduli space, which is locally given by $\Coo\big(U_\alpha,\,GL(2d)/O(2d)\big)$, is therefore broken to $\Coo\big(U_\alpha,\,O(d,d)/\big(O(d)\times O(d)\big)\big)$ and glued by the transition functions $\mathcal{N}_{\alpha\beta}$ of the Courant $2$-algebroid $\mathfrak{at}(\mathcal{M})$. Hence we recovered Generalized Geometry by higher Kaluza-Klein reduction.

\begin{remark}[Strong constrained doubled metric]
By using gerbe connection $\omega_B\in\Omega^2(\mathcal{M})_{\mathrm{glob}}$ of lemma \ref{thm:gerbeconn} we are able to write explicitly the doubled metric in \eqref{eq:doubledmetric0} as a globally defined tensor $\mathcal{H}\in\Omega^2(\mathcal{M})^{\odot 2}_{\mathrm{glob}}$ on the doubled space
\begin{equation}\label{eq:doubledmetric}
    \mathcal{H} = g \oplus g^{\mu\nu}(\mathrm{d}\tilde{x}_\mu + B_{\mu\lambda}\mathrm{d}x^\lambda)\otimes(\mathrm{d}\tilde{x}_\nu + B_{\nu\lambda}\mathrm{d}x^\lambda),
\end{equation}
which is familiar to para-Hermitian framework. 
Notice the gerbe connection $\omega_B = \mathrm{d}\tilde{x}_{\alpha}-B_{\alpha}$ plays an analogous role of the connection $1$-form of a circle bundle in a Kaluza-Klein metric.
\end{remark}

\begin{remark}[Check: invariance of doubled metric]
For any gauge transformation, i.e. circle bundle $(\eta_\alpha,\eta_{\alpha\beta})\in\mathbf{H}(M,\BU)$, sections transform according to $\tilde{x}_\alpha\mapsto\tilde{x}_\alpha+\eta_\alpha$ on each patch, while the connection to $B_\alpha\mapsto B_\alpha + \eta_\alpha$. Hence doubled metric \eqref{eq:doubledmetric} is invariant.
\end{remark}

\begin{remark}[Isometry $2$-group of the doubled metric]\label{rem:isometry}
Given a metric doubled space $(\mathcal{M},\mathcal{H})$ (definition \ref{defdoubledmetric}) which satisfies global strong constraint (definition \ref{post3}) there is a sub-$2$-group $\mathbf{Iso}(\mathcal{M},\mathcal{H})\subset \mathbf{Aut}_/(\Lambda,G) = \mathrm{Diff}(M)\ltimes\mathbf{H}(M,\,\BU)$ of automorphisms of the doubled space which preserve the doubled metric structure:
\begin{equation}
\begin{tikzcd}[row sep=scriptsize, column sep=5ex]
   1 \arrow[r] & \mathbf{H}(M,\,\BU)\, \arrow[r, hook] & \,\mathbf{Iso}(\mathcal{M},\mathcal{H})\, \arrow[r, two heads] & \, \mathrm{Iso}(M,g)\arrow[r] & 1.
\end{tikzcd}
\end{equation}
i.e. of automorphisms covering the group of isometries of the Riemannian base manifold $(M,g)$.
\end{remark}

\begin{remark}[Generalized vielbein]\label{rem:genvielbein}
There are $2d$ global sections of $\mathfrak{at}(\mathcal{M})_{\mathrm{hor}}$ (see remark \ref{rem:horizontalsec}) giving the generalized veilbein $\mathcal{E}_\alpha$ in terms of horizontal generalized vectors
\begin{equation}
    \begin{aligned}
    \mathbbvar{e}_\mathrm{M} :=\begin{cases}(e_\mu + \iota_{e_\mu}B_\alpha, \, 0), & \mathrm{M} = \mu\\(e^\mu, \, 0), & \mathrm{M} = d+\mu\end{cases}.
    \end{aligned}
\end{equation}
where $\{e_\mu\}\subset\Omega^1(M)$ are the $d$ vielbein $1$-forms of the Riemannian metric $g=\delta_{\mu\nu}\,e^\mu\otimes e^\nu$ on the base manifold $M$, while $\{e_\mu\}\subset\mathfrak{X}(M)$ are their dual veilbein vectors.
\end{remark}

\begin{digression}[Recovering Born geometry]
The doubled metric we defined (see definition \ref{defdoubledmetric}) can be immediately reduced to the one introduced  in para-Hermitian geometry by \cite{Svo18} and further clarified by \cite{MarSza18}, \cite{MarSza19}. Recall the setting of digression \ref{digd}. Our doubled metric $\mathcal{H}$ on each patch $T^\ast U_\alpha$ of the doubled space automatically satisfy
\begin{equation}
    \eta^{-1}_B\mathcal{H} = \mathcal{H}^{-1}\eta_B, \qquad \omega^{-1}_B\mathcal{H} = -\mathcal{H}^{-1}\omega_B.
\end{equation}
It is immediate to check it by using the basis of remark \ref{rem:genvielbein}, so that we obtain
\begin{equation}
    (\eta_B)_{\mathrm{MN}} = \begin{pmatrix}0 & 1 \\1& 0 \end{pmatrix}, \quad (\omega_B)_{\mathrm{MN}} = \begin{pmatrix}0 & 1 \\-1 & 0 \end{pmatrix}, \quad \mathcal{H}_{\mathrm{MN}}= \begin{pmatrix}g & 0 \\0 & g^{-1} \end{pmatrix}.
\end{equation}
A diffeomorphism of $T^\ast U_\alpha$ which preserves both the para-Hermitian structure and the generalized metric is then an ordinary isometry of the Riemannian metric $g$ on $U_\alpha$. In terms of structure groups we have the familiar expression
\begin{equation}
O(d,d) \,\cap\, Sp(2d,\mathbb{R}) \,\cap\, O(2d) \,=\, O(d).
\end{equation}
\end{digression}

\noindent Let us now conclude this subsection by giving a quick example of a relevant Higher Kaluza-Klein reduction of some structure which is not the doubled metric.

\begin{example}[D-branes as objects in the doubled space]
We can define the stack $\mathbf{K}U_{\mathrm{conn}}$ by
\begin{equation}
    \mathbf{K}U_{\mathrm{conn}} \; := \; \prod_{k\in\mathbb{N}}\mathbf{B}^{2k+1}U(1)_{\mathrm{conn}}
\end{equation}
which is a moduli stack of abelian $2k$-gerbes for any $k\in\mathbb{N}$. Now, an equivariant map from the doubled space to this stack is immediately Higher Kaluza-Klein reduced as it follows:
\begin{equation}
    \left(\begin{tikzcd}[row sep=5ex, column sep=3ex]
    \mathcal{M} \arrow[r, " "] & \mathbf{K}U_{\mathrm{conn}}
    \end{tikzcd}\right) \quad\mapsto\quad
    \left(\begin{tikzcd}[row sep=5ex, column sep=8ex]
    & \qquad\;\, \mathbf{K}U_{\mathrm{conn}}/\!/\BU \arrow[d, " "]\\
    M \arrow[ru, "\text{KR, RR fields}"]\arrow[r, "(\Lambda\text{,}G)"'] & \mathbf{B}\big(\BU\big)
    \end{tikzcd}\right)
\end{equation}
The reduced map $M\rightarrow\mathbf{K}U_{\mathrm{conn}}/\!/\BU$ is then the \v{C}ech cocycle of $2k$-gerbes on $M$ twisted by the cocycle $(\Lambda_{\alpha\beta},G_{\alpha\beta})$ of the Kalb-Ramond field. Their curvatures will be then the usual
\begin{equation}
    \begin{aligned}
        \mathrm{d}H =0, \quad \mathrm{d}F_{D0} =0, \quad \mathrm{d}F_{D2} + F_{D0}\wedge H=0, \\
        \mathrm{d}F_{D4} + F_{D2}\wedge H=0, \quad \mathrm{d}F_{D6} + F_{D4}\wedge H=0, \quad \mathrm{d}F_{D8} + F_{D6}\wedge H=0,
\end{aligned}
\end{equation}
for the RR fields for Type IIA String Theory, as explained in a more general setting and formally by \cite[p.22]{FSS18x}. This means that RR fields are simpler if they are thought as fields on the doubled space rather than on the base manifold. This appears consistent with the unified algebra description of spacetime and branes by Lie algebra extensions in \cite{FSS18}.
\end{example}

\section{Application: global geometry of T-duality}

In this section we will derive T-duality by Kaluza-Klein reducing our doubled space, exactly in the spirit of \cite{HohSam13KK}. However, since in our definition the doubled space is a stack, we will retain also global higher gauge data on the base manifold, which is impossible with bare manifolds.
From the reduction of the doubled space we will recover the familiar global geometry of abelian T-duality developed by \cite{Bou03}, \cite{Bou03x}, \cite{Bou03xx}, \cite{Bou04} and \cite{Bou08}.
In the last two subsections we will widen the discussion to non-abelian T-duality (see \cite{Bug19}) and to the general Poisson-Lie T-duality by briefly investigating their global geometry in this context.

\subsection{Cascade of Higher Kaluza-Klein reductions}

In this first subsection we will introduce the concept of cascade of Higher Kaluza-Klein reductions. These become possible when the base manifold $M$ of the doubled space $\mathcal{M}\xrightarrow{\bbpi}M$ is itself a torus bundle and hence we can further perform dimensional reductions. Let us first look at a simple example of cascade of ordinary Kaluza-Klein reductions to gain intuition about it.

\begin{example}[Cascade of ordinary Kaluza-Klein reductions]\label{ex:cascade}
Let $P\rightarrow M$ be a circle bundle over a manifold $M$. If $\mathcal{H}$ is a $U(1)$-invariant Riemannian metric on $P$ it can be Kaluza-Klein reduced to a metric $g$ and a $U(1)$-gauge connection $A_\alpha$ on $M$. Now let $M\rightarrow M_0$ be a torus $T^{n}$-bundle over a base manifold $M_0$. If the starting circle bundle is equivariant under the torus action, then its total space is actually a principal $T^{n+1}$-bundle on $M_0$ and it can be written as
\begin{equation}\label{eq:circbundred}
   P = P_0\times_{M_0}M 
\end{equation}
where $P_0$ is some circle bundle on $M_0$. Because of the equivariance we can further KK reduce the connection $A_\alpha$ of $P$ to the connection $A^{(1)}_\alpha$ of $P_0$ and $n$ global scalar fields $A^{(0)}_{i}$. Analogously we can reduce the metric $g$ on $M$ to a metric $g^{(2)}$ on $M_0$ and $n(n+1)/2$ global scalar fields $g^{(0)}_{ij}$ on $M_0$. This is equivalently the result of a Kaluza-Klein reduction from $P$ to the base $M_0$.
\end{example}

\noindent Let us now investigate what happens for the doubled space.

\begin{remark}[Cascade of Higher Kaluza-Klein reductions]
Let us consider a doubled space $\mathcal{M}\xrightarrow{\bbpi} M$ where the spacetime $M$ is itself a principal $T^n$-bundle on a smooth base manifold ${M_0}$
\begin{equation}
\begin{tikzcd}[row sep=5.5ex, column sep=scriptsize]
   T^n \arrow[r,hook] & M \arrow[d,two heads, "\pi"] \\
   & M_0
\end{tikzcd}
\end{equation}
The doubled space $\M$ is then the total space of a principal $\BU$-bundle over a $T^n$-bundle over the base manifold $M_0$. Generalizing the ordinary example \ref{ex:cascade} we will perform a cascade of Higher Kaluza-Klein reductions of the doubled space $\mathcal{M}$ to the base manifold $M_0$. In particular, we will first have a Higher Kaluza-Klein reduction, as explained in lemma \ref{bosfieldsfromgeom}, and then an ordinary one from the torus bundle $M$ to $M_0$. Schematically:
\begin{equation}
    \begin{tikzcd}[row sep=scriptsize, column sep=14ex]
    \M \arrow[rr, bend left=50, dotted, "\text{total Higher KK reduction}"]
    \arrow[r,"\bbpi", two heads] \arrow[r, bend left=-55, dotted, "\text{Higher KK reduction}"'] & M \arrow[r,"\pi", two heads] \arrow[r, bend left=-55, dotted, "\text{KK reduction}"'] & M_0
\end{tikzcd}
\end{equation}
Hence we need to look at the following pullback diagrams of stacks
\begin{equation}
\begin{tikzcd}[row sep=14ex, column sep=10ex]
\M \arrow[r]\arrow[d, "\bbpi"] & \ast \arrow[d] & \\
M \arrow[r, "(\Lambda\text{, }G)"]\arrow[d, "\pi"] & \mathbf{B}(\mathbf{B}U(1)_{\mathrm{conn}}) \arrow[r]\arrow[d] & \ast \arrow[d]\\
M_0 \arrow[r, "\text{?}"]\arrow[rr, "f", bend right=30]& \left[T^n,\mathbf{B}(\mathbf{B}U(1)_{\mathrm{conn}})\right]/T^n \arrow[r] & \mathbf{B}T^n
\end{tikzcd}
\end{equation}
Since $M$ is a principal $T^n$-bundle over $M_0$, we can choose a good cover $\mathcal{V}=\{V_\alpha\}$ for $M$ such that $\mathcal{U}=\{U_\alpha\}$ with $U_\alpha=\pi(V_\alpha)$ is a good cover for the base $M_0$. Since we are not working with a good cover for $M$, we will consider differential forms with integral periods. We know from \eqref{eq:doubledmetric} that a doubled metric satisfying the global strong constraint is of the form
\begin{equation}\label{torusdmetric}
    \mathcal{H} = g \oplus g^{\mu\nu}(\mathrm{d}\tilde{x}_{\mu}+B_{\mu\lambda}\mathrm{d}x^{\lambda})\otimes(\mathrm{d}\tilde{x}_{\nu}+B_{\nu\lambda}\mathrm{d}x^{\lambda})
\end{equation}
where $B_\alpha$ is the gerbe connection. In adapted coordinates we can naturally define the principal connection $\xi$ on $M$ by $g(\partial_i,-)=g^{(0)}_{ij}\xi^j$ with $\partial_i := \partial/\partial\theta^i$ in adapted coordinates. By using the torus connection we can split the metric in horizontal and vertical part:
\begin{equation}
    g = g^{(2)} + \langle g^{(0)},\xi\odot\xi\rangle
\end{equation}
where $g^{(2)}$ is an horizontal symmetric tensor and the $g^{(0)}_{ij}$ are $n(n+1)/2$ moduli fields on $M$.
We can use definition \ref{def:dimred} to Higher Kaluza-Klein reduce the doubled space to the base $M_0$ by
\begin{equation}
    \begin{aligned}
    \mathbf{H}\big(M,\,\mathbf{B}(\BU)\big) \,&\simeq\, \mathbf{H}\big(M_0,\,\big[T^n,\,\mathbf{B}(\BU)\big]/T^n\big) \\
    (\Lambda_{\alpha\beta},\,G_{\alpha\beta\gamma}) \,&\mapsto\, (\Lambda^{(0)}_{\alpha\beta},\,\Lambda^{(1)}_{\alpha\beta},\,G_{\alpha\beta\gamma}).
    \end{aligned}
\end{equation}
We can split the curvature $H\in\Omega^3(M)$ of the doubled space in horizontal and vertical parts by
\begin{equation}\label{splitH}
    H = H^{(3)} + \langle H^{(2)}\wedge \xi \rangle + \frac{1}{2}\langle H^{(1)}\wedge\xi\wedge\xi\rangle +\frac{1}{3!}\langle H^{(0)},\xi\wedge\xi\wedge\xi\rangle
\end{equation}
 where $H^{(k)}$ are globally defined $\wedge^{3-k}\mathbb{R}^n$-valued differential forms on the base manifold $M_0$. Now on patches $V_\alpha$ and overlaps $V_\alpha \cap V_\beta$ of $M$ we can use the connection of the torus bundle to split the connections of the gerbe in a horizontal and vertical part
\begin{equation}\label{splithv}
\begin{aligned}
    B_\alpha &= B_\alpha^{(2)} + \langle B_\alpha^{(1)}\wedge \xi \rangle + \frac{1}{2}\langle B_\alpha^{(0)}, \xi\wedge\xi \rangle \\
    \Lambda_{\alpha\beta} &= \Lambda_{\alpha\beta}^{(1)} + \langle \Lambda_{\alpha\beta}^{(0)}, \xi\rangle
\end{aligned}
\end{equation}
where $B^{(k)}_\alpha$, and $\Lambda^{(k)}_{\alpha\beta}$ are all local horizontal differential forms on spacetime $M$. For the following calculations we will follow the ones in \cite{BelHulMin07}. The expression of the curvature of the doubled space becomes
\begin{equation}
    \begin{aligned}
    H^{(0)}_{ijk}&= \partial_{[i}B^{(0)}_{\alpha\,jk]} \\
    H^{(1)}_{ij}&= \mathbf{d}B^{(0)}_{\alpha\,ij} - \partial_{[i}B^{(1)}_{\alpha\,j]} \\
    H^{(2)}_i &= \mathbf{d}B^{(1)}_{\alpha\,i} + \mathcal{L}_{\partial_i}B^{(2)}_{\alpha} -\langle B^{(0)}_\alpha, F\rangle_i \\
    H^{(3)} &= \mathbf{d}B_\alpha^{(2)} - \langle B^{(1)}_\alpha\wedge F\rangle
    \end{aligned}
\end{equation}
where $\mathbf{d}:\Omega^p(M_0)\rightarrow\Omega^{p+1}(M_0)$ is just the exterior derivative on the base manfiold $M_0$. The patching conditions of the connection $2$-form now become 
\begin{equation}
    \begin{aligned}
    B_{\beta\,ij}^{(0)} - B_{\alpha\,ij}^{(0)} &= \partial_{[i}\Lambda^{(0)}_{\alpha\beta\,j]} \\
    B_{\beta\,i}^{(1)} - B_{\alpha\,i}^{(1)} &= \mathbf{d}\Lambda^{(0)}_{\alpha\beta\,i} - \mathcal{L}_{\partial_i}\Lambda^{(1)}_{\alpha\beta} \\
    B_\beta^{(2)} - B_\alpha^{(2)} &= \mathbf{d}\Lambda^{(1)}_{\alpha\beta} + \langle \Lambda^{(0)}_{\alpha\beta}, F\rangle
    \end{aligned}
\end{equation}
where $F=\mathrm{d}\xi\in\Omega^2(M_0)$ is the curvature of the torus bundle $M\xrightarrow{\pi}M_0$. Finally on three-fold overlaps we get
\begin{equation}
    \begin{aligned}
       \Lambda_{\alpha\beta\,i}^{(0)} + \Lambda_{\beta\gamma\,i}^{(0)} + \Lambda_{\gamma\alpha\,i}^{(0)} &= \frac{\partial}{\partial\theta^i}G_{\alpha\beta\gamma} \\
    \Lambda_{\alpha\beta}^{(1)} + \Lambda_{\beta\gamma}^{(1)} + \Lambda_{\gamma\alpha}^{(1)} &= \mathbf{d}G_{\alpha\beta\gamma}
    \end{aligned}
\end{equation}
where $\theta_\alpha^i$ are the adapted coordinates on the torus fiber and $\partial_i := \partial/\partial\theta^i$. 
\end{remark}

\subsection{Topology of a doubled space over a torus-bundle}
As seen in remark \ref{rem:ddds}, doubled spaces over a base manifold $M$ are topologically classified by their Dixmier-Douady class $[H]\in H^3(M,\mathbb{Z})$, which physically encodes the $H$-flux. In this subsection we will closely follow \cite{Bou04} to introduce the machinery needed to describe the topological data of a doubled space on a torus bundle. Let us start from the trivial example.

\begin{example}[Doubled space over trivial torus bundle]
For trivial torus bundles $M=M_0\times T^n$ one can just use K\"unneth theorem to rewrite the $3$rd cohomology group as
\begin{equation}\label{eq:kunneth}
    H^3(M_0\times T^n,\mathbb{Z}) \;\cong \bigoplus_{k=0,1,2,3}H^{3-k}(M_0,\mathbb{Z})\otimes_\mathbb{Z} H^{k}(T^n,\mathbb{Z})
\end{equation}
The cohomology ring of the torus is $H^\ast(T^n,\mathbb{Z})\cong\mathbb{Z}[u_1,\dots,u_n]/\langle u^2_1,\dots,u^2_n\rangle$ where $u_i$ are generators of $H^1(S^1,\mathbb{Z})$. Hence we have that $u_i:=[\mathrm{d}\theta^i]$ for $i=1,\dots,n$ are the generators of $H^1(T^n,\mathbb{Z})$, while the cup products $[\mathrm{d}\theta^i]\smile[\mathrm{d}\theta^j]=[\mathrm{d}\theta^i\wedge\mathrm{d}\theta^j]$ are the generators of $H^2(T^n,\mathbb{Z})$ and so on. Therefore we can expand the Dixmier-Douady class as
\begin{equation}\label{eq:topgerbe}
    [H] = [H^{(3)}] + \, \langle [H^{(2)}]\smile[\mathrm{d}\theta] \rangle \, + \, \frac{1}{2}\langle[H^{(1)}]\smile[\mathrm{d}\theta\wedge\mathrm{d}\theta]\rangle \, + \, \frac{1}{3!}\langle[H^{(0)}]\smile[\mathrm{d}\theta\wedge\mathrm{d}\theta\wedge\mathrm{d}\theta]\rangle
\end{equation}
where $H^{(k)}\in\Omega^k(M_0,\wedge^{3-k}\mathbb{R}^n)$ are differential forms on the base manifold $M_0$ and $\langle-,-\rangle$ is just the contraction of the exterior $\wedge\,\mathbb{R}^n$ algebra indices. Hence the topology of this doubled space is encoded by the cohomology classes $[H^{(k)}]\in H^k(M_0,\wedge^{3-k}\mathbb{Z}^n)$ with $k=0,1,2,3$.
\end{example}

\noindent However this construction cannot be immediately extended to non-trivial torus bundles.

\begin{definition}[$T^n$-invariant forms]
We define $\Omega^{p}(M)_{\mathrm{inv}}$ as the subset of $1$-forms $\alpha\in\Omega^{p}(M)$ such that $\mathcal{L}_{\partial_i}\alpha = 0$ in adapted coordinates $\partial_i := \partial/\partial\theta^i$. Let us define the differential forms
\begin{equation}
    \Omega^{p,q}(M_0,\wedge\,\mathbb{R}^n) \,:=\, \bigoplus_{k=0}^{p} \Omega^{k}(M_0,\,\wedge^{p+q-k}\mathbb{R}^n)
\end{equation}
Let $\xi\in\Omega^1(M,\mathbb{R}^n)$ be the connection $1$-form of the torus bundle $M\xrightarrow{\pi}M_0$. There is a natural isomorphism $\Omega^{p}(M)_{\mathrm{inv}} \cong \Omega^{p,0}(M_0,\wedge\,\mathbb{R}^n)$ given by
\begin{equation}\label{eq:isomorphism}
    \mathcal{I}:\,\, \alpha = \sum_{k=0}^p\frac{1}{(p-k)!}\langle\alpha^{(k)}\wedge \xi\wedge\cdots\wedge\xi\rangle \;\,\,\mapsto\;\,\, \big(\alpha^{(p)},\,\alpha^{(p-1)},\,\cdots,\,\alpha^{(0)}\big) 
\end{equation}
and we call $D := \mathcal{I}\circ \mathrm{d} \circ \mathcal{I}^{-1}$ the differential under the isomorphism. The sequence $\Omega^{\bullet,q}(M_0,\wedge\,\mathbb{R}^n)$ for a fixed $q$ equipped with differential $D$ defines a cochain complex, whose integer cohomology we call $H^{\bullet,q}_D(M_0,\wedge\,\mathbb{Z}^n)$. Therefore we have  $H^p(M,\mathbb{Z})_{\mathrm{inv}} \cong H^{p,0}_D(M_0,\wedge\,\mathbb{Z}^n)$.
\end{definition}

\noindent It is also possible to prove that there is isomorphism $\Omega^{2}(M,\mathbb{R}^n)_{\mathrm{inv}} \cong \Omega^{2,1}(M_0,\wedge\,\mathbb{R}^n)$ for $T^n$-invariant $\mathbb{R}^n$-valued $2$-forms, which imply $H^{2}(M,\mathbb{Z}^n)_{\mathrm{inv}} \cong H^{2,1}_D(M_0,\wedge\,\mathbb{Z}^n)$.

\begin{theorem}[$T^n$-invariant representatives]
Any element of the cohomology group $H^p(M,\mathbb{Z})$ of a $T^n$-bundle $M\xrightarrow{\pi}M_0$ can be represented by a closed $T^n$-invariant form, i.e. there is an isomorphism $H^p(M,\mathbb{Z})\cong H^p(M,\mathbb{Z})_{\mathrm{inv}}$.
\end{theorem}

\begin{remark}[Dimensionally reduced Gysin sequence]
Any torus bundle $M\xrightarrow{\pi}M_0$ comes with a long exact sequence, which is called \textit{dimensionally reduced Gysin sequence}. This is given by
\begin{equation*}
    \cdots \rightarrow{} H^p(M_0,\mathbb{Z}) \xrightarrow{\pi^\ast} H^{p,0}_D(M_0,\wedge\,\mathbb{Z}^n) \xrightarrow{\pi_\ast} H^{p-1,1}_D(M_0,\wedge\,\mathbb{Z}^n) \xrightarrow{\langle\,-\,\smile\,[F]\,\rangle} H^{p+1}(M_0,\mathbb{Z}) \rightarrow{} \cdots
\end{equation*}
where $[F]$ is the first Chern class of the bundle, while $\pi^\ast$ on a given representative $\alpha$ is just the injection $\alpha \mapsto (\alpha,0,\cdots,0)$, while $\pi_\ast$ on a given representative $\alpha$ is the integration along each circle of the fiber $S^1_i\subset T^n$, which depends only on the homology class $[S^1_i]\in H_1(M,\mathbb{Z})$ of the circle and not on its particular representative. Hence map $\pi_\ast$ will be given on the representative by $(\alpha^{(p)},\alpha^{(p-1)},\cdots,\alpha^{(0)}) \mapsto (\alpha^{(p-1)},\cdots,\alpha^{(0)})$. 
\end{remark}

\noindent The Dixmier-Douady class $[H]\in H^3(M,\mathbb{Z})$ of a doubled space on $M\xrightarrow{\pi}M_0$ then corresponds to a cohomology class $[(H^{(3)},H^{(2)},H^{(1)},H^{(0)})]\in H^{3,0}_D(M_0,\wedge\,\mathbb{Z}^n)$ given by
\begin{equation}
    H = H^{(3)} + \langle H^{(2)}\wedge\xi \rangle + \frac{1}{2} \langle H^{(1)}_{ij}\wedge\xi\wedge\xi \rangle + \frac{1}{6}\langle H^{(0)},\xi\wedge\xi\wedge\xi \rangle
\end{equation}

\begin{remark}[Closedness]
If $H\in\Omega^3(M)_{\mathrm{inv}}$ is closed on $M$ we have $\mathrm{d}H=0$, which is translated under isomorphism \eqref{eq:isomorphism} to $D(H^{(3)},H^{(2)},H^{(1)},H^{(0)})=0$. Hence we get the equations
\begin{equation}\label{eq:closeness}
    \mathbf{d}H^{(p)} +\langle H^{(p-1)}\wedge F \rangle =0,
\end{equation}
where $\mathbf{d}$ is the differential on the base $M_0$ and $F\in\Omega^2(M_0,\mathbb{R}^n)$ is the curvature of the $T^n$-bundle $M\xrightarrow{\pi}M_0$. Notice we recover the trivial case \eqref{eq:kunneth} for $F=0$. Also notice that $H^{(0)}$ is always closed on $M_0$, while $H^{(1)}$ is closed on $M_0$ either if the torus bundle is trivial or if $H^{(0)}=0$.
\end{remark}

\begin{remark}[Exactness]
If $H\in\Omega^3(U)_{\mathrm{inv}}$ is exact on $U\subset M$, there exists a $2$-form $B\in\Omega^2(U)$ such that $H=\mathrm{d}B$ on $U$, which is translated under isomorphism \eqref{eq:isomorphism} to
\begin{equation}\label{eq:exactness}
    \begin{aligned}
    H^{(0)}_{ijk} \;&=\; \mathcal{L}_{\partial_{[i}}B^{(0)}_{jk]} \\
    H^{(1)}_{ij} \;&=\; \mathbf{d}B^{(0)}_{ij} - \mathcal{L}_{\partial_{[i}}B^{(1)}_{j]} \\
    H^{(2)}_i \;&=\; \mathbf{d}B^{(1)}_i + \mathcal{L}_{\partial_{i}}B^{(2)} -\langle B^{(0)}, F\rangle_i \\
    H^{(3)} \;&=\; \mathbf{d}B^{(2)} - \langle B^{(1)}\wedge F\rangle
    \end{aligned}
\end{equation}
Notice that $B$ is not required to be $T^n$-invariant here.
\end{remark}

\subsection{Geometrization of globally geometric T-duality}
It was understood by \cite{BelHulMin07} that a gerbe structure over a principal torus $T^n$-bundle, if it is equivariant under its principal torus action, automatically defines a principal $T^{2n}$-bundle over its base manifold. This bundle is nothing but the \textit{correspondence space} of a T-duality, also known as \textit{doubled torus bundle} in DFT literature. In this subsection we will explain how the correspondence space can be recovered from our doubled space by Higher Kaluza-Klein reduction and how T-duality is naturally encoded.

\vspace{0.2cm}

\noindent We will call \textit{globally geometric T-duality} the following first simpler application.

\begin{theorem}[Globally geometric T-duality]\label{thm:corrspace}
Let $(\M,\mathcal{H})$ be a metric doubled space which satisfies the strong constraint (postulate \ref{post3}) and such that the base manifold of $\M\xrightarrow{\bbpi}M$ is itself a principal $T^n$-bundle $M\xrightarrow{\pi}M_0$. 
Now if the doubled space $\M\xrightarrow{\bbpi}M$ is also an equivariant bundle under the principal $T^n$-action of $M$, by applying Kaluza-Klein we have the following.
\begin{enumerate}
    \item[(a)] The doubled space takes the form
\begin{equation}\label{eq:emcorrspace}
    \M \,\simeq\, \mathcal{M}_0\times_{M_0}K
\end{equation}
where
\begin{itemize}
    \item $\M_0$ is some doubled space (as defined in postulate \ref{post1}) on the base manifold $M_0$,
    \item $K := M\times_{M_0}\widetilde{M}$ is a principal $T^{2n}$-bundle on the base manifold $M_0$ with first Chern classes $[F]$ and $[\pi_\ast H]$, known in the literature as the \textit{correspondence space} of a couple of T-dual spacetimes $M$ and $\widetilde{M}$, i.e.
\end{itemize}
\begin{equation}\label{diag:corrspace}
\begin{tikzcd}[row sep=5.2ex, column sep=2.5ex]
 & & M\times_{M_0}\widetilde{M} \arrow[dl, "1\otimes\tilde{\pi}"', two heads]\arrow[dr, "\pi\otimes 1", two heads] & & \\
 T^n \arrow[r, hook] & M \arrow[dr, "\pi"', two heads] & & \widetilde{M}\arrow[dl, "\tilde{\pi}", two heads] & T^n \arrow[l, hook']  \\
 & & M_0 & &
\end{tikzcd}
\end{equation}
\item[(b)] The doubled metric $\mathcal{H}$ reduces on the base $M_0$ to a metric $g^{(0)}$, a Kalb-Ramond field $B_\alpha^{(2)}$, a $T^{2n}$-connection $(A^i_\alpha, B^{(1)}_{\alpha i})$ for $K$ and a set of global moduli fields $(g^{(0)}_{ij},B^{(0)}_{ij})$.
\end{enumerate}
\end{theorem}

\begin{digression}[Internal and external space]
Because of lemma \ref{thm:corrspace} the sections $\Gamma(M_0,\mathcal{M})$ of the doubled space on the base manifold will be of the form $(\tilde{x}_\alpha,\phi_{\alpha\beta},\theta_\alpha,\tilde{\theta}_\alpha)$ where $(\tilde{x}_\alpha,\phi_{\alpha\beta})$ are sections of $\mathcal{M}_0$ and $(\theta_\alpha,\tilde{\theta}_\alpha)$ are sections of the correspondence space $K=M\times_{M_0}\widetilde{M}$, i.e.
\begin{equation}
    \Gamma(M_0,\,\mathcal{M}) \;\simeq\; \Gamma(M_0,\,\mathcal{M}_0) \,\times\, \Gamma(M_0,\,K)
\end{equation}
By conforming to usual DFT nomenclature we can call $(\tilde{x}_\alpha,\phi_{\alpha\beta})$ \textit{external coordinates} and $(\theta^i_\alpha,\tilde{\theta}_{\alpha i})$ \textit{internal coordinates} of the doubled space.
\end{digression}

\noindent Notice that expression \eqref{eq:emcorrspace} is exactly the analogue of \eqref{eq:circbundred} for the doubled space.

\begin{proof}[Proof of lemma \ref{thm:corrspace}]
We can immediately split the generalized metric $\mathcal{H}$ of $\mathcal{M}$ in a Riemannian metric $g$ and a gerbe connection $B_\alpha$ on $M$. Now, by using the torus connection $\xi\in\Omega^1(M,\mathbb{R}^n)$, we can split these metric and gerbe connection in horizontal and vertical components
\begin{equation}
    \begin{aligned}
        g &= \pi^\ast g^{(2)} + \langle \pi^\ast g^{(0)},\xi\odot\xi\rangle \\
        B_\alpha &= \pi^\ast B_\alpha^{(2)} + \langle \pi^\ast B_\alpha^{(1)}\wedge \xi \rangle + \frac{1}{2}\langle  
        \pi^\ast B_\alpha^{(0)}, \xi\wedge\xi \rangle
    \end{aligned}
\end{equation}
where $g^{(2)}$ and $g^{(0)}$ are respectively a metric and a set of moduli fields on $M_0$, while $B^{(k)}_\alpha$ are local $k$-forms on patches $U_\alpha$. We can do the same for the transition functions $(\Lambda_{\alpha\beta},G_{\alpha\beta\gamma})$ of the doubled space
\begin{equation}
    \begin{aligned}
        \Lambda_{\alpha\beta} &= \pi^\ast \lambda_{\alpha\beta} + \langle \pi^\ast \tilde{f}_{\alpha\beta}, \xi \rangle \\
        G_{\alpha\beta\gamma} &= \pi^\ast g_{\alpha\beta\gamma}
    \end{aligned}
\end{equation}
where $\lambda_{\alpha\beta}$ is a local $1$-form on two-fold overlaps of patches $U_\alpha\cap U_\beta$ and $\tilde{f}_{\alpha\beta}$ and $g_{\alpha\beta\gamma}$ are local functions respectively on two-fold and three-fold overlaps of patches. The patching condition $B_\beta-B_\alpha = \mathrm{d}\Lambda_{\alpha\beta}$ on two-fold overlaps of patches becomes
\begin{equation}\label{eq:nonobpatched}
    \begin{aligned}
    B_{\beta\,ij}^{(0)} - B_{\alpha\,ij}^{(0)} &= \,0 \\
    B_{\beta\,i}^{(1)} \,-  B_{\alpha\,i}^{(1)} \,\, &= \,\mathbf{d}\tilde{f}_{\alpha\beta\,i} \\
    B_\beta^{(2)} \,- B_\alpha^{(2)} \,\,&=\, \mathbf{d}\lambda_{\alpha\beta} + \langle \tilde{f}_{\alpha\beta}, F \rangle
    \end{aligned}
\end{equation}
where $F=\mathrm{d}\xi$ is the curvature of $M\rightarrow M_0$, while the condition $\Lambda_{\alpha\beta} + \Lambda_{\beta\gamma} + \Lambda_{\gamma\alpha} = \mathrm{d}G_{\alpha\beta\gamma}$ on three-fold overlaps of patches becomes
\begin{equation}\label{eq:nonobpatched2}
    \begin{aligned}
    \tilde{f}_{\alpha\beta\,i}+\tilde{f}_{\beta\gamma\,i}+\tilde{f}_{\gamma\alpha\,i} &=0 \\
    \lambda_{\alpha\beta}+\lambda_{\beta\gamma}+\lambda_{\gamma\alpha} &=\mathbf{d}g_{\alpha\beta\gamma}
    \end{aligned}
\end{equation}
From \eqref{eq:nonobpatched} we get that $B^{(0)}$ are globally defined scalar fields on the base manifold $M_0$ and hence $H^{(1)}=\mathrm{d}B^{(0)}$ globally, which assures $\big[H^{(1)}\big]=0$.
From \eqref{eq:nonobpatched} and \eqref{eq:nonobpatched2} we get that $\big(B^{(1)}_\alpha,\tilde{f}_{\alpha\beta}\big)$ is a cocycle defining a torus $T^n$-bundle $\widetilde{M}$ with connection on the base manifold $M_0$. Together with the torus bundle $M$ given by the cocycle $(A_\alpha,f_{\alpha\beta})$ we have the following two torus bundles (or equivalently a single $T^{2n}$-bundle) on $M_0$ with local data
\begin{equation}
\begin{aligned}
    A_\beta - A_\alpha &= \mathbf{d}f_{\alpha\beta}, & \quad
    B_\beta^{(1)} - B_\alpha^{(1)} &= \mathbf{d}\tilde{f}_{\alpha\beta}, \\
    f_{\alpha\beta}+f_{\beta\gamma} + f_{\gamma\alpha} &= 0, & \quad
    \tilde{f}_{\alpha\beta}+\tilde{f}_{\beta\gamma} + \tilde{f}_{\gamma\alpha} &= 0,
\end{aligned}
\end{equation}
with first Chern classes given by
\begin{equation}
    \mathrm{c}_1(M)=\big[\mathbf{d}A_\alpha\big]=\big[F\big], \qquad \mathrm{c}_1(\widetilde{M})=\big[\mathbf{d}B^{(1)}_\alpha\big]=\big[ H^{(2)}+\langle B^{(0)},F\rangle\big] = \big[\pi_\ast H\big]
\end{equation}
The $T^{2n}$-coordinates are given by $\theta_\beta - \theta_\alpha = f_{\alpha\beta}$ and $\tilde{\theta}_\beta - \tilde{\theta}_\alpha = \tilde{f}_{\alpha\beta}$, while the $(\tilde{x}_\alpha,\phi_{\alpha\beta})$ sections of $\mathcal{M}_0$ are defined as usual by using the cocycle $(\lambda_{\alpha\beta},\,g_{\alpha\beta\gamma})$.
Notice that $(\tilde{x}_\alpha,\tilde{\theta}_\alpha)\in\Omega^{1,0}_D(U_\alpha,\wedge\,\mathbb{R}^n)$ is exactly the image of the local $1$-form $\tilde{X}_\alpha$ on $M$ of an equivariant section $(\tilde{X}_\alpha,\phi_{\alpha\beta})\in\Gamma(M,\M)$. Thus we obtain the conclusion of the lemma.
\end{proof}

\begin{digression}[$H$-flux and $F$-flux]
The Dixmier-Douady class $[H]\in H^3(M,\mathbb{Z})$ and the first Chern class $[F_i]\in H^2(M_0,\mathbb{Z})$ are respectively called $H$-flux and $F$-flux (or geometric flux) in String Theory literature. Moreover the first Chern class $[\pi_\ast H]_i=[H^{(2)}_i+\langle B^{(0)},F\rangle_i]\in H^2(M_0,\mathbb{Z})$, which is given by the integral of $[H]$ on a basis of $1$-cycles $[S_i^1]\in H_1(T^n,\mathbb{Z})$ of the torus fiber, represents a non-trivial flux compactification of the Dixmier-Douady class.
\end{digression}

\begin{remark}[Geometrical interpretation of fluxes]
Notice that in our Higher-Kaluza Klein framework the $H$-flux and the $F$-flux are not something one puts on the doubled space by hand (for instance by defining some $3$-form on a $2d$-dimensional manifold like in the usual approach), but they are natural \textit{topological properties} of the geometry itself of the doubled space $\mathcal{M}$.
\end{remark}

\begin{remark}[Gerbe curvature]\label{rem:globgeomcurv}
The equation for the curvature $H=\mathrm{d}B_\alpha$ under isomorphism \eqref{eq:isomorphism} becomes $(H^{(3)},H^{(2)},H^{(1)},H^{(0)})=D(B^{(2)}_\alpha,B^{(1)}_\alpha,B^{(0)})$, which is equivalent to
\begin{equation}
    \begin{aligned}
    H^{(0)} \;&=\; 0 \\
    H^{(1)} \;&=\; \mathbf{d}B^{(0)} \\
    H^{(2)} \;&=\; \mathbf{d}B^{(1)}_{\alpha} -\langle B^{(0)}, F\rangle \\
    H^{(3)} \;&=\; \mathbf{d}B_\alpha^{(2)} - \langle B^{(1)}_\alpha\wedge F\rangle
    \end{aligned}
\end{equation}
\end{remark}

\noindent We Now explain in digressions \ref{dig2a} and \ref{rem:hull} how we recover for free two relevant geometric ideas which have been developed to encode the geometry of a globally geometric T-duality.

\begin{digression}[Recovering differential T-duality structures]\label{dig2a}
As we explained in lemma \ref{thm:corrspace}, the bundle structure of the doubled space $\mathcal{M}$ is Kaluza-Klein reduced to the base manifold $M_0$ to the \v{C}ech cocycle $(\lambda_{\alpha\beta},\,g_{\alpha\beta\gamma},\,\tilde{f}_{\alpha\beta},\,f_{\alpha\beta})$ on $M_0$ appearing in the diagram
\begin{equation}
\begin{tikzcd}[row sep=6ex, column sep=10ex]
M_0 \arrow[r, "(\lambda\text{,}g\text{,}\tilde{f})"]\arrow[rr, "f", bend right=30]& \mathbf{B}(\mathbf{B}U(1)_{\mathrm{conn}})/T^n \arrow[r] & \mathbf{B}T^n
\end{tikzcd}
\end{equation} 
Moreover, if we Kaluza-Klein reduce the doubled metric structure $\mathcal{H}$ we obtain also the connection data $A_{\alpha},B^{(1)}_{\alpha},B^{(2)}_{\alpha}$ for this cocycle. If we put these local data together we obtain exactly a \textit{differential T-duality structure} as defined by \cite{KahVal09} and further explored by \cite{DCCTv2}. In the reference this is formalized by a $\mathrm{String}(T^{n}\times T^n)$-bundle with connection, where $\mathrm{String}(T^{n}\times T^n)$ is the $2$-group uniquely defined by the following pullback diagram of group-stacks:
\begin{equation}
\begin{tikzcd}[row sep=6ex, column sep=10ex]
\mathbf{B}\mathrm{String}(T^n\times T^n) \arrow[r]\arrow[d]&\ast\arrow[d] \\ 
\mathbf{B}(T^n\times T^n) \arrow[r, "\langle\,\text{[}\mathrm{c}_1\text{]}\,\smile\,\text{[}\mathrm{c}_1\text{]}\,\rangle"] & \mathbf{B}^3U(1)
\end{tikzcd}
\end{equation}
Indeed from remark \ref{rem:globgeomcurv} we can check that $\mathbf{d}H^{(3)}=\langle F\wedge\widetilde{F}\rangle$ in accord with the references.
\end{digression}

\begin{digression}[Recovering doubled torus bundles]\label{rem:hull}
Because of lemma \ref{thm:corrspace}, on the doubled space $\M \simeq \M_0 \times_{M_0} (M\times_{M_0}\widetilde{M})$, as we have seen, we can expand the vielbein $e=e^{(1)}+e^{(0)}_i\xi^i$ and gerbe connection $B_\alpha=B^{(2)}_{\alpha}+B^{(1)}_{\alpha\,i}\xi^i+B^{(0)}_{ij}\xi^i\wedge\xi^j$ and pack them again into the moduli field of the doubled metric as it follows
\begin{equation}
      \Theta_\alpha := \begin{pmatrix}
 \theta_\alpha \\
 \tilde{\theta}_\alpha
 \end{pmatrix}, \quad \mathcal{H}^{(0)} := \begin{pmatrix}
 g^{(0)}_{ij} - B_{ik}^{(0)}g^{(0)kl}B_{lj}^{(0)} & B_{ik}^{(0)}g^{(0)kj} \\
 -g^{(0)ik}B_{kj}^{(0)} & g^{(0)ij}
 \end{pmatrix}, \quad \mathcal{A}_\alpha := \begin{pmatrix}
 A_\alpha \\
 B^{(1)}_\alpha
 \end{pmatrix}
\end{equation}
which are respectively the coordinates of the fibers of the $T^{2n}$-bundle, the moduli field globally defined on the base manifold $M_0$ and the principal $T^{2n}$-connection on $M_0$. Now there is a natural group action $O(n,n;\mathbb{Z})$, whose elements $\mathcal{O}$ act by
\begin{equation}
      \Theta_\alpha\mapsto \mathcal{O}^{-1}\Theta_\alpha, \quad \mathcal{H}^{(0)}\mapsto \mathcal{O}^{\mathrm{T}}\mathcal{H}^{(0)}\mathcal{O}, \quad \mathcal{A}_\alpha\mapsto \mathcal{O}^{-1}\mathcal{A}_\alpha,
\end{equation}
so that we recover exactly the Buscher rules. Notice that the first Chern class of the correspondence space is rotated by $[\mathbf{d}\mathcal{A}_\alpha]\mapsto \mathcal{O}^{-1}[\mathbf{d}\mathcal{A}_\alpha]$, while the component $H^{(3)}$ is invariant. Hence the \textit{doubled torus bundle} introduced by \cite{Hull06} is nothing but the principal $T^{2n}$-bundle $K$ that we obtain by higher Kaluza-Klein reduction of the doubled space $\mathcal{M}$ to the base manifold $M_0$. 
\end{digression}

\noindent Notice that the invariance of $H^{(3)}$ implies that we can write both $H^{(3)}= \mathbf{d}B^{(2)}_\alpha-\langle B^{(1)}_\alpha\wedge F\rangle$ and $H^{(3)}=  \mathbf{d}\tilde{B}^{(2)}_\alpha-\langle A_\alpha\wedge \pi_\ast H\rangle$ where $\tilde{B}^{(2)}_\alpha$ is the T-dual of ${B}^{(2)}_\alpha$ under the transformation $\mathcal{O}=\left(\begin{smallmatrix*}0&1\\1&0\end{smallmatrix*}\right)$. But this means that we recover the familiar relation $\tilde{B}_\alpha^{(2)} = B_\alpha^{(2)}+\langle A_\alpha\wedge B^{(1)}_\alpha \rangle$.

\vspace{0.2cm}
\noindent By putting together what we learned from digressions \ref{dig2a} and \ref{rem:hull}, we can refine the result (b) of lemma \ref{thm:corrspace} by explicitly writing the Higher Kaluza-Klein reduction of the doubled metric.

\begin{theorem}[Higher Kaluza-Klein reduction of doubled metric]
In the hypothesis of lemma \ref{thm:corrspace} the doubled metric structure $\M\xrightarrow{\,\mathcal{H}\,}\mathbf{Orth}(T\M)$ is Higher Kaluza-Klein reduced to a map
\begin{equation}
    M_0 \,\longrightarrow\, \mathbf{Orth}(TM_0)\times\mathbf{B}\mathrm{String}(T^n\times T^n)_\mathrm{conn}\times O(n,n).
\end{equation}
where $M_0\rightarrow\mathbf{Orth}(TM_0)$ is a Riemannian metric, while $M_0\rightarrow\mathbf{B}\mathrm{String}(T^n\times T^n)_\mathrm{conn}$ is a differential T-duality structure and $M_0\rightarrow O(n,n)$ is a moduli field $\mathcal{H}^{(0)}$.
\end{theorem}

\noindent Recall digression \ref{digd}: the doubled space $\M$ is naturally locally a para-Hermitian manifold, but not globally, since it is a stack. However, as we will explain in the following remark \ref{rem:parahermitianK}, in the special case of a torus compactification the correspondence space $K$ is naturally equipped with a proper global para-Hermitian structure. This means that a para-Hermitian structure on a $T^{2n}$-bundle can equivalently encode all the geometric data of Hull's doubled torus bundle (see remark \ref{rem:hull}), even if it cannot geometrize the whole gerbe. Therefore our framework can explain by using more fundamental higher geometrical arguments why para-Hermitian geometry has been used so effectively in dealing with DFT and T-duality.

\begin{remark}[The correspondence space is a global para-Hermitian bundle]\label{rem:parahermitianK}
We can Kaluza-Klein reduce the global $2$-form connection from \eqref{eq:omegaB} to
\begin{equation}
    \omega_B \,=\, \omega_B^{(2)} + \big(\omega^{(1)}_{B\,i} - \omega^{(0)}_{B\,ij}\xi^j\big)\wedge\xi^i
\end{equation}
Let us then define $\omega_K := (\omega^{(1)}_{B\,i} - \omega^{(0)}_{B\,ij}\xi^j)\wedge\xi^i$.
Since $\xi$, $\tilde{\xi}$ and $B^{(0)}$ are globally defined on the correspondence space $K=M\times_{M_0}\widetilde{M}$ we have the following global $2$-form and tensor on $K$: 
\begin{equation}\begin{aligned}
    \omega_K \,&=\, \left(\tilde{\xi}_i-B^{(0)}_{ij}\xi^j\right)\wedge\xi^i, \\
    \eta_K \,&=\, \left(\tilde{\xi}_i-B^{(0)}_{ij}\xi^j\right)\odot\xi^i,
\end{aligned}\end{equation}
but also the para-complex structure
\begin{equation}
    J_K \,=\, \frac{\partial}{\partial\tilde{\theta}_i} \otimes \left(\tilde{\xi}_i-B^{(0)}_{ij}\xi^j\right) -  \left(\frac{\partial}{\partial \theta^i} - B^{(0)}_{ij}\frac{\partial}{\partial \tilde{\theta}_j}\right) \otimes \xi^i.
\end{equation}
The triple $(K,J_K,\eta_K)$ is a para-Hermitian bundle with para-Hermitian fibre $T^{2n}$. This is nothing but the para-Hermitian structure of \cite[p.~31]{MarSza18} on the fibre $T^{2n}$, but which is non-trivially fibrated over the base manifold $M_0$.
\end{remark}

\noindent Let us now give a brief look to the relation between our proposal and the local geometry for DFT developed by \cite{DesSae18} and \cite{DesSae19}, that is called Extended Riemannian Geometry. 

\begin{digression}[Recovering Extended Riemannian Geometry by \cite{DesSae18}]
Given a patch $U\subset M_0$, we can write the local algebroids of the infinitesimal symmetries of the doubled space $\M$ pulled back to the correspondence space $K$, on the base $M$ and Kaluza-Klein reduced to $M_0$ by
\begin{equation}
\begin{tikzcd}[row sep=9ex,column sep=2ex,nodes={inner sep=4pt}]
   \tilde{\pi}^\ast\M  \ar[dotted]{d} & \M  \ar["\text{reduction}"']{r} \ar[dotted]{d} \ar[leftarrow, "\tilde{\pi}_\ast"]{l} & [5ex]\M_0\times_{M_0}K  \ar[dotted]{d}  \\
       T^\ast[2]T[1](U\times T^{2n})\ar["\Coo(-)/\langle\mathrm{d}\theta^i\rangle"']{d} \ar{r} & T^\ast[2]T[1](U\times T^n)  \ar{r} & [5ex]T^\ast[2]T[1]U\oplus \mathbb{R}^{2n}[1]  \\
      T^\ast[2]T[1]U\oplus (T[1]\oplus T[2]^\ast) T^{2n} \ar["p"]{ru} && 
 \end{tikzcd}
\end{equation}
where the dotted maps send a doubled space $\mathcal{N}$ to the differential-graded manifold $N$ which describes its local symmetries, i.e. which satisfies $\Coo(N)=\mathrm{CE}\big(\mathfrak{at}(\mathcal{N}|_V)\big)$ on any patch $V$ of the considered base manifold. Notice that the submanifold $T^\ast[2]T[1]U\oplus (T[1]\oplus T[2]^\ast) T^{2n}$ over a patch $U\times T^{2n}$ of the correspondence space $K$ is the structure considered by \cite{DesSae18} and \cite{DesSae19}. Now notice that the $0$-degree space of functions on this manifold are sections $\Gamma\big(TU\oplus T^\ast U \oplus TT^{2n}\big)\cong\Gamma\big(T(U\times T^{n}) \oplus T^\ast (U\times T^{n})\big)$. Hence it is isomorphic to the one of $T^\ast[2]T[1](U\times T^n)$ and this isomorphism defines the projection $p$ in the diagram. Of course these local patches can be glued to give a differential-graded fibration on $K$.
What is indeed called "doubled space" in the references is exactly the correspondence space $K=M\times_{M_0}\widetilde{M}$, but equipped with this geometric structure. Hence Extended Riemannian Geometry can be seen as an infinitesimal description of the doubled space $\M$ pulled back to the correspondence space $K$.
\end{digression}

\subsection{Geometrization of globally non-geometric T-duality}

In this subsection we will relax the hypothesis of lemma \ref{thm:corrspace}: this time the doubled space $\mathcal{M}\xrightarrow{\bbpi} M$ will not be necessarily equivariant under the principal $T^n$-action of $M$. In the literature the case where the gerbe connection is only required to satisfy $\mathcal{L}_{\mathbf{k}_i}B_\alpha=0$ on each patch is called \textit{globally non-geometric T-duality}. In terms of transition functions, the differential forms $(\Lambda_{\alpha\beta}, G_{\alpha\beta\gamma})$ are allowed to depend on the coordinates of the fibers as long as $\mathcal{L}_{\mathbf{k}_i}B_\alpha=0$ is satisfied. From \eqref{eq:exactness} we get that the class $[H^{(0)}]$ is still trivial, but $[H^{(1)}]$ generally is not.

\begin{theorem}[Globally non-geometric T-duality]\label{thm:gencorrspace}
Let $(\M,\mathcal{H})$ be a metric doubled space which satisfies the strong constraint (postulate \ref{post3}) and such that the base manifold of $\M\xrightarrow{\bbpi}M$ is itself a principal $T^n$-bundle $M\xrightarrow{\pi}M_0$. 
Now if the automorphisms $\Coo(M_0,T^n)$ of $M\xrightarrow{\pi}M_0$ are isometries of the doubled metric structure, by applying Kaluza-Klein we have the following.
\begin{enumerate}
    \item[(a)] The doubled space has a subbundle $K\subset\mathcal{M}$ (i.e. there exists an inclusion of groupoids $\Gamma(M,K)\subset\Gamma(M,\mathcal{M})$) which is a principal $T^n$-bundle on spacetime $M$ with first Chern class $\mathrm{c}_1(K)=[\pi_\ast H]=[H^{(2)}-\langle H^{(1)}\wedge\xi\rangle]\in H^2(M,\mathbb{Z})^n$ given by the $H$-flux. In diagrams:
\begin{equation}
\begin{tikzcd}[row sep=4.5ex, column sep=2.5ex]
 & & K \arrow[dl, two heads] & & \\
 T^n \arrow[r, hook] & M \arrow[dr, "\pi"', two heads] & & \quad\; & \quad \\
 & & M_0 & &
\end{tikzcd}
\end{equation}
This can be also seen as an affine $T^{2n}$-bundle over the base manifold $M_0$, known in the literature as the \textit{generalized correspondence space} (see digression \ref{rem:tfold}).
\item[(b)] The doubled metric $\mathcal{H}$ reduces on $M_0$ to a metric $g^{(0)}$, a Kalb-Ramond field $B_\alpha^{(2)}$, a $T^{n}$-connection $A^i_\alpha$ for $M$, a $T^{n}$-connection $\widetilde{A}_{\alpha i}$ for $K$ and a set of global moduli fields $g^{(0)}_{ij}$.
\end{enumerate}
\end{theorem}

\begin{digression}[T-fold]\label{rem:tfold}
The generalized correspondence space $K$ is an affine (non principal) $T^{2n}$-bundle over the base manifold $M_0$. Recall that the affine group of the torus is $\mathrm{Aff}(T^{2n})=GL(2n,\mathbb{Z})\ltimes T^{2n}$ and that an affine torus bundle is defined as the associated bundle $K:=Q\times_{\mathrm{Aff}(T^{2n})}T^{2n}$ to some principal $\mathrm{Aff}(T^{2n})$-bundle $Q$. Generalized correspondence spaces are a special class of affine torus bundles where the structure group is restricted to $\wedge^2\mathbb{Z}^n\ltimes T^{2n}\subset \mathrm{Aff}(T^{2n})$, where $\wedge^2\mathbb{Z}^n$ encodes the monodromy. Since $K$ has monodromy there is no well-defined torus subbundle $\widetilde{M}\rightarrow M_0$ which can be seen as the T-dual to the starting $M\rightarrow M_0$. In fact we could perform T-duality on each patch $U_\alpha$ of $M$, but we would obtain a collection of string-background patches which cannot be glued together. In DFT literature this object has been named \textit{T-fold}. Morally speaking we would have a diagram generalizing \eqref{diag:corrspace} of the form
\begin{equation}\label{eq:tfolddiagram}
\begin{tikzcd}[row sep=4.5ex, column sep=2.8ex]
 & & K \arrow[dl, two heads]\arrow[dr, dotted] & \\
 T^n \arrow[r, hook] & M \arrow[dr, "\pi"', two heads] & & [-1ex]\text{T-fold}\arrow[dl, dotted] \\
 & & M_0 &
\end{tikzcd}
\end{equation}
where the dotted arrows are not actual maps between spaces, but only indicative ones. Since there is no well defined dual manifold, this T-duality has \textit{no underlying topological T-duality}. We will explain what are the differential data of this kind of T-duality in remark \ref{rem:tdualityoftfold}.
\end{digression}

\begin{proof}[Proof of lemma \ref{thm:gencorrspace}]
The assumption that the automorphisms of $M\xrightarrow{\pi}M_0$ are the isometries of the doubled metric, i.e.  $\mathbf{Iso}(\M,\mathcal{H})=\Coo(M_0,T^n)$, implies that $\mathcal{L}_{\mathbf{k}_i}B_\alpha=0$ on each $U_\alpha$ where $\{\mathbf{k}^i\}$ are the fundamental vectors. This assures that $\widetilde{F}_i := \iota_{\mathbf{k}_i}H$ is a closed $2$-form on $M$ by
\begin{equation}
\begin{aligned}
    \widetilde{F}_i \;&=\; -\mathrm{d}\iota_{\mathbf{k}_i}B_\alpha \\ 
    \;&=\; \mathrm{d}(B_{\alpha\,i}^{(1)}-B_{\alpha\,ij}^{(0)}\xi^j)
\end{aligned}
\end{equation}
Now we can define on each patch the local connection $1$-forms for this curvature
\begin{equation}
    \widetilde{A}_{\alpha\,i} \,:=\, -\iota_{\mathbf{k}_i}B_\alpha \,=\, B_{\alpha\,i}^{(1)}-B_{\alpha\,ij}^{(0)}\xi^j.
\end{equation}
We will closely follow \cite{BelHulMin07} for the next calculations. In the reference it is proven that these $1$-forms are indeed patched like the connection of a $T^n$-bundle as
\begin{equation}
    \begin{aligned}
        \widetilde{A}_{\beta\,i}-\widetilde{A}_{\alpha\,i} \;&=\; \mathbf{d}\Lambda^{(0)}_{\alpha\beta\,i} - 2(\partial_{i}\Lambda^{(0)}_{\alpha\beta\,j})\xi^j - \partial_{i}\Lambda^{(1)}_{\alpha\beta} \\
        &=\; \mathbf{d}\tilde{f}_{\alpha\beta\,i} - n_{\alpha\beta\,ij}\Big(\mathrm{d}\theta_{\beta}^j+\frac{1}{2}\mathbf{d}f_{\beta\alpha}^j\Big) \\
        &=\; \mathrm{d}\Big(\tilde{f}_{\alpha\beta\,i} - n_{\alpha\beta\,ij}\Big(\theta_{\beta}^j+\frac{1}{2}f_{\beta\alpha}^j\Big)\Big).
    \end{aligned}
\end{equation}
The principal connection $\Xi\in\Omega^1(K,\mathbb{R}^n)$ of the generalized correspondence space $K\rightarrow M$ seen as a $T^n$-bundle over $M$ is then
\begin{equation}
    \Xi_i \,:=\, \mathrm{d}\tilde{\theta}_{\alpha i} + B^{(1)}_{\alpha i}-B^{(0)}_{\alpha ij}\xi^j.
\end{equation}
If in analogy with geometric T-duality we define the local differential $1$-form $\tilde{\xi}_\alpha := \mathrm{d}\tilde{\theta}_\alpha + B^{(1)}_\alpha$, this cannot clearly be globalized. Since we know that $H^{(1)}$ is a closed $1$-form on $M_0$, i.e. $\mathbf{d}H^{(1)}=0$, this will define a \v{C}ech cocycle with patching conditions
\begin{equation}
    \begin{aligned}
        H^{(1)} &= \mathbf{d}B^{(0)}_\alpha \\
        B^{(0)}_\beta - B^{(0)}_\alpha &= n_{\alpha\beta}
    \end{aligned}
\end{equation}
where $n_{\alpha\beta ij}\in 2\pi\mathbb{Z}$.
On the other hand the principal connection $\xi=\mathrm{d}\theta_{\alpha}+A_{\alpha}$ is global on $M$. Therefore from $B_{\beta}^{(0)}-B_{\alpha}^{(0)} = n_{\alpha\beta}$ we get the patching conditions
\begin{equation}\label{eq:hflux}
\begin{pmatrix}
 \xi_\beta \\
 \tilde{\xi}_\beta
 \end{pmatrix} = \begin{pmatrix}
 1 & 0 \\
 n_{\alpha\beta} & 1
 \end{pmatrix} \begin{pmatrix}
 \xi_\alpha \\
 \tilde{\xi}_\alpha
 \end{pmatrix}
\end{equation}
Where $n_{\alpha\beta}$ is the monodromy matrix of the dual torus fibers. Hence $K$ is equivalently an affine torus $T^{2n}$-bundle on the base manifold $M_0$. The affine transitions functions can be written as
\begin{equation}
\begin{pmatrix}
 \theta_\beta - \frac{1}{2}f_{\alpha\beta} \\
 \tilde{\theta}_\beta - \frac{1}{2}\tilde{f}_{\alpha\beta}
 \end{pmatrix} = \begin{pmatrix}
 1 & 0 \\
 n_{\alpha\beta} & 1
 \end{pmatrix} \begin{pmatrix}
 \theta_\alpha - \frac{1}{2}f_{\beta\alpha} \\
 \tilde{\theta}_\alpha - \frac{1}{2}\tilde{f}_{\beta\alpha}
 \end{pmatrix}
\end{equation}
It has been also proven by \cite{BelHulMin07} that the horizontal components of $B_\alpha$ are patched by
\begin{equation}
    B^{(2)}_\beta - B^{(2)}_\alpha = \mathbf{d}\lambda_{\alpha\beta} + \langle \tilde{f}_{\alpha\beta},F \rangle + 
    \frac{1}{2}\left\langle n_{\alpha\beta}, \, \left(A_{\beta}-\frac{1}{2}\mathbf{d}f_{\alpha\beta}\right) \wedge \left(A_{\beta}-\frac{1}{2}\mathbf{d}f_{\alpha\beta}\right) \right\rangle
\end{equation}
on two-fold overlaps, while on three-fold overlaps the $1$-forms $\lambda_{\alpha\beta}\in\Omega^1(U_\alpha\cap U_\beta)$ by
\begin{equation}
\begin{aligned}
    \lambda_{\alpha\beta} + \lambda_{\beta\gamma} + \lambda_{\gamma\alpha} &= \mathbf{d}g_{\alpha\beta\gamma} + f_{\beta\alpha}(n_{\alpha\beta}+n_{\gamma\beta})f_{\beta\gamma} + \\
    &-\frac{1}{8}(f_{\beta\alpha}n_{\gamma\beta}\mathbf{d}f_{\beta\alpha} + f_{\beta\alpha}n_{\gamma\alpha}\mathbf{d}f_{\beta\gamma} + f_{\beta\gamma}n_{\beta\alpha}\mathbf{d}f_{\beta\gamma} + f_{\beta\gamma}n_{\gamma\alpha}\mathbf{d}f_{\beta\alpha})
\end{aligned}
\end{equation}
and on four-fold overlaps the local functions $g_{\alpha\beta\gamma}\in\Coo(U_\alpha\cap U_\beta \cap U_\gamma)$ are patched by
\begin{equation}
\begin{aligned}
    g_{\alpha\beta\gamma}-g_{\beta\gamma\delta}+g_{\gamma\delta\alpha}-g_{\delta\alpha\beta} = -\frac{1}{6}(f_{\delta\gamma}n_{\delta\beta}f_{\alpha\delta} - f_{\beta\gamma}n_{\alpha\delta}f_{\delta\gamma}+ f_{\beta\delta}n_{\gamma\delta}f_{\delta\alpha})
\end{aligned}
\end{equation}
It is clear that $(\lambda_{\alpha\beta},g_{\alpha\beta\gamma})$ is not the local data of a gerbe.
Finally, a global section $\Gamma(M_0,K)$ of the generalized correspondence space will be of the form $(\theta_\alpha,\tilde{\theta}_\alpha)$, patched as follows:
\begin{equation}
    \begin{aligned}
    \theta_\beta^{i} \,-\, \theta_\alpha^{i} \, &= \,f_{\beta\alpha}^i \\
    \tilde{\theta}_{\beta i} - \tilde{\theta}_{\alpha i}  &= \,\tilde{f}_{\beta\alpha i} - n_{\alpha\beta ij}\left(\theta_\beta^j-\frac{1}{2}f_{\beta\alpha}^j\right)
    \end{aligned}
\end{equation}
Hence we have an affine $T^{2n}$-bunde $K\rightarrow M_0$ and we find the conclusion of the theorem.
\end{proof}

\noindent Notice that in the particular case of a trivial class $[H^{(1)}]=0$ we have $n_{\alpha\beta}=0$ and hence we recover exactly the global geometric case discussed in the previous subsection.

\begin{remark}[$H$-flux and $F$-flux]\label{rem:hfflux}
The $H$-flux and the $F$-flux are still respectively the Dixmier-Douady class $[H]$ and the first Chern class $[F]$.
However in this case the $H$-flux compactification on $1$-cycles $[S^1_i]\in H_1(T^n,\mathbb{Z})$ cannot be seen as a first Chern class on the base manifold $M_0$ like in the previous subsection, but it is a cohomology class $[(H^{(2)},H^{(1)},0)]\in H^{2,1}_D(M_0,\,\wedge\,\mathbb{Z}^n)$. Notice the geometric $F$-flux can be also seen as a class $[(F,0,0)]\in H^{2,1}_D(M_0,\,\wedge\,\mathbb{Z}^n)$. Moreover the integration of the Dixmier-Douady class $[H]$ along $2$-cycles $[T^2_{ij}]\in H_2(T^n,\mathbb{Z})$ in the fiber is now non-trivial and it is given by the integral cohomology class $\big[H^{(1)}_{ij}\big]\in H^1(M_0,\mathbb{Z})$.
\end{remark}

\noindent Let us now generalize our discussion to the less simple case (but still abelian) of spacetime being a torus bundle with monodromy.

\begin{remark}[Torus bundle with monodromy]
Let us now generalize our torus bundle spacetime $M\xrightarrow{\pi}M_0$ to a torus bundle with monodromy given by the matrix $n_{\alpha\beta}^F\in\wedge^2\mathbb{Z}^n$. This can be seen as a cohomology class $[F^{(1)}]\in H^1(M_0,\mathbb{Z})$ such that we have the \v{C}ech cocycle
\begin{equation}
    \begin{aligned}
        F^{(1)}\, &=\, \mathbf{d}A^{(0)}_\alpha \\
        A^{(0)}_\beta - A^{(0)}_\alpha \,&=\, n^F_{\alpha\beta}
    \end{aligned}
\end{equation}
Hence we can update the $F$-flux in remark \ref{rem:hfflux} by adding the monodromy class to the curvature to obtain the class $[(F^{(2)},F^{(1)},0)]\in H^{2,1}_D(M_0,\,\wedge\,\mathbb{Z}^n)$.
Now, by looking at \eqref{eq:hflux}, we can generalize the patching conditions of the generalized correspondence space $K\rightarrow M_0$ by 
\begin{equation}\label{eq:monodromieshf}
\begin{pmatrix}
 \xi_\alpha \\
 \tilde{\xi}_\alpha
 \end{pmatrix} = \left(\begin{array}{@{}cc@{}}
 1+n_{\alpha\beta}^{F} &  0 \\ 
 n_{\alpha\beta}^{H} & 1-(n_{\alpha\beta}^{F})^{\mathrm{T}}
 \end{array}\right) \begin{pmatrix}
 \xi_\beta \\
 \tilde{\xi}_\beta
 \end{pmatrix}
\end{equation}
where $n_{\alpha\beta}^H$ is the monodromy matrix given by the $H$-flux $[H^{(1)}]\in H^1(M_0,\wedge^2\mathbb{Z}^n)$ and $n_{\alpha\beta}^F$ is the one given by the $F$-flux $[F^{(1)}]\in H^1(M_0,\wedge^2\mathbb{Z}^n)$. Hence the generalized correspondence space $K\rightarrow M_0$ will be an affine torus bundle patched by transition functions in the subgroup $\big(GL(n;\mathbb{Z})\ltimes\wedge^2\mathbb{Z}\big)\ltimes T^{2n}\subset \mathrm{Aff}(T^{2n})$ on overlaps $U_\alpha\cap U_\beta \subset M_0$, as found by \cite{Hul09}.
\end{remark}

\begin{remark}[T-duality $O(n,n;\mathbb{Z})$-action]\label{rem:tdualityoftfold}
Similarly to the previous subsection, we can still write our moduli fields of the Kaluza-Klein reduction by using the doubled metric as it follows
\begin{equation}
      \Theta_\alpha := \begin{pmatrix}
 \theta_\alpha \\
 \tilde{\theta}_\alpha
 \end{pmatrix}, \quad \mathcal{H}^{(0)}_\alpha := \begin{pmatrix}
 g^{(0)}_{ij} - B_{\alpha\, ik}^{(0)}g^{(0)kl}B_{\alpha\, lj}^{(0)} & B_{\alpha\, ik}^{(0)}g^{(0)kj} \\
 -g^{(0)ik}B_{\alpha\, kj}^{(0)} & g^{(0)ij}
 \end{pmatrix}, \quad \mathcal{A}_\alpha := \begin{pmatrix}
 A_\alpha \\
 B^{(1)}_\alpha
 \end{pmatrix}
\end{equation}
We still have a natural $O(n,n;\mathbb{Z})$ group action, whose elements $\mathcal{O}$ act by
\begin{equation}
      \Theta_\alpha\mapsto \mathcal{O}^{-1}\Theta_\alpha, \quad \mathcal{H}^{(0)}_\alpha\mapsto \mathcal{O}^{\mathrm{T}}\mathcal{H}^{(0)}_\alpha\mathcal{O}, \quad \mathcal{A}_\alpha\mapsto \mathcal{O}^{-1}\mathcal{A}_\alpha.
\end{equation}
However now the $O(n,n)$-moduli field $\mathcal{H}^{(0)}_\alpha$ is not globally defined on the base manifold $M_0$. This means that by applying a general $O(n,n;\mathbb{Z})$ transformation we obtain new differential data $\tilde{A}_\alpha$, $\tilde{B}_\alpha^{(0)}$, $\tilde{B}_\alpha^{(1)}$ which in general cannot be interpreted anymore as a $T^n$-bundle with gerbe connection. Only a transformation belonging to the \textit{geometric subgroup} $\mathcal{O}\in GL(d,\mathbb{Z})\ltimes\wedge^2\mathbb{Z}\subset O(n,n;\mathbb{Z})$ will send a background consisting of global $T^n$-bundle $M$ with gerbe connection to another one consisting of a $T^n$-bundle $\widetilde{M}$ with gerbe connection.
\end{remark}

\noindent Moreover notice that the transition functions \eqref{eq:monodromieshf} are not closed under $O(n,n;\mathbb{Z})$-action on the torus fiber, but only under its geometric subgroup. Hence if we want an interpretation for the non-geometric T-duals need to introduce the $Q$-flux which encodes T-folds.

\begin{remark}[$Q$-flux]
We can perform a general $O(n,n;\mathbb{Z})$ transformation of the transition functions \eqref{eq:hflux} fo the correspondence space $K\rightarrow M_0$ to obtain the following new ones
\begin{equation}
\begin{pmatrix}
 \xi_\alpha' \\
 \tilde{\xi}_\alpha'
 \end{pmatrix} = \left(\begin{array}{@{}cc@{}}
 1+n_{\alpha\beta}^{F\,\prime} &  n_{\alpha\beta}^{Q\,\prime} \\ 
 n_{\alpha\beta}^{H\,\prime} & 1-(n_{\alpha\beta}^{F\,\prime})^{\mathrm{T}}
 \end{array}\right) \begin{pmatrix}
 \xi_\beta' \\
 \tilde{\xi}_\beta'
 \end{pmatrix}
\end{equation}
Notice we have a new monodromy matrix $n^{Q\, ij}_{\alpha\beta}$, which patches the dual $\tilde{B}^{(0)ij}_\beta - \tilde{B}^{(0)ij}_\alpha = n^{Q\, ij}_{\alpha\beta}$
and hence it is a cohomology class.
Therefore the dual background has a new flux, which we call \textit{locally non-geometric flux} or just $Q$-flux, given by the cohomology class $\big[Q^{(1)ij}\big]\in H^{1}(M_0,\mathbb{Z})$. The \v{C}ech cocycle of the $Q$-flux is then by construction
\begin{equation}
    \begin{aligned}
        Q^{(1)}\, &=\, \mathbf{d}\tilde{B}^{(0)}_\alpha \\
        \tilde{B}^{(0)}_\beta - \tilde{B}^{(0)}_\alpha \,&=\, n^Q_{\alpha\beta}
    \end{aligned}
\end{equation}
where the moduli field $\tilde{B}^{(0)}_\alpha$ is the dual of the original moduli field ${B}^{(0)}_\alpha$.
\end{remark}

\noindent Let us resume everything in a familiar example. If we start from a background with only a $H$-flux $[H^{(1)}_{ij}]\in H^1(M_0,\mathbb{Z})$ and we perform a T-duality along the $i$-th circle $S^1_i$ of the torus fibre we get a background with $F$-flux $[F^{(1)\;\,j}_{\quad \,i}]\in H^1(M_0,\mathbb{Z})$ on the dual $[\tilde{S}^1_i]\in H_1(M_0,\mathbb{Z})$ circle. If now we perform another T-duality along the $j$-th circle $S^1_j$ we end up with non-trivial $Q$-flux $[Q^{(1)ij}]\in H^1(M_0,\mathbb{Z})$ on the dual torus $[\tilde{S}^1_i\times\tilde{S}^1_j]\in H_2(M_0,\mathbb{Z})$ of the fibre. This argument can be condensed in the following commuting diagram
\vspace{-1.5cm}\begin{equation}
    \begin{tikzcd}[row sep=13ex, column sep=10ex]
    \left[H^{(1)}_{ij}\right] \arrow[out=180,in=90,loop, leftrightarrow, "\mathcal{O}_B"] \arrow[r, leftrightarrow, "\mathcal{T}_i"] \arrow[d, leftrightarrow, "\mathcal{T}_j"] & \left[F^{(1)\;\,j}_{\quad \,i}\right] \arrow[d, leftrightarrow, "\mathcal{T}_j"']\arrow[out=90,in=0,loop, leftrightarrow, "\mathcal{O}_{N_i}"] \\
    \left[F^{(1)i}_{\quad\;\; j}\right] \arrow[out=270,in=180,loop, leftrightarrow, "\mathcal{O}_{N_j}"] \arrow[r, leftrightarrow, "\mathcal{T}_i"] & \left[Q^{(1)ij}\right] \arrow[out=270,in=360,loop, leftrightarrow, "\mathcal{O}_{\widetilde{B}}"']
\end{tikzcd}\vspace{-1.5cm}
\end{equation}
where $\mathcal{O}_{B}:=\left(\begin{smallmatrix*}1&0\\b&1\end{smallmatrix*}\right)\in O(n,n;\mathbb{Z})$ is any  $B$-shift, while $\mathcal{O}_{N_i} := \mathcal{T}_i^\mathrm{T}\mathcal{O}_B\mathcal{T}_i$, $\mathcal{O}_{N_j} := \mathcal{T}_j^\mathrm{T}\mathcal{O}_B\mathcal{T}_j$ and $\mathcal{O}_{\widetilde{B}} := (\mathcal{T}_j\mathcal{T}_i)^\mathrm{T}\mathcal{O}_B(\mathcal{T}_j\mathcal{T}_i)$ are transformations which preserve the respective fluxes.

\begin{remark}[Geometrical interpretation of fluxes]
We remark again that the $H$-flux, the geometric flux and the locally non geometric flux are not put on the doubled space by hand like in the usual approaches, but they are \textit{topological properties} of the doubled space $\mathcal{M}$ itself.
\end{remark}

\noindent Again we show it is possible to see $K$ as a global para-Hermitian bundle, which is effective in encoding just the data of the generalized correspondence space, but again this formalism cannot geometrize the whole gerbe. Hence from our formalism can derived an explanation of why para-Hermitian geometry is been effectively used to deal with DFT and T-duality.

\begin{remark}[The correspondence space is a global para-Hermitian bundle]
Since $\xi^i$ and $\Xi_i = \tilde{\xi}_{\alpha i} - B^{(0)}_{\alpha ij}\xi^j$ are globally defined $1$-forms on the generalized correspondence space $K$, we can define the same global $2$-form and global tensor of remark \ref{rem:parahermitianK}:
\begin{equation}\begin{aligned}
    \omega_K &:= \left(\tilde{\xi}_{\alpha i}-B^{(0)}_{\alpha ij}\xi^j\right)\wedge\xi^i ,\\
    \eta_K &:= \left(\tilde{\xi}_{\alpha i}-B^{(0)}_{\alpha ij}\xi^j\right)\odot\xi^i,
\end{aligned}\end{equation}
but also a para-complex structure
\begin{equation}
    J_K \,=\,   \frac{\partial}{\partial\tilde{\theta}_i}\otimes\left(\tilde{\xi}_i-B^{(0)}_{\alpha ij}\xi^j\right)  -   \left(\frac{\partial}{\partial \theta^i} - B^{(0)}_{\alpha ij}\frac{\partial}{\partial \tilde{\theta}_j}\right)\otimes\xi^i.
\end{equation}
The only difference with remark \ref{rem:globparaherm} is that $B^{(0)}_\alpha$ is not anymore a globally defined scalar. The triple $(K,J_K,\eta_K)$ is again a para-Hermitian structure which is the global non-trivially fibrated version of the one presented by \cite[p.~31]{MarSza18}.
\end{remark}

\subsection{Geometrization of general abelian T-duality}

Until now we investigated simple examples of T-dualities. In this section we want to give some insight of the general case by starting from a result by \cite{BelHulMin07}. In the previous sections we assumed the invariance of the gerbe connection under the principal torus action. This condition can be immediately relaxed by requiring just $\mathcal{L}_{\mathbf{k}_i}H=0$ and hence that the connection gauge transforms like $\mathcal{L}_{\mathbf{k}_i}B_\alpha = \mathrm{d}\eta_\alpha$ for some local $1$-form $\eta_\alpha$ under it.

\begin{theorem}[Gerbe with $T^n$-invariant curvature]\label{lemma:flatlie}
It was proven by \cite{BelHulMin07} that the principal $T^n$-action on $M$ can be lifted to a $T^n$-action on a gerbe with $T^n$-invariant curvature. The local data of this action on a good cover $\{U_\alpha\}$ for $M$ are given by a collection of $1$-forms $\eta_\alpha\in\Omega^1(U_\alpha,\mathbb{R}^n)$ on patches, of functions $\eta_{\alpha\beta}\in\Coo(U_\alpha\cap U_\beta,\mathbb{R}^n)$ on two-fold overlaps of patches and of constants $c_{\alpha\beta\gamma}\in\mathbb{R}^n$ on three-fold overlaps of patches which satisfy
\begin{equation}\label{eq:hullaction}
    \begin{aligned}
    \mathcal{L}_{\mathbf{k}_i}B_{\alpha} \,&=\, \mathrm{d}\eta_{\alpha\, i} \\
    \mathcal{L}_{\mathbf{k}_i}\Lambda_{\alpha\beta}  \,&=\, \eta_{\beta\, i} - \eta_{\alpha\, i} -\mathrm{d}\eta_{\alpha\beta\, i}\\
    \mathcal{L}_{\mathbf{k}_i}G_{\alpha\beta\gamma} \,&=\, \eta_{\alpha\beta\, i} + \eta_{\beta\gamma\, i} + \eta_{\gamma\alpha\, i} + c_{\alpha\beta\gamma\, i} \\
    & \quad\;\; c_{\alpha\beta\gamma\, i} - c_{\beta\gamma\delta\, i} + c_{\gamma\delta\alpha\, i} - c_{\delta\alpha\beta\, i} \,\in\, 2\pi\mathbb{Z}.
    \end{aligned}
\end{equation}
\end{theorem}
\begin{proof}
We can see $\mathcal{L}_{\mathbf{k}_i}H=0$ as the curvature of a flat gerbe on $M$ for any $i=1,\dots,n$. Therefore we obtain the gluing conditions for the local data $(\mathcal{L}_{\mathbf{k}_i}B_{\alpha},\,\mathcal{L}_{\mathbf{k}_i}\Lambda_{\alpha\beta},\,\mathcal{L}_{\mathbf{k}_i}G_{\alpha\beta\gamma})$ from \eqref{eq:flatgerbe}.
\end{proof}

\noindent Notice that the cohomology class $[c_{\alpha\beta\gamma}]\in H^2(M,T^n)$ can be interpreted as the flat holonomy class (definition \ref{eq:holonomyclass}) of this flat gerbe. 

\vspace{0.2cm}

\noindent The case in which the gerbe connection simply satisfies $\mathcal{L}_{\mathbf{k}_i}B_{\alpha}=0$ on each patch $U_\alpha$ is clearly a particular case of the general case \eqref{eq:hullaction}.

\begin{remark}[Underlying generalized vector]\label{rem:translation}
If the holonomy class $[c_{\alpha\beta\gamma}]$ is trivial, then the collection $(\mathbf{k}_i,\, \eta_{\alpha\beta i},\, \eta_{\alpha i})$ is exactly the \v{C}ech data of a section of the stack $T\mathcal{M}$ of the form \eqref{eq:localsymdata} from lemma \ref{thm:dtbild}. If we reparametrize the scalars by $\hat{\eta}_{\alpha\beta i} := \eta_{\alpha\beta i} - \iota_{\mathbf{k}_i}\Lambda_{\alpha\beta}$ according to lemma \ref{thm:dtbild} we get the local data of a global generalized vector $\mathbbvar{k}_i:=(\mathbf{k}_i+\eta_{\alpha i},\, \hat{\eta}_{\alpha\beta i})\in\mathfrak{at}(\mathcal{M})$.
\end{remark}

\noindent Notice this is an application where our global definition of doubled vector (see lemma \ref{thm:dtbild}), which is equipped with scalars $\hat{\eta}_{\alpha\beta}^i$ on two-fold overlaps of patches, is indispensable. Indeed global differential T-duality is formalized in terms of our generalized vectors $\mathbbvar{k}_i:=(\mathbf{k}_i+\eta_{\alpha i},\, \hat{\eta}_{\alpha\beta i})$ (see definition \ref{def:courantalg}), but not of the usual generalized vectors from Generalized Geometry.

\begin{definition}[Fundamental generalized vector]
The principal torus action on $M$ induces a Lie algebra homomorphism $\mathfrak{u}(1)^n \rightarrow\mathfrak{X}(M)$ which maps an element of the algebra to a fundamental vector $\mathbf{k}_i$. Thus this can be lifted to a Lie $2$-algebra homomorphism
\begin{equation}\label{eq:alghom}
    \mathfrak{u}(1)^n \longrightarrow \mathfrak{at}(\mathcal{M})
\end{equation}
which maps an element of the algebra in a generalized vector $\mathbbvar{k}_i:=(\mathbf{k}_i+\eta_{\alpha}^i,\, \hat{\eta}_{\alpha\beta}^i)$, where $\mathbf{k}_i$ is the fundamental vector of the action $\mathfrak{u}(1)^n \rightarrow\mathfrak{X}(M)$ and the local data $(\eta_{\alpha}^i,\, \hat{\eta}_{\alpha\beta}^i)$ are defined by conditions \eqref{eq:hullaction} with redefinition $\hat{\eta}_{\alpha\beta i} := \eta_{\alpha\beta i} - \iota_{\mathbf{k}_i}\Lambda_{\alpha\beta}$ of remark \ref{rem:translation}. 
\end{definition}

\noindent It is easy to check that $\llbracket\mathbbvar{k}_i,\mathbbvar{k}_j\rrbracket = 0$ for $i,j=1,\dots,n$.

\begin{definition}[Killing generalized vector]
A generalized vector $\mathbbvar{k}$ is Killing if $\llbracket \mathbbvar{k}, \mathbbvar{e}_\mathrm{I} \rrbracket =0$. We will use the symbol $\mathfrak{iso}(\mathcal{M},\mathcal{H})\subset\mathfrak{at}(\mathcal{M})$ for the sub-$2$-algebra of Killing generalized vectors.
\end{definition}

\begin{theorem}[General T-duality on the doubled space]\label{thm:generaldifftduality}
There exists a $T^n$-bundle $K\rightarrow M$ with first Chern class $\mathrm{c}_1(K)=[\pi_\ast H]=[H^{(2)}-\langle H^{(1)}\wedge\xi\rangle]\in H^2(M,\mathbb{Z})^n$ if and only if the fundamental generalized vectors $\{\mathbbvar{k}_i\}$ of the principal $T^n$-action on $M$ are Killing.
\end{theorem}

\begin{proof}
From lemma \ref{lemma:flatlie} we get the following patching conditions for the closed form $\widetilde{F}_i := \iota_{\mathbf{k}_i}H$
\begin{equation}\label{eq:torusinthecore}
    \begin{aligned}
        \widetilde{F}_i \,&=\, \mathrm{d}(\eta_{\alpha\,i} - \iota_{\mathbf{k}_i}B_{\alpha}) \\
        (\eta_{\beta\,i} - \iota_{\mathbf{k}_i}B_{\beta}) - (\eta_{\alpha\,i} - \iota_{\mathbf{k}_i}B_{\alpha}) \,&=\, \mathrm{d}(\iota_{\mathbf{k}_i}\Lambda_{\alpha\beta}+\eta_{\alpha\beta\,i}) \\
        (\iota_{\mathbf{k}_i}\Lambda_{\alpha\beta}+\eta_{\alpha\beta\,i}) + (\iota_{\mathbf{k}_i}\Lambda_{\beta\gamma}+\eta_{\beta\gamma\,i}) + (\iota_{\mathbf{k}_i}\Lambda_{\gamma\alpha}+\eta_{\gamma\alpha\,i}) \,&=\, c_{\alpha\beta\gamma\, i}
    \end{aligned}
\end{equation}
These are precisely the \v{C}ech data for a principal $T^n$-bundle on $M$ if and only if $c_{\alpha\beta\gamma \,i}\in 2\pi\mathbb{Z}$.
If the fundamental generalized vectors $\{\mathbbvar{k}_i\}$ are Killing, then ordinary vectors $\{\mathbf{k}_i\}$ are Killing respect to the ordinary metric $g$, i.e. $\mathcal{L}_{\mathbf{k}_i}g=0$. Also the flat gerbes defined by $\mathcal{L}_{\mathbf{k}_i}H=0$ in lemma \ref{lemma:flatlie} are trivial, which is equivalent to the fact that the flat holonomy classes $[c_{\alpha\beta\gamma i}]$ are trivial. This means that the local data \eqref{eq:torusinthecore} define a $T^n$-bundle and hence the conclusion.
\end{proof}

\begin{remark}[Generalized correspondence space]\label{rem:gencorrspacegen}
From the proof of lemma \ref{thm:generaldifftduality} we obtain the \v{C}ech data of a principal $T^n$-bundle $K\rightarrow M$, which we call \textit{generalized correspondence space}. This has first Chern class $\big[\widetilde{F}_i\big]=\big[\iota_{\mathbf{k}_i}H\big]$ and global connection $1$-form $\Xi\in\Omega^1(K,\mathbb{R}^n)$ given by
\begin{equation}
\begin{aligned}
    \Xi_i \,:=\, \mathrm{d}\tilde{\theta}_{\alpha i} + \eta_{\alpha i}^{(1)} + B^{(1)}_{\alpha i} + (\eta_{\alpha ij}^{(0)} - B^{(0)}_{\alpha ij})\xi^j
\end{aligned}
\end{equation}
where we called its local fiber coordinates $(\tilde{\theta}_\alpha)$ and we split $\eta_{\alpha i}=\eta_{\alpha i}^{(1)}+ \eta_{\alpha ij}^{(0)}\xi^j$ in horizontal and vertical components. Again the picture \eqref{eq:tfolddiagram} holds:
\begin{equation}
\begin{tikzcd}[row sep=4.5ex, column sep=2.8ex]
 & & K \arrow[dl, two heads]\arrow[dr, dotted] & \\
 T^n \arrow[r, hook] & M \arrow[dr, "\pi"', two heads] & & \text{T-fold}\arrow[dl, dotted] \\
 & & M_0 &
\end{tikzcd}
\end{equation}
\end{remark}

\noindent Recall that generalized vectors are infinitesimal automorphisms of the doubled space. Notice that Killing generalized vectors $\mathfrak{iso}(\mathcal{M},\mathcal{H})$ are infinitesimal isometries of the doubled space and therefore they can be integrated to finite isometries of the doubled metric structure (remark \ref{rem:isometry}), i.e. to elements of $\mathbf{Iso}(\mathcal{M},\mathcal{H})=\mathrm{Iso}(M,g)\ltimes\mathbf{H}(M,\BU)$. By integrating lemma \ref{thm:generaldifftduality} we get the following statement.

\begin{theorem}[General T-duality on the doubled space, finite formulation]
There exists a $T^n$-bundle $K\rightarrow M$ with first Chern class $\mathrm{c}_1(K)=[\pi_\ast H]=[H^{(2)}-\langle H^{(1)}\wedge\xi\rangle]\in H^2(M,\mathbb{Z})^n$ if and only if any automorphism in $\Coo(M_0,\,T^n)$ of spacetime $M\xrightarrow{\pi}M_0$ is lifted to an isometry in $\mathbf{Iso}(\mathcal{M},\mathcal{H})$ of the doubled space.
\end{theorem}

\noindent More generally this suggests that the \textit{presence of finite isometries of the doubled space implies T-duality}. Therefore again the generalized correspondence space $K$ is inside the total space of the doubled space, and it is well-defined whenever the metric doubled space has an isometry.

\begin{figure}[ht]\centering
\includegraphics[scale=0.18]{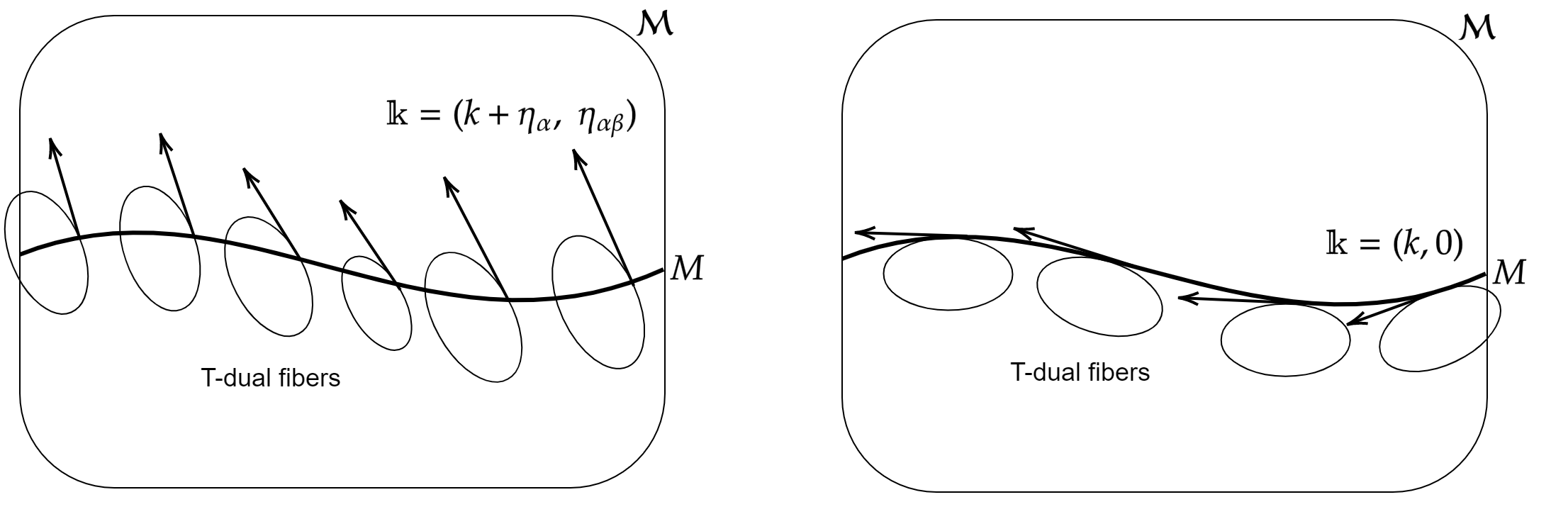}
\caption{Generalized correspondence space $K$ is defined by Killing generalized vectors.}\label{FIGtduality}\end{figure}

\begin{remark}[Physical interpretation of general T-duality conditions]
The holonomy of the gerbe is nothing but a global expression for the Wess-Zumino-Witten action of a world-sheet
\begin{equation}
    \exp 2\pi i S_{\mathrm{WZ}}(\Sigma) := \mathrm{Hol}_{(B_\alpha,\Lambda_{\alpha\beta},G_{\alpha\beta\gamma})}(\Sigma).
\end{equation}
Then the holonomy of the gerbe of lemma \ref{lemma:flatlie} is the global variation of the Wess-Zumino term
\begin{equation}
    \exp 2\pi i \delta_j S_{\mathrm{WZ}}(\Sigma) := \mathrm{Hol}_{(\mathcal{L}_{\mathbf{k}_j}\!B_\alpha,\, \mathcal{L}_{\mathbf{k}_j}\!\Lambda_{\alpha\beta},\, \mathcal{L}_{\mathbf{k}_j}\!G_{\alpha\beta\gamma})}(\Sigma).
\end{equation}
This needs to vanish for any closed surface $\Sigma\subset M$ to make the background T-dualizable.
\end{remark}

\noindent Again, even if para-Hermitian geometry cannot geometrize the whole gerbe data, it is enough to encode the data of the generalized correspondence space. Similarly to the previous subsection the para-Hermitian $T^{2n}$-fiber is non-trivially fibrated over the base manifold $M$.

\begin{remark}[The generalized correspondence space is a global para-Hermitian bundle]\label{rem:globparaherm}
We can use again the fact that the connection $\Xi\in\Omega^1(K,\mathbb{R}^n)$ is a global $1$-form on the generalized correspondence $K\xrightarrow{}M_0$ to define the following global tensors
\begin{equation}
    \begin{aligned}
    \omega_K &= \Xi_i \wedge \xi^i \\
    \eta_K &= \Xi_i \odot \xi^i 
    \end{aligned}
\end{equation}
and also the para-complex structure
\begin{equation}
    J_K \,=\,  \frac{\partial}{\partial\tilde{\theta}_i}\otimes\Xi_i -  \bigg(\mathbf{k}_i + (\eta_{\alpha ij}^{(0)} - B^{(0)}_{\alpha ij}) \frac{\partial}{\partial\tilde{\theta}_j}\bigg) \otimes \xi^i 
\end{equation}
It is easy to check that the correspondence space $K$ comes with a para-Hermitian structure $(K,J_K,\eta_K)$.
\end{remark}

\noindent After digression \ref{dig:preq} we will now highlight another analogy between doubled geometry and Higher Prequantum Geometry, this time regarding the conditions for a general T-duality.

\begin{digression}[A formal similarity with Higher Prequantum Geometry]
In \cite{Rog11} and \cite{Rog13} Higher Prequantization of a $2$-plectic manifold $(M,\omega)$ is presented, where the $2$-plectic form $\omega\in\Omega^3_{\mathrm{cl}}(M)$ is non-degenerate and closed. In the references the $2$-algebra of Hamiltonian forms $\Coo(U)\xrightarrow{\mathrm{d}}\Omega^1_{\mathrm{Ham}}(U)$ on a patch $U\subset M$ is defined. Notice this is a sub-$2$-algebra of $\mathbf{H}(U,\BU)\simeq\big(\Coo(U)\xrightarrow{\mathrm{d}}\Omega^1(U)\big)$. Hence the stackification on $M$ of this $2$-algebra of Hamiltonian forms will be a sub-$2$-algebra of $\mathbf{H}(M,\BU)$ of circle bundles $L$ satisfying
\begin{equation}\label{eq:higherprequantum}
    \iota_{X_L}\omega = \mathrm{curv}(L)
\end{equation}
for some vector $X_L\in\mathfrak{X}(M)$ that we will call Hamiltonian vector field of $L$. If we interpret our background $(M,H)$ as a $2$-plectic manifold we can see that the T-duality condition
\begin{equation}
    \iota_{\mathbf{k}_i}H = \mathrm{curv}(K)_i
\end{equation}
is formally identical to \eqref{eq:higherprequantum}, where $\mathrm{curv}(K)_i=\widetilde{F}_i$ is the curvature of the generalized correspondence space $K\rightarrow M$ from remark \ref{rem:gencorrspacegen}. Therefore we can reformulate the conditions for T-duality in the language of Higher Prequantum Geometry as it follows: to have T-duality the fundamental vector fields $\{\mathbf{k}_i\}$ of the bundle $M\xrightarrow{\pi}M_0$ must be both Killing and Hamiltonian.
\end{digression}

\subsection{Geometrization of non-abelian T-duality}

In this subsection we will briefly deal with non-abelian T-duality, to show that it is encompassed in our formalism. This means that in the following discussion we can drop at once the assumptions that $[H^{(0)}]=0$ and that spacetime is an abelian principal bundle.

\vspace{0.2cm}
\noindent Let us assume that the spacetime $M$ is a principal $G$-bundle over a smooth base manifold, i.e.
\begin{equation}
\begin{tikzcd}[row sep=5ex, column sep=4.5ex]
   G \arrow[r,hook] & M \arrow[d,two heads, "\pi"] \\
   & M_0
\end{tikzcd}
\end{equation}
The doubled space $\M\xrightarrow{\bbpi}M$ must now be reduced by a composition of an Higher Kaluza-Klein reduction from $\M$ to its base $M$ and an ordinary non-abelian Kaluza-Klein reduction from $M$ to its base $M_0$. See the following pullback diagram for the total reduction
\begin{equation}
\begin{tikzcd}[row sep=11ex, column sep=8ex]
\M \arrow[r]\arrow[d, "\bbpi"] & \ast \arrow[d] & \\
M \arrow[r, "(\Lambda\text{,}G)"]\arrow[d, "\pi"] & \mathbf{B}(\mathbf{B}U(1)_{\mathrm{conn}}) \arrow[r]\arrow[d] & \ast \arrow[d]\\
M_0 \arrow[r, "?"]\arrow[rr, "f", bend right=30]& \left[G,\mathbf{B}(\mathbf{B}U(1)_{\mathrm{conn}})\right]/G \arrow[r] & \mathbf{B}G
\end{tikzcd}
\end{equation}
The curvature $3$-form of the doubled space can as usual be expanded in the connection $1$-form $\xi\in\Omega^1(M,\mathfrak{g})$ of the $G$-bundle $M$ by
\begin{equation}
    H = H^{(3)} + H^{(2)}_{i}\wedge\xi^i + \frac{1}{2}H^{(1)}_{ij}\wedge\xi^i\wedge\xi^j + \frac{1}{6}H^{(0)}_{ijk}\xi^i\wedge\xi^j\wedge\xi^k
\end{equation}
Now we can define the usual $2$-form dual curvature by $\widetilde{F}_i := \iota_{\mathbf{k}_i}H$. Therefore we obtain the following $2$-form curvature on spacetime $M$
\begin{equation}
\begin{aligned}
    F^i &= \mathrm{d}\xi^i &+&\, \frac{1}{2} \,F^{(0)\, i}_{\;\;jk}\xi^j\wedge\xi^k \\
    \widetilde{F}_i &= \big(H^{(2)}_{i} - H^{(1)}_{ij}\wedge\xi^j\big) \!\!\!\!&+&\, \frac{1}{2} \, H^{(0)}_{ijk}\,\xi^j\wedge\xi^k
\end{aligned}
\end{equation}
We assume that there is a non-abelian group of isometries of the doubled metric space $\mathbf{Iso}(\M,\mathcal{H})=\Gamma\big(M_0,\,\mathrm{Ad}(M)\big)$ given by the group of automorphisms of the $G$-bundle $M\xrightarrow{\pi} M_0$. We will see that in this case the generalized correspondence space will be a $T^n$-bundle $K\rightarrow M$ on spacetime, where we defined $n:=\mathrm{dim}\,G$.
Let us expand the differential data of the doubled space
\begin{equation}\label{eq:nonabelianb}
    \begin{aligned}
    B_\alpha &= B^{(2)} + \langle B^{(1)}\wedge\xi \rangle + \frac{1}{2}\langle B^{(0)}_{ij},\xi\wedge\xi\rangle \\
    \mathrm{d}\Lambda_{\alpha\beta} &= \mathbf{d}\Lambda_{\alpha\beta}^{(1)} + \langle\Lambda_{\alpha\beta}^{(0)},F\rangle + \langle\mathbf{d}\Lambda_{\alpha\beta}^{(0)},\xi\rangle -\frac{1}{2}\langle \Lambda_{\alpha\beta}^{(0)}, F^{(0)\, }_{\;\;jk}\xi^j\wedge\xi^k \rangle
    \end{aligned}
\end{equation}
where a new final vertical term appears respect to the abelian case.
Now, to hugely simplify the discussion, let us assume that $(\Lambda_{\alpha\beta},\,G_{\alpha\beta\gamma}):M\rightarrow\mathbf{B}(\BU)$ is an equivariant structure under the principal $G$-action on $M$. The patching condition $B_\beta -B_\alpha = \mathrm{d}\Lambda_{\alpha\beta}$ becomes
\begin{equation}
    \begin{aligned}
        B^{(0)}_{\beta ij} - B^{(0)}_{\alpha ij}\, &\,=\, -\varepsilon_{ij}^{\;\;\,k}\tilde{f}_{\alpha\beta k} , \\
        B^{(1)}_{\beta i} - B^{(1)}_{\alpha i} &\,=\,  \mathbf{d}\tilde{f}_{\alpha\beta i} ,\\
        B^{(2)}_\beta - B^{(2)}_\alpha\,\, &\,=\, \mathbf{d}\lambda_{\alpha\beta} + \langle \tilde{f}_{\alpha\beta},F \rangle ,
    \end{aligned}
\end{equation}
where $F$ is the curvature of the principal $G$-bundle and where we called $\Lambda_{\alpha\beta}^{(0)}=:\tilde{f}_{\alpha\beta}$ in analogy with the abelian case. On three-fold overlaps we have the simple patching conditions
\begin{equation}
    \begin{aligned}
    \tilde{f}_{\alpha\beta}+\tilde{f}_{\beta\gamma}+\tilde{f}_{\gamma\alpha} &=0 \\
    \lambda_{\alpha\beta}+\lambda_{\beta\gamma}+\lambda_{\gamma\alpha} &=\mathbf{d}g_{\alpha\beta\gamma}
    \end{aligned}
\end{equation}
We can define the $1$-forms $\widetilde{A}_{\alpha}:=-\iota_{\mathbf{k}_i}B_\alpha$ on spacetime, which must be patched by \eqref{eq:nonabelianb} as
\begin{equation}
    \widetilde{A}_{\beta i} - \widetilde{A}_{\alpha i} = \mathbf{d}\Lambda_{\alpha\beta i}^{(0)} + [\xi,\Lambda_{\alpha\beta}^{(0)}]_i - 2\mathcal{L}_{\mathbf{k}_i}\Lambda_{\alpha\beta j}^{(0)}\xi^j -\mathcal{L}_{\mathbf{k}_i}\Lambda_{\alpha\beta}^{(1)}
\end{equation}
In the simple case where the doubled space is equivariant under the principal $G$-action we only have $\widetilde{A}_{\beta} - \widetilde{A}_{\alpha} = \mathbf{d}\Lambda_{\alpha\beta}^{(0)} + [\xi,\Lambda_{\alpha\beta}^{(0)}]$, which can be written as
\begin{equation}
    \widetilde{A}_{\beta} - \widetilde{A}_{\alpha} = \mathrm{D}_\xi\Lambda_{\alpha\beta}^{(0)}
\end{equation}
where $\mathrm{D}_\xi$ is the covariant derivative defined by the connection $\xi\in\Omega^1(M,\mathfrak{g})$.
The principal connection $\Xi\in\Omega^1(K,\mathbb{R}^3)$ on the non-abelian generalized correspondence space $K\rightarrow M$, seen as a torus $T^n$-bundle on spacetime $M$, will be the usual $1$-form
\begin{equation}
    \Xi \,=\, \mathrm{d}\tilde{\theta}_\alpha + \widetilde{A}_{\alpha} \,=\, \tilde{\xi}_\alpha - B^{(0)}_{\alpha ij}\xi^j
\end{equation}
where we called $\tilde{\xi}_\alpha := \mathrm{d}\tilde{\theta}_\alpha + B^{(1)}_\alpha$.
According to \eqref{eq:exactness} the $[H^{(0)}]$ class of the $H$-flux satisfies
\begin{equation}
    H^{(0)}_{ijk} \,=\, \mathcal{L}_{\mathbf{k}_{[i}} B^{(0)}_{\alpha ij]}
\end{equation}
which can be solved by imposing
\begin{equation}
    B^{(0)}_{\alpha ij} = b_{\alpha ij} + H_{ijk}^{(0)}\theta^k_\alpha + F^{(0)\, k}_{\;\;\,ij}\tilde{\theta}_{\alpha k}
\end{equation}
where $b_{\alpha ij}$ are scalars depending only on the base manifold $U_\alpha$, where $\theta_\alpha$ is a local coordinate of the fiber $G$ and where $\tilde{\theta}_{\alpha}$ must satisfy $\tilde{\theta}_{\alpha}-\tilde{\theta}_{\beta} = \tilde{f}_{\alpha\beta}$. Hence $\tilde{\theta}_\alpha$ are the coordinates of the linear $\mathbb{R}^{n}$ fiber of the local principal bundles $U_\alpha \times \mathbb{R}^{n}$ with local connection $\tilde{\xi}_\alpha := \mathrm{d}\tilde{\theta}_\alpha + B^{(1)}_\alpha$. 
Hence we can define a globally defined $1$-form connection $\Xi\in\Omega^1(K,\mathbb{R}^n)$ on the generalized correspondence space given by
\begin{equation}
    \Xi_i \,=\, \tilde{\xi}_{\alpha i} - \big(H_{ijk}^{(0)}\theta^k_\alpha+F_{\;\;\, ij}^{(0)\,k}\tilde{\theta}_{\alpha k}\big)\xi^j.
\end{equation}

\begin{remark}[Non-abelian T-fold]
The concept of non-abelian T-fold was recently introduced for $S^3$ in \cite{Bug19}. We will directly generalize this idea to a $G$-bundle spacetime, in analogy with the abelian T-fold. The non-abelian T-dual space will necessarily have a non-trivial locally non-geometric $Q$-flux, which implies it will be a T-fold. We then have a picture
\begin{equation}
\begin{tikzcd}[row sep=4.5ex, column sep=5ex]
 & & K \arrow[dl, two heads]\arrow[dr, dotted, end anchor={[yshift=-2ex]}] & \\
 G \arrow[r, hook] & M \arrow[dr, "\pi"', two heads] & & \begin{tabular}{c}non-abelian \\T-fold\end{tabular} \arrow[dl, dotted, start anchor={[yshift=2ex]}] \\
 & & M_0 &
\end{tikzcd}
\end{equation}
Here again, like for abelian T-fold, the arrows on the right side are not actual maps between actual spaces, but they are only indicative. This is because the non-abelian T-fold, similarly to its abelian version, can only be geometrized inside the non-principal $G\ltimes T^n$-bundle $K\rightarrow M_0$.
\end{remark}

\begin{remark}[Recovering the traditional formulation of non-abelian T-duality]\label{rem:natdtraditional}
We use the basis $\xi^I_\alpha:=(\xi^i,\,\tilde{\xi}_{\alpha i})$ of local connections of respectively $U_\alpha\times G$ and $U_\alpha\times \mathbb{R}^n$. In this basis we must then write the moduli field $\mathcal{H}_\alpha^{(0)}$ of the doubled metric by using the moduli field of the Kalb-Ramond field $B^{(0)\prime}_{\alpha ij} := b_{\alpha ij} + H^{(0)}_{ijk}\theta^k_{\alpha}$. Now we can express non-abelian T-duality as the following $O(n,n)$-transformation
\begin{equation}
    \mathcal{T}_{\mathrm{NATD}} :=
    \begin{pmatrix}
 0 & 1 \\
 1 & F_{\;\;\, ij}^{(0)\,k}\tilde{\theta}_{\alpha k}
 \end{pmatrix}
\end{equation}
which encodes the shift $B^{(0)\prime}_{\alpha ij} \mapsto B^{(0)}_{\alpha ij}=B^{(0)\prime}_{\alpha ij}+F_{\;\;\, ij}^{(0)\,k}\tilde{\theta}_{\alpha k}$. Therefore the non-abelian T-dual of the moduli field of the doubled metric in this basis will be $\widetilde{\mathcal{H}}_\alpha^{(0)} = \mathcal{T}_{\mathrm{NATD}}^{\mathrm{T}}\mathcal{H}_\alpha^{(0)}\mathcal{T}_{\mathrm{NATD}}$. Its local components $\tilde{g}^{(0)}_\alpha$ and $\tilde{B}^{(0)}_\alpha$ will be the local data of the non-abelian T-fold.
\end{remark}

\noindent We will now explore a simple example of background which has non-abelian T-duality.

\begin{example}[$3$-sphere bundle]
Let us consider an $S^3$-bundle spacetime given by the sequence of manifolds $S^3\hookrightarrow M \twoheadrightarrow M_0$. Recall that $H^3(S^3,\mathbb{Z})\cong\mathbb{Z}$, so $[H^{(0)}]$ is determined by a single integer. For simplicity let us assume the doubled space is trivial with Dixmier-Douady class $[H]=0$. Even if the gerbe structure is trivial on the $3$-sphere fiber, in contrast with the abelian case, the moduli fields $B_\alpha^{(0)}$ are not globally defined scalars. $B_{\beta ij}^{(0)}-B_{\alpha ij}^{(0)} = \varepsilon_{ij}^{\;\;\, k} (\tilde{\theta}_{\beta k}-\tilde{\theta}_{\alpha k})$ 
Thus the connections will be patched on two-fold overlaps by
\begin{equation}
    \begin{pmatrix}
 \xi_\alpha \\
 \tilde{\xi}_\alpha
 \end{pmatrix} =
    \begin{pmatrix}
 1 & 0 \\
 \varepsilon_{ij}^{\;\;\, k}\tilde{f}_{\alpha\beta k} & 1
 \end{pmatrix}    \begin{pmatrix}
 \xi_\beta \\
 \tilde{\xi}_\beta
 \end{pmatrix}
\end{equation}
where $\tilde{f}_{\alpha\beta}=\tilde{\theta}_{\alpha}-\tilde{\theta}_{\beta}$ are the transition functions for the dual torus coordinates. Therefore the monodromy matrix of these coordinates will be given by
\begin{equation}
    n_{\alpha\beta}^H = \begin{pmatrix}
 0 & \tilde{f}_{\alpha\beta 3} & -\tilde{f}_{\alpha\beta 2} \\
 -\tilde{f}_{\alpha\beta 3} & 0 & \tilde{f}_{\alpha\beta 1} \\
 \tilde{f}_{\alpha\beta 2} & -\tilde{f}_{\alpha\beta 1} & 0
 \end{pmatrix}.
\end{equation}
Hence actually the $\mathbb{R}^3$ fibers are not non-compact, but they are glued as a non-abelian T-fold. Differently from its abelian counterpart, notice that in the non-abelian T-fold the monodromy matrix is in general not constant.
Now let us write the locally defined moduli field $\mathcal{H}^{(0)}_\alpha$ of the reduction of the doubled metric to the base manifold $M_0$.
We will suppress few indices, by defining $b := \star_{S^3}B^{(0)\prime}_\alpha$ (with $B^{(0)\prime}_\alpha$ defined like in remark \ref{rem:natdtraditional}) so we can write
\begin{equation}
    \mathcal{H}^{(0)}_\alpha \,=\, 
    \left(\begin{array}{@{}ccc|ccc@{}}
    1+b_2^2+b_3^2&-b_1b_2&-b_1b_3&0&-b_3&b_2\\
    -b_1b_2&1+b_1^2+b_3^2&-b_2b_3&b_3&0&-b_1\\
    -b_1b_3&-b_2b_3&1+b_1^2+b_2^2&-b_2&b_1&0\\\hline
    0&b_3&-b_2&1&0&0\\
    -b_3&0&b_1&0&1&0\\
    b_2&-b_1&0&0&0&1
    \end{array}\right)
\end{equation}
where the metric is just the round metric of the $3$-sphere.
Let us call $B_i := b_i + \tilde{\theta}_i$ the Hodge dual of the full $B^{(0)}_\alpha$ moduli field. Now we can perform the T-duality transformation $\mathcal{T}_{S^3}=\left(\begin{smallmatrix*}0&1\\1&0\end{smallmatrix*}\right)$ to obtain the non-abelian T-dual doubled metric moduli field
\begin{equation}
    \widetilde{\mathcal{H}}^{(0)}_\alpha \,=\, 
   \left(\begin{array}{@{}ccc|ccc@{}}
     1&0&0&0&B_3&-B_2\\
    0&1&0&-B_3&0&B_1\\
    0&0&1&B_2&-B_1&0 \\\hline
    0&-B_3&B_2&1+B_2^2+B_3^2&-B_1B_2&-B_1B_3\\
    B_3&0&-B_1&-B_1B_2&1+B_1^2+B_3^2&-B_2B_3\\
    -B_2&B_1&0&-B_1B_3&-B_2B_3&1+B_1^2+B_2^2
    \end{array}\right)
\end{equation}
Thus the non-abelian T-dual background takes the following familiar expression
\begin{equation}
    \begin{aligned}
    \tilde{g}^{(0)}_\alpha &\,=\, \frac{1}{1+B_1^2+B^2_2+B_3^2} \begin{pmatrix}
 1+B_1^2 & B_1B_2 & B_1B_3 \\
 B_1B_2 & 1+B_2^2 & B_2B_3\\
 B_1B_3 & B_2B_3 & 1+B_3^2
 \end{pmatrix}, \\
    \tilde{B}^{(0)}_\alpha &\,=\, \frac{1}{1+B_1^2+B^2_2+B_3^2} \begin{pmatrix}
 0 & -B_3 & B_2 \\
 B_3 & 0 & -B_1\\
 -B_2 & B_1 & 0
 \end{pmatrix}.
        \end{aligned}
\end{equation}
The new local metric and Kalb-Ramond field will be respectively $ \tilde{g}^{(0)}_\alpha =  \tilde{g}^{(0)ij}_\alpha\tilde{\xi}_{\alpha i}\odot\tilde{\xi}_{\alpha j}$ and $ \tilde{B}^{(0)}_\alpha =  \tilde{B}^{(0)ij}_\alpha\tilde{\xi}_{\alpha i}\wedge\tilde{\xi}_{\alpha j}$. These are the differential data of the fibers of our non-abelian T-fold and indeed they cannot be patched globally on the base manifold.
\end{example}

\begin{example}[Twisted torus bundle]
In the case of the twisted torus bundle $G\hookrightarrow M \twoheadrightarrow M_0$, where $G$ is a twisted torus with $\mathrm{dim}\,G>2$, we recover the general commuting diagram of T-dualities by \cite{Bug19} along any circle of the fiber, i.e.
\begin{equation}
   \begin{tikzcd}[row sep=38,column sep=28]
    & H^{(0)}_{ijk}  \arrow[rr, "\mathcal{T}_i"] \arrow[dd, "\mathcal{T}_j" near start] \arrow[dl, "\mathcal{T}_k"'] & &   F^{(0)i}_{\quad\;\, jk}  \arrow[dd, "\mathcal{T}_j"] \arrow[dl, "\mathcal{T}_k"] \\
    F^{(0)\;\;k}_{\quad ij} \ar[crossing over, "\mathcal{T}_i" near start]{rr} \arrow[dd, "\mathcal{T}_j"'] & & Q^{(0)i\; k}_{\quad\; j} \\
      &  F^{(0)\;j}_{\quad\; i\; k}  \arrow[rr, "\mathcal{T}_i" near start] \arrow[dl, "\mathcal{T}_k"'] & &  Q^{(0)ij}_{\quad\;\; k}  \arrow[dl, "\mathcal{T}_k"] \\
    Q^{(0)\;jk}_{\quad i} \arrow[rr, "\mathcal{T}_i"'] && R^{(0)ijk} \arrow[from=uu,crossing over, "\mathcal{T}_j" near start]
 \end{tikzcd} 
\end{equation}
\end{example}

\noindent Let us now give a quick final look to the general case.

\begin{remark}[General non-abelian case]
In the general case, where we require just the local $\mathcal{L}_{\mathbf{k}_i}B_\alpha = \mathrm{d}\eta^i_{\alpha}$, analogously to its abelian version, we obtain a connection $1$-form $\Xi\in\Omega^1(K,\mathbb{R}^n)$ of a principal $T^n$-bundle over spacetime $M$ by
\begin{equation}
\begin{aligned}
    \Xi_i \,=\, \mathrm{d}\tilde{\theta}_{\alpha i} + \eta_{\alpha i}^{(1)} + B^{(1)}_{\alpha i} + \big(\eta_{\alpha ij}^{(0)} - B^{(0)}_{\alpha ij}\big) \xi^j
\end{aligned}
\end{equation}
where $\eta_{\alpha i} = \eta_{\alpha i}^{(1)} + \eta_{\alpha ij}^{(0)}\xi^j$ is split in horizontal and vertical part.
Again the generalized vectors $\mathbbvar{k}_i:=(\mathbf{k}_i+\eta_{\alpha i},\,\eta_{\alpha\beta i})$ are Killing and hence therefore can be integrated to finite isometries of the doubled space in $\Gamma\big(M_0, \, \mathrm{Ad}(M) \big)\ltimes\mathbf{H}(M,\BU)\simeq\mathbf{Iso}(\M,\mathcal{H})$.
\end{remark}

\subsection{Geometrization of Poisson-Lie T-duality}

In this subsection we will briefly deal with Poisson-Lie T-duality (see \cite{Has17}) in our framework.

\vspace{0.2cm}
\noindent Let spacetime be a principal $G$-bundle $M\xrightarrow{\pi}M_0$. Let us assume that the fundamental vector fields $\{\mathbf{k}_i\}\subset \mathfrak{X}(M)$ are Killing, i.e. $\mathcal{L}_{\mathbf{k}_i}g=0$, but such that the gerbe connection satisfies
\begin{equation}\label{eq:plcond}
    \mathcal{L}_{\mathbf{k}_i}B_\alpha = \big[\widetilde{A}_\alpha\wedge\widetilde{A}_\alpha\big]_i^{\tilde{\mathfrak{g}}}
\end{equation}
where $\widetilde{A}_{\alpha i}:=-\iota_{\mathbf{k}_i}B_\alpha$ is a local $\mathfrak{g}^\ast$-valued $1$-form on the total spacetime $M$ and $[-,-]^{\tilde{\mathfrak{g}}}$ is the commutator of some Lie algebra $\tilde{\mathfrak{g}}$ whose underlying vector space is $\mathfrak{g}^\ast$. Notice that this implies that the automorphisms fo the $G$-bundle are not lifted to isometries of the doubled space, i.e. $\Gamma\big(M_0,\,\mathrm{Ad}(M)\big)\not\subset\mathbf{Iso}(\M,\mathcal{H})$. However these transformations defy the isometry in a very controlled way. The \eqref{eq:plcond} implies that $\widetilde{F}_i:=\iota_{\mathbf{k}_i}H$ can be seen as the curvature of a principal $\widetilde{G}$-bundle with $\mathrm{Lie}(\widetilde{G})=\tilde{\mathfrak{g}}$, indeed we have
\begin{equation}
    \begin{aligned}
     \widetilde{F}_i&\,:=\,\iota_{\mathbf{k}_i}H\\
     &\;\,=\,\mathcal{L}_{\mathbf{k}_i}B_\alpha - \mathrm{d}\iota_{\mathbf{k}_i}B_\alpha \\
     &\;\,=\, \big[\widetilde{A}_\alpha\wedge\widetilde{A}_\alpha\big]_i^{\tilde{\mathfrak{g}}} + \mathrm{d}\widetilde{A}_{\alpha i}
    \end{aligned}
\end{equation}
Hence the generalized correspondence space $K\rightarrow M$ will be a principal $\widetilde{G}$-bundle on spacetime $M$ with curvature $\widetilde{F}_i$. Notice we have $\widetilde{A}_{\alpha i}: = \iota_{\mathbf{k}_i}B_\alpha = B^{(1)}_{\alpha i} - B^{(0)}_{\alpha ij}\xi^j$ where $\xi\in\Omega^1(M,\mathfrak{g})$ is the connection of the $G$-bundle $M$. Hence it can be interpreted as the pullback $\widetilde{A}_{\alpha}=\sigma^\ast_\alpha\Xi$ of a global $1$-form connection $\Xi\in\Omega^1(K,\tilde{\mathfrak{g}})$ by a choice of local sections $\sigma_\alpha$.

\begin{remark}[$Q$-flux]
Hence we have both the geometric flux
$F^{(0)\,k}_{\;\;\, ij} := \varepsilon^{\;\;\,k}_{ij}$, given by the bracket $[-,-]_{\mathfrak{g}}$, and locally non-geometric flux
$Q^{(0)ij}_{\quad\;\; k} := \tilde{\varepsilon}^{ij}_{\;\; k}$, given by the bracket $[-,-]^{\tilde{\mathfrak{g}}}$. We will see now that this means that we have a Manin triple of Lie algebras $\mathfrak{d}:=\mathfrak{g}\oplus\tilde{\mathfrak{g}}$ which can be integrated to a Drinfel'd double $D=G\bowtie\widetilde{G}$.
\end{remark}

\begin{remark}[Poisson-Lie T-fold]
Now the generalized correspondence space $K$ will be a principal $\widetilde{G}$-bundle on spacetime $M$ with connection $1$-form $\Xi\in\Omega^1(K,\tilde{\mathfrak{g}})$ given by
\begin{equation}
    \Xi_i \,=\,  \tilde{\theta}_\alpha^{\ast}\tilde{\tau}_i + \mathrm{Ad}_{\tilde{\theta}_\alpha^{-1}}\big( B^{(1)}_{\alpha i} - B^{(0)}_{\alpha ij}\xi^j \big)
\end{equation}
where $\tilde{\theta}_\alpha:U_\alpha\times \widetilde{G}\rightarrow \widetilde{G}$ is a canonical local trivialization of the $\widetilde{G}$-bundle given by $\sigma_\alpha$ and where $\tilde{\tau}$ is the left-invariant Maurer-Cartan form on $\widetilde{G}$. Locally the generalized correspondence space $K$ will be given by patches of the form $U_\alpha\times D$ where $U_\alpha$ is a patch of $M_0$ and the group $D:= G\bowtie\widetilde{G}$ is a Drinfel'd double. In the special case where $\widetilde{A}_\alpha$ is the pull-back of a local $1$-form on the base manifold $M_0$, then we have a proper principal $D$-bundle on $M_0$, otherwise the patching conditions will be more complicated. We have the following picture
\begin{equation}
\begin{tikzcd}[row sep=4.5ex, column sep=5.3ex]
 & & K \arrow[dl, two heads]\arrow[dr, dotted, end anchor={[yshift=-2ex]}] & \\
 G \arrow[r, hook] & M \arrow[dr, "\pi"', two heads] & & [-2ex]\begin{tabular}{c}Poisson-Lie \\T-fold\end{tabular} \arrow[dl, dotted, start anchor={[yshift=2ex]}] \\
 & & M_0 &
\end{tikzcd}
\end{equation}
where the dotted arrows are again not actual maps, but they are only indicative.
\end{remark}

\begin{remark}[Recovering abelian and non-abelian T-duality]
Notice that in the particular case of a Drinfel'd double $D = T^\ast T^{n}= T^{2n}$ we recover exactly abelian T-duality.
Moreover in the particular case of a Drinfel'd double $D = T^\ast G \simeq G\ltimes\mathbb{R}^{\mathrm{dim}\,G}$ with dual fiber $\widetilde{G}:=\mathbb{R}^{\mathrm{dim}\,G}$ we recover exactly the non-abelian T-duality of the previous subsection. 
\end{remark}

\noindent Remarkably, as proven by \cite{MarSza18}, Drinfel'd doubles have a natural para-Hermitian structure. We will now see that this is still true for the generalized correspondence space $K$ of the Lie-Poisson T-duality (clearly including non-abelian T-duality of the previous subsection).

\begin{digression}[The generalized correspondence space is a global para-Hermitian bundle]
As usual, we have the following tensors
\begin{equation}
    \begin{aligned}
    \omega_K &= \Xi_i \wedge \xi^i \\
    \eta_K &= \Xi_i \odot \xi^i 
    \end{aligned}
\end{equation}
and also the para-complex structure
\begin{equation}
    J_K \,=\,  \tilde{\mathbf{k}}^i\otimes\Xi_i -  \big(\mathbf{k}_i + (\eta_{\alpha ij}^{(0)} - B^{(0)}_{\alpha ij}) \tilde{\mathbf{k}}^j\big) \otimes \xi^i 
\end{equation}
where $\{\tilde{\mathbf{k}}^i\}$ are the fundamental vector fields of the $\widetilde{G}$-bundle. This is a para-Hermitian structure $(K,J_K,\eta_K)$ on the correspondence space $K$.
\end{digression}

\noindent A more in-depth analysis of the non-abelian and Poisson-Lie T-duality cases would need to develop a complete technology of non-abelian Higher Kaluza-Klein reduction. 

\section{Application: NS5-brane is Higher Kaluza-Klein monopole}
In this section we will present a new, globally defined monopole for Higher Kaluza-Klein Theory, by directly generalizing the ordinary Kaluza-Klein monopole by \cite{GrosPer83}. This can be interpreted as a globally defined monopole for DFT which does not need compactified dimensions to be well-defined. We will show that this monopole is an NS5-brane with non-trivial $H$-charge by Higher Kaluza-Klein reduction. Finally we will prove that by smearing it we recover the familiar Berman-Rudolph DFT monopole.

\subsection{Higher Dirac monopole of the Kalb-Ramond field}

Let us give a quick review of the Dirac monopole in classical electromagnetism in this subsection. Then we will directly generalize this notion to a Kalb-Ramond field monopole.

\begin{definition}[Dirac monopole]
A Dirac monopole is a circle bundle of the form
\begin{equation}
    \mathbb{R}^{1}\times\left(\mathbb{R}^{3}-\{0\}\right) \longrightarrow \mathbf{B}U(1).
\end{equation}
with non-trivial first Chern class $[F]\in H^2\big(\mathbb{R}^{3}-\{0\},\,\mathbb{Z}\big)$ on the transverse space $\mathbb{R}^{3}-\{0\}$ and trivial on the time line $\mathbb{R}^{1}$. Here $\mathbb{R}^{1}$ can be seen as a magnetically charged world-line.
\end{definition}

\begin{remark}[Dirac charge-quantization]
This spacetime can be alternatively written as
\begin{equation}
    \mathbb{R}^{1}\times\left(\mathbb{R}^{3}-\{0\}\right) \,\simeq\, \mathbb{R}^{1} \times \mathbb{R}^+ \times S^2
\end{equation}
where $\mathbb{R}^+$ is the radial direction in the transversal space and $S^2$ embodies the angular directions. Since $\mathbb{R}^+ \times S^2$ is homotopy equivalent to the $2$-sphere, its cohomology groups will be clearly isomorphic to the ones of the $2$-sphere. The underlying topological space of the stack $\mathbf{B}U(1)$ is the classifying space $BU(1)$ of circle bundles, i.e. the second Eilenberg–MacLane space
\begin{equation}
    \left|\mathbf{B}U(1)\right| = BU(1) = K(\mathbb{Z},2),
\end{equation}
where $|-|$ gives the geometric realization of a stack. Circle bundles over the $2$-sphere are then classified by maps $S^2\rightarrow K(\mathbb{Z},2)$ whose group is just $\pi_2\big(K(\mathbb{Z},2)\big) \cong \mathbb{Z}$. Dually the second cohomology group of the $2$-sphere is $H^2(S^2,\mathbb{Z})\cong \mathbb{Z}$. Hence the first Chern number of any such bundle will be an integer
\begin{equation}
    \frac{1}{2\pi}\int_{S^2} F = m\in\mathbb{Z}
\end{equation}
and the curvature of the bundle will be a closed non-exact form $F=m\,\mathrm{Vol}(S^2)/2$. The trivial fibration $S^2\times S^1\rightarrow S^2$ corresponds to $m=0$ and the Hopf fibration $S^3\rightarrow S^2$ to $m=1$, while in general we will have a Lens space fibration $L(1,m)\rightarrow S^2$ for any $m\in\mathbb{Z}$.
\end{remark}

\begin{remark}[Local description of Dirac monopole]\label{rem:dirac}
Let us quickly look at what this means in terms of gauge fields. As very well known we can cover the $2$-sphere with just two open sets $\mathcal{U}=\{U,U'\}$ such that in spherical coordinates $(\phi,\theta)$ they are overlapping spherical caps $U=[0,2\pi)\times[0,\pi/2+u)$ and $U'=[0,2\pi)\times(\pi/2-u,\pi]$ for some $u\ll\pi/2$. The curvature can be explicitly written as $F=\sin{\theta}\,\mathrm{d}\theta\wedge\mathrm{d}\phi$. On the two charts the connection of the bundle $S^2\rightarrow K(\mathbb{Z},2)$ will then be given respectively by
\begin{equation}
    \begin{aligned}
        A = \frac{m}{2}(1-\cos{\theta})\mathrm{d}\phi, \quad A' = -\frac{m}{2}(1+\cos{\theta})\mathrm{d}\phi.
    \end{aligned}
\end{equation}
We can see that on the overlap $U\cap U'=(\pi/2-u,\pi/2+u)\times S^1$ we have $A-A'=m\,\mathrm{d}\phi$, which, integrated along the equator, gives $4\pi m$ and equivalently the Dirac quantization condition
\begin{equation}
    \begin{aligned}
        \frac{1}{2\pi}\int_{S^1}A-A' = m.
    \end{aligned}
\end{equation}
\end{remark}

\begin{definition}[Higher Dirac monopole]\label{def:hdm}
An Higher Dirac monopole is a gerbe of the form
\begin{equation}
    \mathbb{R}^{1,5} \times \left(\mathbb{R}^4-\{0\}\right) \longrightarrow \mathbf{B}^2U(1),
\end{equation}
with non-trivial Dixmier-Douady class $[H]$ on the transverse space $\mathbb{R}^4-\{0\}$ and trivial over $\mathbb{R}^{1,5}$. Here $\mathbb{R}^{1,5}$ can be seen as a magnetically $H$-charged world-volume.
\end{definition}

\begin{remark}[Higher Dirac charge-quantization]
This spacetime can be also written as
\begin{equation}
    \mathbb{R}^{1,5} \times \left(\mathbb{R}^4-\{0\}\right) \,\simeq\, \mathbb{R}^{1,5} \times \mathbb{R}^+ \times S^3,
\end{equation}
where $\mathbb{R}^+$ is the radial direction in the transversal space, while $S^3$ embodies the angular directions. Since $\mathbb{R}^+ \times S^3$ is homotopy equivalent to the $3$-sphere, the cohomology groups of the transversal space will be immediately isomorphic to the ones of the $3$-sphere. The classifying space of abelian gerbes is the third Eilenberg–MacLane space 
\begin{equation}
    \left|\mathbf{B}^2U(1)\right| = B^2U(1) = K(\mathbb{Z},3)
\end{equation}
and the gerbes on the $3$-spheres are given by maps $S^3\rightarrow K(\mathbb{Z},3)$. The group of these maps is just the third homotopy group of the Eilenberg–MacLane space 
\begin{equation}
    \pi_3\big(K(\mathbb{Z},3)\big) \cong \mathbb{Z}
\end{equation}
which is isomorphic to the integers. Dually the third cohomology group of the $3$-sphere is $H^3(S^3,\mathbb{Z})\cong\mathbb{Z}$. Hence the Dixmier-Douady number of any such bundle will be an integer
\begin{equation}
    \frac{1}{4\pi^2}\int_{S^3} H = m\in\mathbb{Z}
\end{equation}
that we may call \textit{higher magnetic charge} or just $H$-charge. Then the curvature of the gerbe will be in general a non-exact $3$-form
\begin{equation}
    H = \frac{m}{2}\,\mathrm{Vol}(S^3),
\end{equation}
in direct analogy with the ordinary Dirac monopole.
\end{remark}

\begin{remark}[Atlas for the $3$-sphere]
A $3$-sphere $S^3$ can be seen as the submanifold of $\mathbb{C}^2$ defined by the condition $w_1^\ast w_1+w_2^\ast w_2=1$ on the complex coordinates $(w_1,w_2)\in\mathbb{C}^2$. This condition can be solved by
\begin{equation}
    w_1 = e^{i(\psi_1+\psi_2)}\sin{\chi}, \quad  w_2 = e^{i(\psi_1-\psi_2)}\cos{\chi},
\end{equation}
where $(\chi,\psi_1,\psi_2)$ with ranges $\chi\in[0,\pi/2],\psi_1\in[0,2\pi)$ and $\psi_2\in[0,\pi)$, are called \textit{Hopf coordinates}. Topologically $S^3-\{\ast\}\simeq \mathbb{R}^3$, thus we can give $S^3$ an open cover $\mathcal{U}=\{U,U'\}$ of just two open sets, for example the ones of $S^3$ deprived respectively of north and south pole.
\end{remark}

\begin{remark}[Local description of higher Dirac monopole]
We can cover the $3$-sphere with two open spherical caps $\mathcal{U}=\{U,U'\}$ and solve the gerbe curvature $H=m\mathrm{Vol}(S^3)/2$ by
\begin{equation}
    \begin{aligned}
        B = \frac{m}{2}(1-\cos{2\chi})\,\mathrm{d}\psi_1\wedge\mathrm{d}\psi_2, \quad B' = -\frac{m}{2}(1+\cos{2\chi})\,\mathrm{d}\psi_1\wedge\mathrm{d}\psi_2.
    \end{aligned}
\end{equation}
The overlap of the patches $U\cap U'\simeq (u,\pi/2-u)\times T^2$ for some $u\ll\pi/2$ is homotopy equivalent to $T^2$, so that $H^2(T^2,\mathbb{Z})\cong\mathbb{Z}$ and hence we have 
\begin{equation}
    \begin{aligned}
        \frac{1}{4\pi^2}\int_{T^2}B-B' = m,
    \end{aligned}
\end{equation}
which is exactly the charge-quantization condition for String Theory with the constant $\alpha'=1$.
\end{remark}

\subsection{Higher Kaluza-Klein monopole in 10d}

In this subsection we will define the Higher Kaluza-Klein monopole and we will look at its properties. Moreover we will show that it physically reduces to the NS5-brane. First of all let us give a quick review of the ordinary Kaluza-Klein monopole that we are going to generalize

\begin{digression}[Kaluza-Klein monopole by \cite{GrosPer83}]
A Kaluza-Klein monopole is a spacetime $(M,g)$ such that $M=\mathbb{R}^{1}\times\mathbb{R}^+\times L(1,m)$, where $L(1,m)$ is a Lens space, and the metric is
\begin{equation}
    \begin{aligned}
        g &= -\mathrm{d}t^2 + h(r)\delta_{ij}\mathrm{d}y^i\mathrm{d}y^j + \frac{1}{h(r)}(\mathrm{d}\tilde{y}+A_i\mathrm{d}y^i)^2
    \end{aligned}
\end{equation}
where we called $r^2:=\delta_{ij}y^iy^j$ the the radius in the transverse space, $t$ the coordinate on $\mathbb{R}^{1}$, $\{y^i\}_{i=1,2,3}$ the coordinate of the transverse space and $\tilde{y}$ the coordinate of the fiber $S^1$. This spacetime encompass a Dirac monopole on the base manifold $\mathbb{R}^{1}\times\left(\mathbb{R}^3-\{0\}\right) \simeq \mathbb{R}^{1}\times\mathbb{R}^+\times S^2$
 The gauge field $A$ is required to satisfy the following conditions
\begin{equation}
    F = \star_{\mathbb{R}^3}\mathrm{d}h, \quad h(r)=1+\frac{m}{r}
\end{equation}
for $m\in\mathbb{Z}$. In other words the curvature of the bundle will be $F=m\mathrm{Vol}(S^2)/2$ with first Chern number $m$, representing the magnetic charge.
If we cover the $2$-sphere with two charts we can rewrite this metric in polar coordinates $(r,\theta,\phi)$ and we obtain on the first one
\begin{equation}
    g = -\mathrm{d}t^2 + h(r)(\mathrm{d}r^2 +r^2\mathrm{d}\theta^2 + r^2 \sin^2{\theta}\,\mathrm{d}\phi^2) + \frac{1}{h(r)}\bigg(\mathrm{d}\tilde{y}+\frac{m}{2}(1-\cos{\theta})\mathrm{d}\phi\bigg)^2  \\
\end{equation}
while on the second one the gauge field is $A'=-m(1+\cos{\theta})\mathrm{d}\phi/2$, in accord with remark \ref{rem:dirac}.
\end{digression}

\noindent Now we can give a precise definition of a new monopole, which directly generalizes the ordinary Kaluza-Klein monopole and which geometrizes the Higher Dirac monopole of definition \ref{def:hdm}.

\begin{definition}[Higher Kaluza-Klein monopole]\label{def:hkkm}
A Higher Kaluza-Klein monopole is a metric doubled space $(\mathcal{M},\mathcal{H})$ such that 
\begin{itemize}
    \item $\mathcal{M}$ is a doubled space on the base manifold $M=\mathbb{R}^{1,5}\times\left(\mathbb{R}^4-\{0\}\right) \simeq \mathbb{R}^{1,5}\times\mathbb{R}^+\times S^3$ with global connection $\omega_B = \mathrm{d}\tilde{y}_\alpha-B_\alpha$ and which is trivial on $\mathbb{R}^{1,5}\times\mathbb{R}^+$, where $\{x^\mu\}$ and $\{y^i\}$ are respectively coordinates for $\mathbb{R}^{1,5}$ and $\mathbb{R}^{+}\times S^3$,
    \item $\mathcal{H}$ is a (global) doubled metric given by
\end{itemize}
\begin{equation}\label{eq:hkkmonopole}
    \begin{aligned}
        \mathcal{H} &= \eta_{\mu\nu}\mathrm{d}x^\mu\mathrm{d}x^\nu + \eta^{\mu\nu}\mathrm{d}\tilde{x}_\mu\mathrm{d}\tilde{x}_\nu + h(r)\delta_{ij}\mathrm{d}y^i\mathrm{d}y^j + \frac{\delta^{ij}}{h(r)}(\mathrm{d}\tilde{y}_i+B_{ik}\mathrm{d}y^k)(\mathrm{d}\tilde{y}_j+B_{jk}\mathrm{d}y^k)
    \end{aligned}
\end{equation}
where the curvature of the gerbe and the harmonic function are respectively constraint to
\begin{equation}\label{eq:transversalcond}
    H = \star_{\mathbb{R}^4}\mathrm{d}h, \quad h(r)=1+\frac{m}{r^2}
\end{equation}
for any $m\in\mathbb{Z}$ and where $r^2:=\delta_{ij}y^iy^j$ is the radius in the four dimensional transverse space. 
\end{definition}

\noindent This doubled metric encompasses a Higher Dirac monopole from definition \ref{def:hdm} on the base manifold, just as the Kaluza-Klein monopole does with an ordinary Dirac monopole. In other words the curvature of the gerbe will be $H=m\mathrm{Vol}(S^3)/2$ with non-trivial $H$-charge $m$.

\vspace{0.2cm}
\noindent Let us spend a couple of words to remark that the Higher Kaluza-Klein monopole from definition \ref{def:hkkm} in is globally well-defined. Given an open good cover $\{U_\alpha\}$ of the transverse space $\mathbb{R}^4-\{0\}$ we can write $\tilde{y}_\alpha - \tilde{y}_\beta = - \Lambda_{\alpha\beta}+\mathrm{d}\phi_{\alpha\beta}$ on each $U_\alpha\cap U_\beta$. Thus $\mathrm{d}\tilde{y}_\alpha - \mathrm{d}\tilde{y}_\beta = \mathrm{d}\Lambda_{\alpha\beta}$, which combined with $B_\beta-B_\alpha = \mathrm{d}\Lambda_{\alpha\beta}$ assures that $\omega_B=\mathrm{d}\tilde{y}_\alpha-B_\alpha$ is globally defined, by lemma \ref{thm:gerbeconn}.

\begin{remark}[Doubled space of Higher Kaluza-Klein monopole]
By using the fact that $S^3$ is a circle bundle (Hopf fibration) on $S^2$ we can apply the the result of lemma \ref{thm:corrspace} and write
\begin{equation}
\mathcal{M} \,\simeq\, \mathcal{M}_{0} \times_{S^2} \big(S^3\times_{S^2}L(1,m)\big)
\end{equation}
where $\mathcal{M}_0$ is a doubled space on $M_0:=\mathbb{R}^{1,5}\times\mathbb{R}^+\times S^2$. Since $\mathcal{M}_0$ is still trivial on $\mathbb{R}^{1,5}\times\mathbb{R}^+$ we can rewrite it as $\mathcal{M}_0 \simeq \mathcal{M}_{\mathrm{triv}}\times \mathcal{M}_{S^2}$, where $\mathcal{M}_{S^2}$ is the non-trivial doubled space on $S^2$ with non-trivial connection $B_\alpha^{(2)}$.
The correspondence space is the $T^2$-bundle $K=S^3\times_{S^2}L(1,m)$.
\end{remark}

\begin{remark}[NS5-brane is Higher Kaluza-Klein monopole]
By Higher Kaluza-Klein reduction of \eqref{eq:hkkmonopole} to $M=\mathbb{R}^{1,5}\times\mathbb{R}^+\times S^3$ we get the following metric and gerbe connection:
\begin{equation}
    \begin{aligned}
        g &= \eta_{\mu\nu}\mathrm{d}x^\mu\mathrm{d}x^\nu + h(r)\delta_{ij}\mathrm{d}y^i\mathrm{d}y^j,\\
        B &= B_{ij}\,\mathrm{d}y^i\wedge\mathrm{d}y^j
    \end{aligned}
\end{equation}
which satisfy the conditions \eqref{eq:transversalcond} on the transversal space.
If we rewrite this metric and $B$-field on a chart in spherical coordinates $(r,\chi,\psi_1,\psi_2)$ we obtain
\begin{equation}
    \begin{aligned}
         g &=   \eta_{\mu\nu}\mathrm{d}x^\mu\mathrm{d}x^\nu + h(r)\mathrm{d}r^2 + h(r)r^2\left(\mathrm{d}\chi^2+\mathrm{d}\psi_1^2+\mathrm{d}\psi_2^2-2\cos{2\chi}\,\mathrm{d}\psi_1\mathrm{d}\psi_2\right) \\
         B &= -\frac{m}{2}\cos{2\chi}\,\mathrm{d}\psi_1\wedge\mathrm{d}\psi_2
    \end{aligned}
\end{equation}
These are exactly the metric and Kalb-Ramond field of an NS5-brane with non-trivial $H$-charge $m$ in $10d$ spacetime $M$. Hence in our Higher Kaluza-Klein framework encompassing NS5-branes is as natural as considering the direct higher version of a Kaluza-Klein monopole.
\end{remark}

\begin{remark}[Geometric interpretation of the NS5-brane]
We know that the Kaluza-Klein brane appears when spacetime $P\rightarrow M$ is a non-trivial circle bundle with some first Chern class $[F]\in H^2(M,\mathbb{Z})$. Perfectly analogously the NS5-brane appears when the doubled space $\mathcal{M}\rightarrow M$ is a non-trivial circle $2$-bundle with Dixmier-Douady class $[H]\in H^3(M,\mathbb{Z})$.
\end{remark}

\begin{remark}[Angular T-dual of the NS5-brane]
The $3$-sphere is nothing but the Lens space $L(1,1)=S^3$ corresponding to the Hopf fibration. As we have seen the transverse space of the NS5 brane with $H$-charge $m$ is $\mathbb{R}^+\times S^3$. Let us perform a T-duality along the $\psi_1$ circle fiber
\begin{equation*}
    \begin{aligned}
         \tilde{g} &=  \eta_{\mu\nu}\mathrm{d}x^\mu\mathrm{d}x^\nu + h(r)\mathrm{d}r^2 + h(r)r^2\left(\mathrm{d}\chi^2+\sin^2{2\chi}\,\mathrm{d}\psi_2^2\right) +\frac{1}{h(r)r^2}\bigg(\frac{\mathrm{d}\tilde{\psi}_1}{2}-\frac{m}{2}\cos{2\chi}\,\mathrm{d}\psi_2\bigg)^2 \\
         \tilde{B} &= -\frac{1}{2}\cos{2\chi}\,\mathrm{d}\tilde{\psi}_1\wedge\mathrm{d}\psi_2
    \end{aligned}
\end{equation*}
This is again a supergravity solution, but it is not asymptotically flat: the $\tilde{\psi}_1$ circle is not Hopf-fibered over the $2$-sphere, but it is be fibered with first Chern number $m$, generally making the whole bundle a Lens space $L(1,m)$. In other words the $H$-charge of the NS5-brane is mapped to a NUT charge $m$ under T-duality. However this background is not actually Taub-NUT, because the harmonic function is $h(r)=1+1/r^2$. 
\end{remark}

\begin{digression}[Recovering angular T-dualities]
The previous is the Plauschinn-Camell solution appearing in \cite{Pla18}. The authors perform angular T-dualities of NS5-brane backgrounds and speculate about implementing them in DFT. In our formulation this is totally natural as long as the angular directions are isometries of the doubled metric $\mathcal{H}$.
\end{digression}

\begin{remark}[Array of higher Kaluza-Klein monopoles]
Since Higher Kaluza-Klein monopoles do not interact, we can construct a multi-monopole solution. This will be a metric doubled space $(\mathcal{M},\mathcal{H})$ such that $\mathcal{M}$ is a doubled space on the base manifold $\mathbb{R}^{1,5}\times\left(\mathbb{R}^4-\{y_p\}\right)$ with global connection $\mathrm{d}\tilde{y}_\alpha-B_\alpha$ and  $\mathcal{H}$ is the (global) doubled metric given by
\begin{equation}
    \begin{aligned}
        \mathcal{H} &= \eta_{\mu\nu}\mathrm{d}x^\mu\mathrm{d}x^\nu + \eta^{\mu\nu}\mathrm{d}\tilde{x}_\mu\mathrm{d}\tilde{x}_\nu + h(r)\delta_{ij}\mathrm{d}y^i\mathrm{d}y^j + \frac{\delta^{ij}}{h(r)}(\mathrm{d}\tilde{y}_i+B_{ik}\mathrm{d}y^k)(\mathrm{d}\tilde{y}_j+B_{jk}\mathrm{d}y^k)
    \end{aligned}
\end{equation}
which satisfies the conditions
\begin{equation}
    H = \star_{\mathbb{R}^4}\mathrm{d}h, \quad h(y)=1+\sum_p\frac{m_p}{|y-y_p|^2}
\end{equation}
where $y_p$ are the positions of the monopoles in the transverse space and $m_p$ are their $H$-charges.
\end{remark}

\subsection{Berman-Rudolph DFT monopole in 9d}
In this subsection we will give a global definition of the usual DFT monopole in our formalism. See \cite{BR14} for its original definition and \cite{Jen11} for seminal work. Then we will show that it is immediately related to our Higher Kaluza-Klein monopole.

\begin{definition}[DFT monopole]\label{def:dftm}
The Berman-Rudolph DFT monopole \cite{BR14} is a metric doubled space $(\mathcal{M},\mathcal{H})$ such that 
\begin{itemize}
    \item $\mathcal{M}$ is a doubled space on the base $M=\mathbb{R}^{1,5}\times\left(\mathbb{R}^3-\{0\}\right)\times S^1 \,\simeq\, \mathbb{R}^{1,5}\times\mathbb{R}^+\times S^2\times S^1$ which is trivial on $\mathbb{R}^{1,5}\times\mathbb{R}^+\times S^1$ and which has global connection $\omega_B = (\mathrm{d}\tilde{z}_\alpha+A_\alpha)\wedge\mathrm{d}z$, where $\{x^\mu\}$, $\{y^i\}$ and $\{z\}$ are respectively coordinates for $\mathbb{R}^{1,5}$, $\mathbb{R}^{+}\times S^2$ and $S^1$,
    \item $\mathcal{H}$ is the (global) doubled metric given by
\end{itemize}
\begin{equation}\label{eq:smearedhkkmonopole}
    \begin{aligned}
        \mathcal{H} \,&=\, \eta_{\mu\nu}\mathrm{d}x^\mu\mathrm{d}x^\nu + \eta^{\mu\nu}\mathrm{d}\tilde{x}_\mu\mathrm{d}\tilde{x}_\nu + h(r)\delta_{ij}\mathrm{d}y^i\mathrm{d}y^j + \frac{\delta^{ij}}{h(r)}\mathrm{d}\tilde{y}_i\mathrm{d}\tilde{y}_j\\
        &\hspace{4cm} +h(r)\mathrm{d}z^{2} + \frac{1}{h(r)}(\mathrm{d}\tilde{z}+{A}_{k}\mathrm{d}y^k)^{2}
    \end{aligned}
\end{equation}
where the curvature of the connection and the harmonic function are respectively constraint to
\begin{equation}\label{eq:transversalcond2}
    \mathrm{d}A = \star_{\mathbb{R}^3}\mathrm{d}h, \quad h(r)=1+\frac{m'}{r}
\end{equation}
\end{definition}

\begin{remark}[Doubled space of the DFT monopole]
The global doubled space of the DFT monopole of definition \ref{def:dftm} is similar, but simpler respect to the one of a general Higher Kaluza-Klein monopole. The circle coordinate $z\in\mathbb{R}/\mathbb{Z}$ is trivially fibered over the $2$-sphere, while its dual $\tilde{z}_{\alpha}$ is in general non-trivially fibered by a connection $A_{\alpha}$ over it. Hence, by lemma \ref{thm:corrspace},
\begin{equation}
    \mathcal{M} \,\simeq\, \mathcal{M}_0 \times_{S^2} L(1,m') \times S^1
\end{equation}
where $\mathcal{M}_0$ is a doubled space on the base $M_0 := \mathbb{R}^{1,5}\times\mathbb{R}^+\times S^2$. In this case $\mathcal{M}_0$ is trivial since the original gerbe connection $\omega_B$ has no $\omega_B^{(2)}$ component. On the other hand the component $\omega_B^{(1)} = \mathrm{d}\tilde{z}_\alpha+A_\alpha$ is exactly the connection of $L(1,m')$ on $S^2$.
\end{remark}

\begin{theorem}[Recovering the DFT monopole]
A Berman-Rudolph DFT monopole (definition \ref{def:dftm}) is a smeared Higher Kaluza-Klein monopole (definition \ref{def:hkkm}).
\end{theorem}
\begin{proof}
The transverse space is a trivial circle fibration $(\mathbb{R}^3-\{0\})\times S^1$ with trivial connection $\mathrm{d}z$. As usual can decompose $B_\alpha = B_\alpha^{(2)} + B_\alpha^{(1)} \wedge \mathrm{d}z$, where $B_\alpha^{(2)}$ and $B_\alpha^{(1)}$ are respectively a gerbe connection and a circle connection on $\mathbb{R}^3-\{0\}$. Now we can gauge away the component $B_\alpha^{(2)}$ on $\mathbb{R}^3-\{0\}$ since this space is homotopy equivalent to a $2$-sphere and $H^3(S^2,\mathbb{Z})=0$ implies that any gerbe over $S^2$ is trivial. Thus the only non-trivial contribution to the $H$-flux will come from the connection ${A}_\alpha:=B_\alpha^{(1)}$ of the dual circle bundle and the gerbe curvature is $\mathrm{d}B_\alpha = \mathrm{d}A_\alpha\wedge\mathrm{d}z$.
\begin{equation}
    h(r,z)=1+\sum_{p\in\mathbb{Z}}\frac{m}{r^2 + (z-2\pi p)^2} \;\,\xrightarrow{\;\;\; r\gg 1 \;\;\;}\;\, 1 + \frac{m'}{r}
\end{equation}
with modified charge $m':=m/2$. Moreover we have $\star_{\mathbb{R}^4}\mathrm{d}h = (\star_{\mathbb{R}^3}\mathrm{d}h)\wedge\mathrm{d}z$ and thus the condition $\mathrm{d}B_\alpha=\star_{\mathbb{R}^4}\mathrm{d}h$ becomes the equation \eqref{eq:transversalcond2}.
\end{proof}

\begin{remark}[DFT monopole is smeared NS5-brane]
By Higher Kaluza-Klein reduction of \eqref{eq:smearedhkkmonopole} to $M=\mathbb{R}^{1,5}\times \mathbb{R}^+\times S^2\times S^1$ we get the metric and gerbe connection
\begin{equation}
    \begin{aligned}
        g &= \eta_{\mu\nu}\mathrm{d}x^\mu\mathrm{d}x^\nu + h(r)\big(\delta_{ij}\mathrm{d}y^i\mathrm{d}y^j+\mathrm{d}z^{2}\big), \\
        B &= {A}_{k}\mathrm{d}y^k\wedge\mathrm{d}z.
    \end{aligned}
\end{equation}
This solution is unsurprisingly the smeared NS5-brane background. Notice the asymptotic geometry is the trivial Lens space $L(1,0)=S^2\times S^1$.
\end{remark}

\begin{remark}[KK5-brane is the T-dual of NS5-brane]
By Higher Kaluza-Klein reduction of \eqref{eq:smearedhkkmonopole} to the dual bundle $\widetilde{M}=\mathbb{R}^{1,5}\times\mathbb{R}^+\times L(1,m)$ we get the metric and gerbe connection
\begin{equation}
    \begin{aligned}
        \tilde{g} &=  \eta_{\mu\nu}\mathrm{d}x^\mu\mathrm{d}x^\nu + h(r)\delta_{ij}\mathrm{d}y^i\mathrm{d}y^j + \frac{1}{h(r)}(\mathrm{d}\tilde{z}+{A}_{k}\mathrm{d}y^k)^{2} \\
        \tilde{B} &= 0.
    \end{aligned}
\end{equation}
The transverse space is a Taub-NUT space with asymptotic geometry $L(1,m)$ and it has zero $H$-charge. This solution is exactly the KK5-brane with isometry along the $\tilde{z}_\alpha$ circle.
\end{remark}

\begin{digression}[Localization of DFT monopole in the previous literature]
In \cite{BR14} we have a (local) definition of the DFT monopole and then a bottom-up generalization to its non-smeared version on $\mathbb{R}^{1,5}\times(\mathbb{R}^3-\{0\})\times S^1$. The winding mode corrections of this process are studied in \cite{Kim13}. The resulting generalized metric is a (locally defined) version of Higher Kaluza-Klein monopole on this particular background. In \cite{Ber19} it is argued that, in the case of an torus compactified spacetime we can write Higher Dirac monopole in terms of an ordinary Dirac monopole by $H^3(S^2\times S^1,\mathbb{Z}) \cong H^2(S^2,\mathbb{Z})\otimes_\mathbb{Z} H^1(S^1,\mathbb{Z})$. But also that a full DFT monopole should require a geometrization of the gerbe which is impossible to achieve with just manifolds: Higher Kaluza-Klein geometry is hopefully an answer to this.
\end{digression}

\subsection{Berman-Rudolph DFT monopole in 8d}

In this subsection we will take a quick look to a further dimensional reduction of our monopole.

\begin{remark}[Reduction to $8d$ and $5_2^2$-brane]
If we compactify again spacetime to a trivial torus bundle $M=\mathbb{R}^{1,5}\times(\mathbb{R}^2-\{0\})\times T^2$, we can further smear and Kaluza-Klein reduce our Higher Kaluza-Klein monopole to recover the zoo of exotic branes by \cite{Bak16}. By explicitly writing the T-dualities along the two directions we have the following commutative diagram
\begin{equation}
    \begin{tikzcd}[row sep=21ex, column sep=21ex]
    \mathrm{NS5}_{12} \arrow[r, leftrightarrow, "\mathcal{T}_1"] \arrow[d, leftrightarrow, "\mathcal{T}_2"] & \mathrm{KK5}^{\;2}_{1} \arrow[d, leftrightarrow, "\mathcal{T}_2"'] \\
    \mathrm{NS5}^1_{\;2} \arrow[r, leftrightarrow, "\mathcal{T}_1"] & (5^2_2)^{12} 
\end{tikzcd}
\end{equation}
where $\mathrm{NS5}_{12}$ is the NS5-brane smeared along both the directions of the $T^2$ fiber, while $\mathrm{NS5}^a_{\;b}$ is the KK-brane with isometry along the $a$-th direction and smeared along the $b$-th direction, while $(5^2_2)^{12}$ is the $5^2_2$-brane with isometry along both the directions. 
\end{remark}

\noindent Let $\{z^a\}$ be coordinates of $T^2$ with $a=1,2$ and $\{y^i\}$ be coordinates of $\mathbb{R}^2-\{0\}$ with $i=3,4$. The generalized correspondence space $K$ of the transverse space will be a non-trivial $T^2$-bundle on $(\mathbb{R}^2-\{0\})\times T^2$ with curvature $\widetilde{F}_a = H^{(1)}_{ab}\wedge\mathrm{d}z^b$. 

\begin{remark}[Geometric interpretation of $5_2^2$-brane]
The NS5-brane is associated to a non-trivial Dixmier-Douady class $[H]\in H^3(\mathbb{R}^4-\{0\},\mathbb{Z})$ on the transverse space. If spacetime is a trivial $T^2$-fibration and the NS5-brane is smeared along these directions, there will be a non-trivial flux compactification $[H^{(1)}_{12}]\in H^1(\mathbb{R}^2-\{0\},\mathbb{Z})\cong\mathbb{Z}$ on the reduced transverse space. A T-duality along the $1$st direction gives a $[F^{(1)1}_{\quad \;\;2}]\in H^1(\mathbb{R}^2-\{0\},\mathbb{Z})$. Indeed notice that $H^{>1}(\mathbb{R}^2-\{0\},\mathbb{Z})=0$ so there are no non-trivial abelian gauge field on the transverse space. Now by performing a T-duality in the $2$nd direction we get a flux $[Q^{(1)12}]\in H^1(\mathbb{R}^2-\{0\},\mathbb{Z})$.
Just like the Kaluza-Klein brane generally appears when spacetime is a circle fibration with non-trivial first Chern class, a $5_2^2$ brane appears when spacetime is $T^2$-fibration with nontrivial $Q$-flux given by the cohomology class $[Q^{(1)12}]\in H^1(\mathbb{R}^2-\{0\},\mathbb{Z})\cong\mathbb{Z}$. Notice that this class corresponds to an integer number which we can name $Q$-charge of the $5^2_2$-brane.
\end{remark}

\begin{remark}[Reduction to $8d$ and $5_2^3$-brane]
Now recall that our spacetime manifold is $M=\mathbb{R}^{1,5}\times(\mathbb{R}^2-\{0\})\times T^2$. Now translations along the two circles of the torus are legit isometries, while translations along any fixed direction in $\langle y^3\rangle\subset\mathbb{R}^2-\{0\}$ are not. As explained by \cite{Bak16}, we can still perform a local $\mathcal{T}_3\in O(10,10)$ transformation along this direction on the patches $T^\ast U_\alpha$ of the doubled space $\mathcal{M}$, even if it will not be a global T-duality. Hence we will recover the usual picture of exotic branes by the diagram
\begin{equation}
   \begin{tikzcd}[row sep=38,column sep=28]
    & \mathrm{NS5}_{123}  \arrow[rr, "\mathcal{T}_1"] \arrow[dd, "\mathcal{T}_2" near start] \arrow[dl, "\mathcal{T}_3"'] & &   \mathrm{KK5}^{1}_{\;23}  \arrow[dd, "\mathcal{T}_2"] \arrow[dl, "\mathcal{T}_3"] \\
    \mathrm{KK5}^{\;\;\,3}_{12} \ar[crossing over, "\mathcal{T}_1" near start]{rr} \arrow[dd, "\mathcal{T}_2"'] & & (5^2_2)^{1\; 3}_{\; 2} \\
      &  \mathrm{KK5}^{\;2}_{1\; 3}  \arrow[rr, "\mathcal{T}_1" near start] \arrow[dl, "\mathcal{T}_3"'] & &  (5^2_2)^{12}_{\;\;\, 3}  \arrow[dl, "\mathcal{T}_3"] \\
    (5^2_2)^{\;23}_{1} \arrow[rr, "\mathcal{T}_1"'] && (5_2^3)^{123} \arrow[from=uu,crossing over, "\mathcal{T}_2" near start]
 \end{tikzcd} 
\end{equation}
where the superscript $3$ and the subscript $3$ this time do not mean isometry/smearing like for $1$ and $2$, but they respectively mean dependence either on the coordinate $y^3$ or its dual $\tilde{y}_3$.
\end{remark}

\begin{remark}[Geometric interpretation of $5_2^3$-brane]
A geometric interpretation can be adopted also for the hypothetical $5_2^3$-brane, which will correspond to a $T^2$-compactification carrying a non-trivial $R$-flux class $[R\,]\in H^1(\mathbb{R}^2-\{0\},\mathbb{Z})\cong\mathbb{Z}$ which we may call $R$-charge.
\end{remark}

\section{Outlook}

In this section we will sum up the results of our construction and we will take a quick look to the future natural directions for this proposal.

\subsection{Discussion}
From a small set of assumptions we developed a global formulation for DFT geometry by directly generalizing Kaluza-Klein proposal to higher gauge fields. From our formalism we recovered many previous relevant proposals of a finite DFT geometry and we clarified how to recover them in a global picture. In particular we got a globally defined version of Park's coordinate gauge from \cite{Park13}, a formalization for Papadopoulos' C-spaces from \cite{Pap14} and, on local patches, para-K\"{a}hler geometry by \cite{Vais12}. On torus bundle base manifolds we recovered the original Hull's doubled torus bundles from \cite{Hull06}. Para-Hermitian geometry by \cite{Svo18} can be recovered locally on the total doubled space, but globally on the correspondence space. The local symmetries of the theory give us sections of the higher algebroids, which can be described by differential-graded manifolds, meaning that Extended Riemannian Geometry by \cite{DesSae18} can be seen as an infinitesimal (but still globally defined) counterpart of Higher Kaluza-Klein theory.

\begin{table}[ht!]\begin{center}
\begin{tabular}{m{1em}|c c|}
    & $\quad$ Infinitesimal $\quad$ & $\qquad\quad$ Finite $\qquad\quad$ \\ 
    \hline &&\\
    \rotatebox{90}{Local}& \makecell{original\\DFT\\} & \makecell{para-Hermitian\\geometry} \\
    && \\
    \rotatebox{90}{Global}& \makecell{Extended\\Riemannian\\geometry} & \makecell{Higher\\Kaluza-Klein\\geometry} \\
\end{tabular}\end{center}\caption{A brief moral comparison of DFT geometries.}\end{table}

\vspace{0.25cm}
\noindent Notice also that Higher Kaluza-Klein geometry has been proved to be strictly linked with other field of research such as Higher Prequantum Geometry, which provides an alternative interpretation of the extra coordinates as a generalization of the usual phase of wavefunctions for strings, and such as non-associative physics.

\vspace{0.25cm}
\noindent In this paper we did only deal with geometry of DFT and not with its field equations. However there are encouraging hints, such as the discussion in \cite{Par18x} and \cite{Park19stringy}, that the field equations of DFT must be regarded as Einstein-like equations for the doubled metric on the doubled space. This seems very coherent with the spirit and fundamental ideas of our Higher Kaluza-Klein proposal.

\subsection{Towards a geometrized M-Theory?}
Higher Kaluza-Klein does not only provides a theoretical explanation for the existing geometric features of DFT (such as para-Hermitian geometry), but also it can overcome their difficulty in being directly generalized to M-theory. Indeed the Higher Kaluza-Klein Theory we presented in this paper has a large number of immediate natural generalizations:
\begin{itemize}
    \item \textbf{The $n$-bundle is with $n>2$:} this could allow us to formulate global ExFT.
    \item \textbf{The $n$-bundle is non-abelian:} this could not only allow us to formulate global Heterotic DFT, but also to go slightly beyond ExFT to embody the global spin-twisted structures by \cite{Sat09} and thus to geometrize the complicated interplay of gravity and gerbe. 
    \item \textbf{The base manifold is a super-manifold:} this could immediately allow us to generalize everything we mentioned to their global super-space formulation.
\end{itemize}
We are intrigued by the possibility that a \textbf{super non-abelian higher Kaluza-Klein Theory} on the total space of the (twisted) M2-M5-brane gerbe over the $11d$ super-spacetime of \cite{FSS18} can be something closer to a geometrized M-theory than what previously allowed. In the next paragraph we will give a little preview of these generalizations.

\subsection{A quick preview of global Heterotic Double Field Theory}
A non-abelian $2$-group leads immediately to Heterotic Double Field Theory by \cite{Hohm11}. The \textit{heterotic doubled space} $\mathcal{M}_{\mathrm{het}}\xrightarrow{\,\bbpi\,}M$ will be the total space of a principal $\mathrm{String}(d)_{\mathrm{conn}}$-bundle on spacetime $M$, where the \textit{string }$2$\textit{-group} $\mathrm{String}(d)$ is defined by the commuting diagram
\begin{equation}\begin{tikzcd}[row sep=6ex, column sep=5ex]
\mathbf{B}\mathrm{String}(d) \arrow[d]\arrow[r] &\ast \arrow[d]  \\
\mathbf{B}\mathrm{Spin}(d)\arrow[r, "\frac{1}{2}\mathbf{p}_1"] &\mathbf{B}^3U(1)
\end{tikzcd}\end{equation}
and where we called $\mathbf{B}(\mathrm{String}(d)_{\mathrm{conn}})$ the moduli-stack of $\mathrm{String}(d)$-bundles with connection data up to two-fold overlaps. This means that the doubled space $\mathcal{M}_{\mathrm{het}}$ is locally given by a collection of $T^\ast U_\alpha \times \mathrm{Spin}(d)$. These are glued such that the subbundles $U_\alpha\times \mathrm{Spin}(d)$ give a global $\mathrm{Spin}(d)$-bundle $P_{\mathrm{Spin}}\rightarrow M$, while the $T^\ast U_\alpha$ are patched with gerbe $2$-gauge transformations twisted by non-abelian $\mathrm{Spin}(d)$-gauge transformations.Consequently the automorphisms of this heterotic doubled space will extend not only the group of diffeomorphisms of the base $M$, but also the group of automorphisms of the $\mathrm{Spin}(d)$-bundle $P_{\mathrm{Spin}}\rightarrow M$ (see example \ref{ex:nonabelianaut}) by
\begin{equation}
    \Aut_/(f,\Lambda,G) \,=\, \Big(\Diff(M)\ltimes\Gamma\big(M,\,\mathrm{Ad}(P_{\mathrm{Spin}})\big)\Big)\ltimes\mathbf{H}(M,\,\BU)
\end{equation}
From the definition we will have that heterotic doubled spaces are topologically classified by $P_{\mathrm{Spin}}$-twisted cohomology classes. Notice that we will have a \textit{non-abelian global strong constraint}
\begin{equation}
    \M_{\mathrm{het}} /\!/\!_\rho \,\mathrm{String}(d)_{\mathrm{conn}} \,\cong\, M,
\end{equation}
which remarkably combines the abelian strong constraint with the cylindricity condition on the $\mathrm{Spin}(d)$-bundle.
Thus doubled metric will depend only on physical coordinates of spacetime $M$.

\subsection{A quick preview of global Exceptional Field Theory}
It sounds reasonable to use our Higher Kaluza-Klein geometry to formalize ExFT (see \cite{HohSam13}).
Let us consider the M2-M5-brane gerbe from \cite{FSS15x} over the $11d$ spacetime manifold $M$. This will be encoded by a particular \v{C}ech cocycle, which includes a collection of $2$-forms $\Lambda_{\alpha\beta}^{\mathrm{M2}}$ and $5$-forms $\Lambda_{\alpha\beta}^{\mathrm{M5}}$ on two-fold overlaps of patches $U_\alpha\cap U_\beta$. This cocycle will be the transition functions for the \textit{extended space} $\mathcal{M}_{\mathrm{ex}}\xrightarrow{\;\bbpi\;}M$ in analogy with doubled space from postulate \ref{post1}. Noticed these extended spaces will be classified by twisted cohomology classes.

\vspace{0.2cm}
\noindent The $5$-group $\Aut_{/}\big(\text{M2-M5}\big)$ of automorphisms of the extended space $\mathcal{M}_{\mathrm{ex}}\xrightarrow{\;\bbpi\;}M$ will be therefore defined according to higher geometry by the following short exact sequences 
\begin{equation}\begin{aligned}
    1\longrightarrow\mathbf{H}\big(M,\mathbf{B}^2U(1)_{\mathrm{conn}}\big)\longrightarrow \;\;\;\Aut_{/}\big(\text{M2}\big)\;\;\; \longrightarrow \;\;\mathrm{Diff}(M)\;\;\longrightarrow 1 \\
    1\longrightarrow\mathbf{H}\big(M,\mathbf{B}^5U(1)_{\mathrm{conn}}\big)\longrightarrow \Aut_{/}\big(\text{M2-M5}\big) \longrightarrow \Aut_{/}(\text{M2})\longrightarrow 1 
\end{aligned}
\end{equation}
which can be immediately recognized as the finite version of the short exact sequences defining the algebroid $E_{\text{M2-M5}}\rightarrow M$ appearing in exceptional generalized geometry (see \cite{Wald08E})
\begin{equation}\begin{aligned}
    0\longrightarrow\wedge^2T^\ast M\longrightarrow \;\;\,E_{\text{M2}}\;\; \longrightarrow \, TM\longrightarrow 0 \\
    0\longrightarrow\wedge^5T^\ast M\longrightarrow E_{\text{M2-M5}} \longrightarrow E_{\text{M2}}\longrightarrow 0  
\end{aligned}
\end{equation}

\begin{example}[ExFT in $5d$]
Now assume that spacetime $M$ is itself a torus $T^6$-bundle over a $5d$ base manifold $M_0$ with transition functions $f_{\alpha\beta}^i$ and that the gerbe is equivariant under its torus action. Analogously to the DFT case the $2$-forms will dimensionally reduce to a collection of $2,1,0$-forms $\Lambda_{\alpha\beta}^{\mathrm{M2}(2)}$, $\Lambda_{\alpha\beta i}^{\mathrm{M2}(1)}$ and $\Lambda_{\alpha\beta ij}^{\mathrm{M2}(0)}$ on the base $M_0$ describing the winding modes of the M2-brane on the base. Similarly the $5$-forms on $M$ will split in a collection of $5,4,3,2,1,0$-forms on $M_0$ describing the winded M5-brane on the base. Let us simplify the assumptions by requiring that the $5$-gerbe is not twisted by the $2$-gerbe. Hence transition functions $f_{\alpha\beta}^i$ will describe a $T^{6}$-bundle, while $\Lambda_{\alpha\beta ij}^{(0)\mathrm{M2}}$ a $T^{15}$-bundle and $\Lambda_{\alpha\beta ijkln}^{(0)\mathrm{M5}}$ a $T^6$-bundle on the base $M_0$. In total this is a $T^{27}$-bundle over the $5d$ manifold $M_0$ and can be interpreted as the \textit{extended manifold} of ExFT in $5d$ (see \cite{HohSam14}) where the $T^{27}$ fiber is the so-called \textit{internal space}. We can then equip it with a $E_{6(6)}(\mathbb{Z})$-action which mixes, topologically, the first Chern classes of the $T^{27}$-bundle and, differentially, the components of the moduli fields $g_{ij}^{(0)}$, $C^{(0)}_{ijk}$ and $\star C^{(0)}_{ijklnm}$. Hence the extended manifold (and generally a U-fold) will be possible to be understood through Higher Kaluza-Klein Theory. Moreover the $1,2,3,4,5,6$-form components of the dimensional reduction of the connection of the gerbe and its $n$-gauge transformations can describe ExFT tensor hierarchy. See \cite{Hohm19DFT} for remarkable hints in this direction, in infinitesimal fashion. If we drop the simplifying unphysical assumption that the $5$-gerbe is not twisted, the extended manifold will be a more complicated globalization of $U_\alpha\times T^{27}$.
\end{example}

\begin{example}[ExFT in $7d$]
Let us make another simpler example. If we choose $M$ to be a $T^4$-bundle over a $7d$ base manifold $M_0$ with transition functions $f_{\alpha\beta}^i$, the local data $\Lambda_{\alpha\beta}^{\mathrm{M2}}$ will include the transition functions $\Lambda_{\alpha\beta ij}^{\mathrm{M2}(0)}$ of a $T^6$-bundle over $M_0$. On the other hand the local data $\Lambda_{\alpha\beta}^{\mathrm{M5}}$ will not. Hence, coherently with ExFT in $7d$ (see \cite{BlaMal15}), the transition functions $f_{\alpha\beta}^i$ and $\Lambda_{\alpha\beta ij}^{\mathrm{M2}(0)}$ will define a just a $T^{10}$-bundle over $M_0$, which can be equipped with a $SL(5,\mathbb{Z})$-action mixing the components of the moduli fields $g_{ij}^{(0)},C^{(0)}_{ijk}$.
\end{example}

\vspace{0.15cm}
\noindent We will fully apply Higher Kaluza-Klein framework to HetDFT and ExFT in papers to appear.


\section*{Acknowledgement}
I would like to thank my supervisor prof.$\,\,$David Berman for essential discussion about String Theory, M-theory and duality-covariant theories.
I would also like to thank the organizers and participants of the Durham Symposium 2018 \textit{Higher Structures in M-Theory} and of the workshop \textit{Geometry and Duality} 2019 in Potsdam.
Thanks to David Svoboda for helpful discussion about para-Hermitian geometry.
I am grateful to Queen Mary University of London (QMUL) for its partial support.

\medskip
\addcontentsline{toc}{section}{References}
\bibliographystyle{fredrickson}
\bibliography{sample}

\end{document}